%% file: GD-FixedPoint-Deterministic.tex
\documentclass[11pt,letterpaper]{article}

\title{Settling the complexity of Nash equilibrium in congestion games}
\author{Yakov Babichenko\footnote{Technion, Israel institute of Technology. E-mail: yakovbab@technion.ac.il.} \ and Aviad Rubinstein\footnote{Stanford University. E-mail: aviad@cs.stanford.edu.}}
\date{}

\usepackage{color}

\usepackage{footmisc}

\usepackage{booktabs}
\usepackage{ltxtable}
\usepackage{complexity}
\newclass{\CLS}{CLS}
\newclass{\CCLS}{CCLS}
\newclass{\SCCLS}{Smooth-CCLS}
\newclass{\EOPL}{EOPL}
\newclass{\EOML}{EOML}
\newclass{\UEOPL}{UniqueEOPL}
\usepackage{amssymb}
\usepackage{amsmath}
\usepackage{amsthm}
\usepackage{MnSymbol}
\usepackage{float}
\usepackage{hyperref}
\usepackage{tikz}
\usetikzlibrary{calc,patterns,shapes}
\usetikzlibrary{arrows.meta}
\usepackage{footmisc}
\usepackage{url}

\usepackage[left=1.00in, right=1.00in, top=1.00in, bottom=1.00in]{geometry}

\newcommand{\Ex}{\mathbb{E}}

\newcommand{\Prob}{\mathbb{P}}

\newcommand{\vzero}{\vec{0}}
\newcommand{\vone}{\vec{1}}

\renewcommand{\cL}{\mathcal{L}}

\usepackage{bbm}
\newcommand{\1}{\mathbbm{1}}

\newcommand{\hx}{\hat{x}}
\newcommand{\hy}{\hat{y}}
\newcommand{\hX}{\hat{X}}

\newcommand{\ox}{\overline{x}}
\newcommand{\oy}{\overline{y}}
\newcommand{\oX}{\overline{X}}
\newcommand{\oY}{\overline{Y}}
\newcommand{\oox}{\mathring{x}}

\newcommand{\tx}{\tilde{x}}
\newcommand{\ty}{\tilde{y}}
\newcommand{\tX}{\tilde{X}}
\newcommand{\tY}{\tilde{Y}}
\newcommand{\tO}{\tilde{O}}
\newcommand{\ui}{\underline{i}}
\newcommand{\oi}{\overline{i}}

\newcommand{\oep}{\overline{\epsilon}}

\DeclareMathOperator{\supp}{supp}
\DeclareMathOperator{\Var}{Var}

\newtheorem{lemma}{Lemma}

\newtheorem*{lemma*}{Lemma}
\newtheorem{theorem}{Theorem}
\newtheorem{corollary}{Corollary}
\newtheorem*{corollary*}{Corollary}
\newtheorem*{proposition*}{Proposition}
\newtheorem{proposition}{Proposition}
\newtheorem{assumption}{Assumption}
\newtheorem{fact}{Fact}
\newtheorem*{theorem*}{Theorem}
\newtheorem*{tay-theorem}{Taylor's Theorem}
\newtheorem*{mtheo*}{Main Theorem}

\newtheorem{innercustomgeneric}{\customgenericname}
\providecommand{\customgenericname}{}
\newcommand{\newcustomtheorem}[2]{%
  \newenvironment{#1}[1]
  {%
   \renewcommand\customgenericname{#2}%
   \renewcommand\theinnercustomgeneric{##1}%
   \innercustomgeneric
  }
  {\endinnercustomgeneric}
}

\newcustomtheorem{customthm}{Theorem}
\newcustomtheorem{customlemma}{Modified Lemma}

\theoremstyle{definition}
\newtheorem{definition}{Definition}
\newtheorem{meta}{Meta Definition}
\newtheorem{challenge}{Challenge}
\theoremstyle{remark}
\newtheorem{remark}{Remark}
\newtheorem*{remark*}{Remark}

\begin{document}

\maketitle
\begin{abstract}

We consider (i) the problem of finding a (possibly mixed) Nash equilibrium in congestion games, and (ii) the problem of finding an (exponential precision) fixed point of the gradient descent dynamics of a smooth function $f:[0,1]^n \rightarrow \mathbb{R}$.
We prove that these problems are equivalent.
Our result holds for various explicit descriptions of $f$, ranging from (almost general) arithmetic circuits, to degree-$5$ polynomials.
By a very recent result of~\cite{FGHS20}, this implies that
these problems are $\PPAD\cap\PLS$-complete. As a corollary, we also obtain the following equivalence of complexity classes:
 $$\CCLS = \PPAD\cap\PLS.$$

\end{abstract}

\setcounter{page}{0}
\thispagestyle{empty}
 
\newpage
\section{Introduction}

\subsubsection*{Game theory}

Nash equilibrium is the central solution concept in Game Theory. 
In recent decades, its universality as a solution concept has come under a computational critique: Finding a Nash equilibrium is, in general, a computationally intractable problem~\cite{CDT,DGP}; and if centralized, specially designed algorithms cannot find an equilibrium, how can we expect decentralized, selfish agents to converge to one?

Identifying classes of games where Nash equilibrium is tractable is an important open problem (e.g.~\cite{Daskalakis-thesis}).
Congestion games~\cite{Rosenthal73} are arguably the most important class of games for computer scientists, capturing useful applications like resource sharing and routing in networks. 
They play a fundamental role in Algorithmic Game Theory since the inception of the field two decades ago.  
In particular, a celebrated%
\footnote{Including well-deserved citations for prizes named after Godel, Tucker, von Neumann (twice), and Kalai.} research program on Price of Anarchy has focused on characterizing the outcomes of these games at equilibrium. 
But an important question remains about this modeling assumption:
\begin{quote}
Do agents in congestion games play equilibrium strategies?
\end{quote}
Again, if we're hoping for selfish, decentralized agents to converge to an equilibrium, at the very least an algorithm should be able to find one.

At the source of computational hardness for general games is the combinatorial number of {\em mixed} strategies that the algorithm has to consider.
In contrast, congestion games admit pure equilibria, which makes them a prime candidate for tractable equilibrium computation. 
The seminal work of~\cite{FPT} (see also~\cite{AS08,SV08}) ruled out the most immediate application of this idea: while pure equilibria exist for congestion games, finding them is again intractable. 
But more generally, it is natural to hope that the special structure of congestion games allows for efficient computation of (possibly mixed) Nash equilibria.

\subsubsection*{Optimization}

Our paper also touches on a fundamental problem in continuous optimization: Minimize a continuous function $$f:[0,1]^n \rightarrow [0,1].$$
Clearly we cannot hope to find a global minimum, 
but, under a mild smoothness assumption (see Remark~\ref{rem:smoothness}), there is a useful and total local solution concept: (approximate) fixed points of the gradient descent dynamics (``KKT points''). 
By a simple potential argument, gradient descent itself always converges to such $\epsilon$-approximate fixed point in $\poly(1/\epsilon)$ iterations. 
But in some cases, e.g.~when $f$ is also convex, algorithms like the center of gravity method converge in $\poly(n,\log(1/\epsilon))$ iterations~\cite[Theorem 2.1]{Bubeck15}.
For what nonconvex functions can we achieve exponentially fast (in $1/\epsilon$) convergence?
(Indeed, accelerated methods make it possible for some nonconvex functions, e.g.~\cite[Theorem 6.3]{CDHS18}.)

In the worst case, when $f$ is given as a black-box oracle, unconditional query complexity results (e.g.~\cite{Vavasis93,BM20,CDHS20}) rule out exponential convergence.
{\em When $f$ has a succinct representation, e.g.~as a sum of monomials, query complexity arguments can only rule out exponential convergence for specific algorithms.} 
To address the possibility of exponential speedup by general algorithms, we need to turn to computational complexity. 
In this paper, we study the computational problem {\sc GD-FixedPoint} (see Meta Definition~\ref{def:gdfp}), where the goal is to find an approximate fixed point of the gradient descent dynamics. We consider 4 variants of this problem, where the potential $f$ is succinctly represented assuming any of the following different structures: 

\begin{enumerate}
\item In the most general case, $f$ is succinctly described by an arithmetic circuit%
\footnote{Formally we need to require some restrictions on this arithmetic circuit to avoid doubly-exponentially large numbers by repeated squaring, see~\cite{DP20,FGHS20} for details.}
\item When defining the class \CCLS,~\cite{CLS} consider the case where $f$ is guaranteed%
\footnote{Formally to avoid promise problems the algorithm may return a certificate that the component-wise convexity is violated. In contrast, for the fourth variant it is easy to tell if a monomial is component-wise convex.} to be component-wise convex.
\item 
%When defining one of the seven computational problems, {\sc KKT},~
\cite{CLS} consider the restriction that $f$ is explicitly given as a sum of monomials (note that this is orthogonal to the component-wise convexity condition). This problem is called {\sc KKT} in \cite{CLS}.
\item In the most restrictive case, $f$ is a sum of component-wise convex monomials of constant degree.
\end{enumerate}

\begin{remark}\label{rem:smoothness}
For all of these cases we make the assumption that $f$ has a finite (but possibly exponentially large) Lipschitz constant of the gradients. We call this Lipschitz constant the \emph{smoothness parameter}.
%This assumption trivially holds for polynomials. For arithmetic circuit, one way to obtain a total problem is to consider only smooth gates (i.e. no max/min gates).
\end{remark}

\subsubsection*{Computational complexity}

From a technical perspective, the problem {\sc Congestion} of finding a Nash equilibrium in congestion games, and the aforementioned {\sc GD-FixedPoint}, belong to a special class, \TFNP, of {\em total} search problems, i.e.~ones that always have a solution.
The aforementioned works prove that finding a Nash equilibrium in a general game is complete for the class~\PPAD, which captures fixed point computations for {\em continuous functions}. Finding a pure equilibrium in congestion games is complete for the class~\PLS, which captures local maximum computation for {\em discrete potential funcions}. 
Finding a Nash equilibrium in congestion games is a special case of both problems, i.e.~it lies in $\PPAD \cap \PLS$, and hence unlikely to be complete for either class%
\footnote{\PPAD~and~\PLS~are generally believed to be incomparable, e.g.~due to oracle separations~\cite{BM04, BJ12}.}.

Since the seminal paper of~\cite{CLS}, the class  $\PPAD \cap \PLS$ has been at the frontier of research on total search problems (e.g.~\cite{FGMS17,GDHKS18-CLS,ConMap1,ConMap2,GKRS19-CC-TFNP,EPRY20,GS20,FS21,FGHS20}).
 This research direction is particularly exciting because these problems lie quite low in the \TFNP~hierarchy, so a-priori any of them could plausibly admit a polynomial time algorithm.
Several natural problems have been identified, and a few complexity sub-classes have been defined (including \CLS\footnote{Very recent breakthrough shows that in fact $\CLS = \PPAD \cap \PLS$~\cite{FGHS20}!}, \CCLS, \EOPL, and \UEOPL).
Yet, until this paper, neither $\PPAD \cap \PLS$ nor any of its sub-classes were known to have any natural%
\footnote{\label{fn:natural} In the sense that all known \CLS-complete~\cite{FGMS17,ConMap1}, \UEOPL-complete~\cite{ConMap2}, and $\PPAD\cap\PLS$-complete~\cite{FGHS20} problems are described in terms of an oracle function succinctly represented by a circuit. 
From an algorithm-designer point of view, it is typically more natural -and often much easier- to study the {\em query complexity} of such oracle problems. 
(The computational complexity of oracle problems is still very interesting, e.g.~as a stepping stone for understanding natural problems!)
See further discussions in~\cite{FG18,SZZ18}.} 
complete problems.

\subsection*{Our results and their implications}

\begin{theorem*}[Main result, informal] \hfill

{\sc Congestion} and all 4 variants of {\sc GD-FixedPoint} are computationally equivalent. 
\end{theorem*}
In the context of optimization, it immediately follows from our main theorem that local maximization (in the {\sc GD-FixedPoint} sense) of an arbitrary smooth function can be reduced to an optimization of a degree 5 polynomial. 
Very recently%
\footnote{In fact,~\cite{FGHS20} was developed concurrently and independently to a previous version of our paper. The proofs in this version (while still unfortunately long) have been simplified significantly since~\cite{FGHS20}  show that {\sc GD-FixedPoint} is already hard in only two dimensions.}~\cite{FGHS20} proved that {\sc GD-FixedPoint} is $\PLS \cap \PPAD$-complete. We therefore obtain the following corollaries:
\begin{corollary*}[Complexity consequences]\hfill

\begin{itemize}
\item {\sc Congestion} and all 4 variants of {\sc GD-FixedPoint} are $\PLS \cap \PPAD$-complete.\footnote{The completeness for the first variant of GD-FixedPoint was proved by \cite{FGHS20}.} This is the first case of a natural%
\footnote{See Footnote~\ref{fn:natural}.} complete problem for $\PLS \cap \PPAD$ or any of its subclasses.
\item $\CCLS = \PLS \cap \PPAD$.
\end{itemize}
\end{corollary*}

\subsection{Related work}
Our proof combines ideas from the two long lines of work on~\PLS-completness (e.g.~\cite{PLS,Krentel90,SY91,FPT,SV08,AS08,ARV08,HHKS13}) and \PPAD-completeness (e.g.~\cite{DGP,CDT,CD09,CDDT10,GV14,Mehta14,CDO,CPY17,RubGraph,DFS20,MM20}). It also uses ideas from our recent paper on the {\em communication} complexity of Nash equilibrium in potential and congestion games~\cite{BR20}.

Recently,~\cite{DSZ20} also study a variant of the problem {\sc GD-FixedPoint}. The focus of their paper is on $1/\poly(n)$ or even constant approximation, whereas we seek exponentially precise approximation. Note that for {\sc GD-FixedPoint}, it is easy to compute a $1/\poly(n)$-approximate fixed point in polynomial time because the potential is bounded. In contrast, for the gradient descent-ascent dynamics ({\sc GDA-FixedPoint}) the potential may increase and decrease alternatingly, which makes the problem much harder for approximation. 
Indeed, for an even harder variant of approximate {\sc GDA-FixedPoint} where the domain is not a cartesian product,~\cite{DSZ20} prove that the problem is \PPAD-complete. (Approximate {\sc GDA-FixedPoint} over e.g.~the hypercube remains an exciting open problem.)

Historically, the first Price of Anarchy (PoA) results for congestion games focused on pure equilibrium (e.g.~\cite{KP09,RT00,CV07}), and were then extended to mixed Nash equilibrum (e.g.~\cite{Vetta02,ADGMS11,AAE13}).
The barriers to equilibrium computation in congestion games inspired the development of techniques, in particular Roughgarden's smoothness framework, for bounding the PoA with respect to tractable solution concepts such as correlated equilibrium (e.g.~\cite{BHLR08,Roughgarden15}). Our results show that mixed Nash equilibrium in congestion game is intractable, further motivating the latter line of work on analyzing tractable solution concepts.

Another interesting approach to circumventing the computational barriers to computing pure equilibrium in congestion games considers beyond worst case analysis. Specifically a recent line of works applies smoothed analysis to the special case of polymatrix identical interest games~\cite{ER17,ABPW17,BCC19,BKM20,CGVYZ20}. It is an interesting open question whether considering mixed Nash equilibrium can improve those results. Another interesting question is whether this analysis can extend to more general congestion games (for an appropriate definition of smoothing).

\subsection{Technical highlights}\label{sec:highlights}
In this section we exposite some techincal highlights from the proof of our main reduction. 
Specifically, we reduce the $\PLS \cap \PPAD$-complete {\sc 2D-GD-FixedPoint} to finding a Nash equilibrium in a sub-class of congestion games that we now introduce.

\paragraph{Polytensor identical-interest games}
We consider a class of games that satisfies two orthogonal constraints:
\begin{description}
\item[$5$-Polytensor\footnotemark games:]
\footnotetext{These games are sometimes called {\em hypergraphical games} following the seminal~\cite{PR}, but we believe that the name polytensor better reflects the fact that they generalize {\em polymatrix games}~\cite{Yanovskaya68} (which would be $2$-polytensor games).}
 The utility of Player $i$ is given as a sum of sub-utilities, each depending on the strategies of only 5 players (including $i$).
\item[Identical interest\footnotemark games:]
\footnotetext{In our context these are equivalent to {\em potential} games, and sometimes they're also called {\em coordination} games.}
All players have the same utility function.
\end{description}

%\newtext{
\begin{remark}\label{rem:identical-interest}
In identical interest $5$-polytensor games the local sub-utilities have identical interest. Note that this does not create an identical-interest game (because every player has different local interactions). However, we can alternatively consider a strategically equivalent game where every player gets the sum of \emph{all} sub-utilities  in \emph{all} local interactions (including those where she does not participate). This latter game is an identical interest game. The sets of Nash equilibria in these two games coincide.
\end{remark} 
%}

We reduce {\sc 2D-GD-FixedPoint} to finding a Nash equilibrium in the class of {\em polytensor identical-interest games}, i.e.~games that satisfy both desiderata. 
Our primary motivation for studying this class of games is as an intermediary problem for our reductions, but it may also be of independent game-theoretic interest (e.g.~it is a natural generalization of the popular class of polymatrix identical-interest games~\cite{CLS,ER17,ABPW17,BCC19,BKM20,CGVYZ20}). See Section~\ref{sec:polytensor} for details.

\subsubsection*{Potential+imitation game}

Our reduction begins with a potential function $\phi:[0,1]^2 \rightarrow [0,1]$, which we will try to {\em maximize}.\footnote{There is somewhat of a discrepancy between the optimization literature which mostly talks about minimization and game theory where players try to maximize the potential function. We side with the game theorists here, but of course the gradient descent dynamics should be applied to $-\phi$.}
Consider the $2$-player infinite-actions potential game where Player $j=1,2$ chooses a value for coordinate $x_j \in [0,1]$, and every player's utility is $\phi(x)$. 
In any pure equilibrium, the players collectively choose a point $x \in [0,1]^2$ where any unilateral change to one coordinate, in particular an infinitesimal change, does not improve the potential --- hence in particular this is a fixed point of the gradient descent dynamics.
Before we talk about reducing the number of actions, there is an infamous technical obstacle around dealing with {\em mixed} strategies: a mixed strategies of the players need not correspond to any particular point $x$, let alone one that is a fixed point of the gradient descent dynamics. 

\begin{challenge}
Mixed strategies need not correspond to any particular point.
\end{challenge}

It is not possible to completely rule out mixed strategies, but we will incentivize players to choose strategies that correspond to $x_j$'s that are very tightly concentrated. If the potential function satisfies a mild smoothness condition (see Remark~\ref{rem:smoothness}), we can guarantee that the gradient is approximately constant in any sufficiently small neighborhood. Therefore if the mixed $x$'s are  supported in a neighborhood that does not contain an approximate gradient descent fixed point, there is some direction where the gradient always shows an improving deviation.

To force the mixed strategies to concentrate, we use the {\em imitation game}, which has been the driving force in recent advances on complexity of Nash equilibrium (e.g.~\cite{QCNash,RW16,Rub2p,RubGraph,CCNash,GR18,GKP19}, and in particular the authors' own recent work on communication complexity of Nash equilibrium in potential games~\cite{BR20}.
Specifically, we add $2$ additional players who choose a second point $y \in [0,1]^2$, and update the identical utility function to 
\begin{gather}\label{eq:imitation}
U(x,y) :=-||x-y||_2^2+\epsilon \phi(x),
\end{gather} 
for a sufficiently small $\epsilon>0$. 

First, since the main incentive of each player is to minimize the square distance from the opponent the best-reply against any \emph{mixed} strategy of the opponent is $\epsilon$-close to the expectation of the opponent's strategy. Mutual imitation of both players implies that their (mixed) strategies are supported $\epsilon$-close to each other.\footnote{This concentration of mixed strategies, although intuitive, is the most challenging part of our proofs. The formal proof contains quite a few intermediate steps. See Section~\ref{sec:analysis} for details.} If in this region the gradient in some direction is positive (and does not get out of the region $[0,1]^2$), then the $x$-player has an incentive to mildly deviate toward this direction: her imitation loss is quadratic while her improvement in the gradient is linear. For sufficiently small deviation the linear term is more significant and hence it cannot be an equilibrium.

This trick allows us to focus on pure strategies (for now - mixed strategies will come back to haunt us throughout the proof). We now move to reduce the number of actions.

\subsubsection*{Strategic binary representation and veto players}

\begin{challenge}
We are interested in games with a polynomial number of actions for each player.
\end{challenge}

To obtain a finite game, we discretize the real numbers to finite precision.
But we need to work with exponential precision (approximations up to $\pm 1/\poly(n)$ are trivial), and so can no longer afford to let a single player represent each coordinate.
 
We consider binary representation of each $x_j,y_j \in [0,1]$ and we replace a single player who chooses a real number by a polynomial number of \emph{bit-players}, each chooses a single bit in the binary representation. This modification is insufficient: if $x_j$ wants to move infinitesimally from the number  $011...1$ to the number $100...0$ (as is indeed the case in the potential+imitation game~\eqref{eq:imitation}) this requires a \emph{joint} switching of actions of all bit players. In particular, a bad equilibrium may arise where players are playing $011...1$, the gradient in this coordinate is positive, but no bit-player wants to switch her action unilaterally. 

\begin{challenge}
We are interested in a binary representation where there is always a player who can unilaterally shift the number infinitesimally in either direction.
\end{challenge}

Over the last few years we've tried several different approaches to overcoming this obstacle. For example, one natural idea is to use the Gray code (or variants) where every two consecutive strings differ by only one bit. The issue with this approach is that it greatly complicates bit-player utilities: for the least-significant-bit-player to know in which direction it prefers to go it must know the parity of all other bit players for that coordinate. (And things get even messier when those more significant bit players use mixed strategies...)

We overcome this obstacle by introducing an additional \emph{veto player}: the veto player has the power to impose a veto on any suffix of the binary string. The possible vetoes are $011...1$ and  $100...0$. If the veto player imposes such a veto then the actions of the bit players in the suffix are essentially ignored and, instead, the veto action of the veto player is counted. This resolves the problematic issue above: if the $x$-players are playing $011...1$ then they can move to $100...0$ by a unilateral deviation of the veto player.
The veto player has an additional important role in our reduction: it has an incentive to veto bit players who are playing mixed (because those mixed strategies cause a high imitation loss). So the existence of the veto player allows us to argue that the bits of $x$ will essentially be pure.\footnote{For this task of vetoing players that are playing mixed we need to mildly enrich the set of possible vetoes. In addition to $011...1$ and $100...0$ the veto player can choose also the vetoes $00...0$, $11...1$, $0011...1$ and $1100...0$. It turns out that for our reduction this enrichment is sufficient.}

The next problem that arises with the basic imitation+potential game~\eqref{eq:imitation} is that in a ``natural'' problem the utility should not be specified by the arithmetic circuit computing $\phi()$.

\subsubsection*{Circuit players for computing the potential and its gradients}

\begin{challenge}
The utilities of players in our game should be explicitly given, i.e. we cannot use the term $\phi(x)$ in the utility. 
\end{challenge}

Similarly to the existing literature on the complexity of Nash equilibrium (see e.g., \cite{DGP,CDT,RubGraph}) in general games (i.e., not potential games), instead of writing the potential $\phi$ in the utility  we introduce \emph{circuit players} that implement the circuit gates and are responsible for computing $\phi$.

\begin{remark}\label{rem:decreasing-weights}
Unlike \cite{DGP,CDT,RubGraph} our circuit players have identical interest both with their predecessors and their successors in the circuit. We set the utility in the interaction with predecessors higher than the interaction with predecessors. This incentivizes each circuit player to match its bit to the correct computation of the predecessors (rather than adjusting her input in a way that her successor's action will be correct). This results in an exponential decay of the utilities in the depth of the circuit. In particular, the computed potential $\phi$ in the output has exponentially lower weight in the common utility relative to the input. 
\end{remark}

Circuit players allow us to correctly compute the potential, but that comes with the cost of separating the potential $\phi(x)$ from the players who determine $x$:

\begin{challenge}
Before we introduced circuit players, the gain of deviation from $x$ to $x'$ causes an immediate reward once the potential increases from $\phi(x)$ to $\phi(x')>\phi(x)$. Namely, when $x$-players switch from $x$ to $x'$ they gain the increase in the potential. Once we introduce circuit players, the $x$-players will not get the potential gain until the circuit players will update their actions and will calculate $\phi(x')$ instead of the currently computed $\phi(x)$. So how should we incentivize the $x$-players to unilaterally deviate in the direction of the gradient?\footnote{Moreover, such a deviation creates a loss in the circuit computations. This issue will be discussed in Challenge~\ref{challenge:input-circuit}.}
\end{challenge}

In our reduction, those that will be responsible for deviation toward better potential are the $y$-players. To give them immediate reward for moving toward better potential we modify the circuit to calculate $\nabla_j \phi(x)$ too. We add to the utility of the $y$-team the terms $\epsilon y_j \nabla_j \phi(x)$ (for $j=1,2$). Whenever $\nabla_j \phi(x)>0$ the $y$-team has an incentive to increase $y_j$. Whenever $\nabla_j \phi(x)<0$ the $y$ team has an incentive to decrease $y_j$. Once the $y$-team moves toward the gradient, the $x$-team has an incentive to move there too because it tries to imitate $y$. 

The next challenge in the reduction is to bootstrap this tiny incentive for $y$-players to deviate slightly in the gradient direction to a large movement of $x$-players.

\subsubsection*{Guide player and sampling gadgets}

\begin{challenge}\label{challenge:input-circuit}
When the $x$-team moves, it should flip the input to the circuit, causing the first gate to pay for wrong computation. This force is exponentially more significant than the force of trying to imitate $y$ that have mildly moved toward the better potential in the direction of the gradient.\footnote{Correct circuit computations has exponentially higher weight than the computed gradient (see e.g.~Remark~\ref{rem:decreasing-weights}). The $y$-team moves toward the gradient in a step size that is proportional to the gradient. Therefore, indeed the force of maintaining correct input computations is more significant.} 
\end{challenge}

We resolve this obstacle by introducing three modifications. First, we let the $x$-team and the $y$-team choose real numbers with precision that is much finer than the bits that influence the circuit's input. We denote by $N_{in}$ the precision of the circuit's input, where the total number of bits for the $x$-team is $K\gg N_{in}$. This allows the $x$-team to imitate the move of the $y$-team in the direction of the gradient whenever \emph{$x$ is far from a multiple of $2^{-N_{in}}$}. If $x$ is located close to the $2^{-N_{in}}$-grid, the $x$-team does not want to change the first $N_{in}$ bits due to the reason of Challenge~\ref{challenge:input-circuit}.

The second modification applies the \emph{sampling technique} from the \PPAD-literature. This technique was first introduced for 3-D in~\cite{DGP} (who called it the ``averaging maneuver''), and later adapted to high dimensions by~\cite{CDT}. The basic idea is that instead of having a single $x$-team we introduce many copies for the $x$-team. Each team tries to imitate a different \emph{shift} of the $y$-team. We design the shits to ensure that if one $x$-team is located close to the $2^{-N_{in}}$-grid then the other teams necessarily are located far from the grid.

Challenge~\ref{challenge:input-circuit} deals with an exponential gap between the force pushing towards the gradient and the force resisting flipping input bits for the circuit, so taking a uniform average over a polynomial number of samples cannot resolve this issue. 
Instead, our third modification introduces a gadget that carefully and adaptively assigns weights to the samples. This gadget is inspired by a classic \PLS-completeness proof of~\cite{Krentel90}.
Specifically, we add a \emph{guide player}, whose role is to choose a ``well behaved'' sample of $x$: a sample that performs accurate imitation, computes the potential correctly, and has a relatively high value of the potential. 
The circuit computation of a particular $x$-sample is given high weight only if it is chosen by the guide player. Otherwise, if an $x$-sample is not chosen by the guide player, its circuit computation weight is negligible (relative to the incentive to imitate $y$). Our desire is that these problematic $x$-samples that are close to the boundary in one of the coordinates will not be chosen by the guide player, which will allow the $x$-sample to change the first $N_{in}$ bits to accurately imitate the $y$-team (with a shift). Thereafter, the circuit players will adjust their computation to the new input, and then the guide player could choose this $x$-sample again.

%\Yakov{New text} 
The intuition for  guide player's avoidance from choosing close-to-grid problematic samples is as follows. We compare the utility for the guide player in the problematic sample $x_{i,j}$ with her utility in a well behaved sample $x_{i',j}$ that is located far from the grid. The imitation in $x_{i',j}$ sample is very precise because the phenomenon of avoiding switching bits due to wrong circuit computations does not occur when the imitation target is far from the grid: an infinitesimal movement toward the imitation target does not change the first $N_{in}$ bits. Circuit's computation in $x_{i',j}$ is also correct because the first $N_{in}$ input bits are deterministic. The only term which, in principle, might be inferior for $x_{i',j}$ is the potential term. However, this cannot be the case either: assume that $\nabla_j \phi(x)>0$. In such a case $y_j$ has an incentive to (mildly) \emph{over-imitate} her target $\hx_j$. If the problematic team $x_{i,j}$ is located from the left side of the grid-point this means that the potential that is computed there is \emph{weakly lower} than the potential in all other samples, and in particular, the potential at $x_{i',j}$ is either identical or better than in the problematic sample; see Figure \ref{fig:sam} (a). If, on the other hand, the problematic team of $x_{i,j}$ is located from the right side of the grid-point this means that the imitation target $\hy_{i,j}$ is also located from the right side the grid-point (recall that $y_j$ over-imitates $\hx_j$). In such a case $x_j$ can smoothly move toward her imitation target $\hy_{i,j}$ without changing the first $N_{in}$ bits. See Figure \ref{fig:sam}(b).

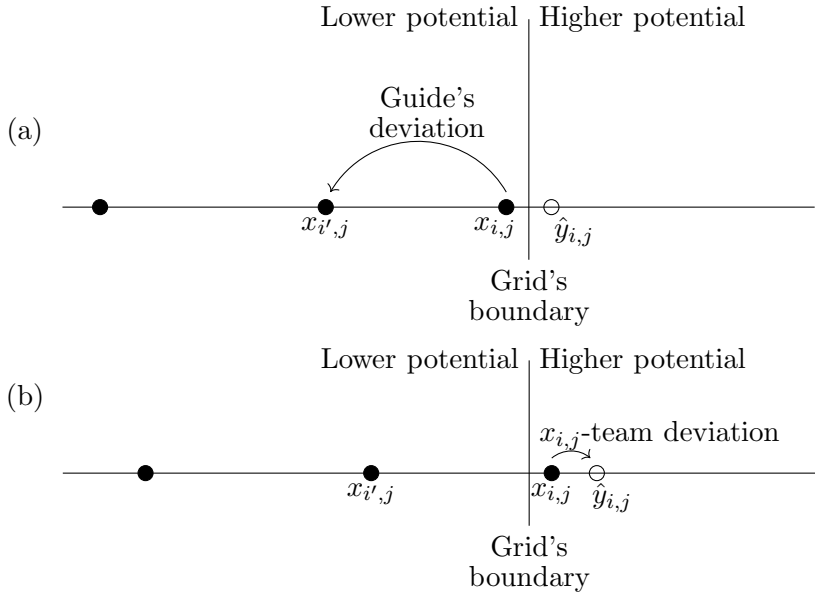
\begin{figure}
\caption{Figure \ref{fig:sam}(a) demonstrates the profitable deviation of the guide player in case the problematic sample $x_{i,j}$ is located from left side of the boundary. Figure \ref{fig:sam}(b) demonstrates the profitable deviation of $x_{i,j}$ team in case this problematic sample is located from right side of the boundary.}\label{fig:sam}
\begin{center}
\begin{tikzpicture}
\draw (1,0)--(11,0);
\filldraw (1.5,0) circle(0.1);
\filldraw (4.5,0) circle(0.1);
\draw (7.5,0) circle(0.1);
\filldraw (6.9,0) circle(0.1);
\draw (7.2,-0.7)--(7.2,2.5);
\node[below] at (4.5,0) {$x_{i',j}$};
\node[below] at (6.75,0) {$x_{i,j}$};
\node[below] at (7.8,0) {$\hy_{i,j}$};
\node[below] at (7.2,-0.7) {Grid's};
\node[below] at (7.2,-1.1) {boundary};
\draw[->] (6.9,0.2) arc (30:150:1.35);
\node[above] at (5.85,1.2) {Guide's};
\node[above] at (5.85,0.8) {deviation};
\node[left] at (7.2,2.5) {Lower potential};
\node[right] at (7.2,2.5) {Higher potential};
\node at (0.5,1) {(a)};
\end{tikzpicture}

\begin{tikzpicture}
\draw (1,0)--(11,0);
\filldraw (2.1,0) circle(0.1);
\filldraw (5.1,0) circle(0.1);
\filldraw (7.5,0) circle(0.1);
\draw (8.1,0) circle(0.1);
\draw (7.2,-0.7)--(7.2,1.5);
\node[below] at (5.1,0) {$x_{i',j}$};
\node[below] at (7.5,0) {$x_{i,j}$};
\node[below] at (8.3,0) {$\hy_{i,j}$};
\node[below] at (7.2,-0.7) {Grid's};
\node[below] at (7.2,-1.1) {boundary};
\node[left] at (7.2,1.5) {Lower potential};
\node[right] at (7.2,1.5) {Higher potential};
\draw[->] (7.5,0.2) arc (130:50:0.4);
\node[right] at (7.2,0.5) {$x_{i,j}$-team deviation};
\node at (0.5,1) {(b)};
\end{tikzpicture}
\end{center}
\end{figure}

\section{Definitions and Preliminaries}\label{sec:def}

In this section we formally define the computational problems that will be discussed in this paper. Along the definitions we provide some (simple) observations.

\subsection{Game Theoretic Definitions}

Our main interest game-theoretic object in this paper is the class of (explicit) congestion games with a polynomial number of facilities and actions for each player. 
We show that, in terms of complexity of finding a Nash equilibrium, this class is equivalent to $c$-polytensor identical interest  games (a generalization of polymatrix identical interest games\footnote{Also called {\em network coordination games} in \cite{CLS}.}).
We now formally introduce each of this classes.

\subsubsection{Congestion games}\label{sec:cong}

In congestion games (see  \cite{Rosenthal73,MS}) we have a set of facilities $F$. Each player $i$ has a collection of subsets $S^i_1,...,S^i_{m_i}\subset F$ which are her actions. Given an action profile $S=(S^1,...,S^n)$ the congestion on each facility $f$ is defined to be the number of players that use this facility (i.e., $c_f(S)=|\{i:f\in S^i\}|$). Each facility has a cost function $l:[n]\rightarrow \mathbb{R}$ and the utility of player $i$ is defined to be 
$u_i(S^1,...,S^n)=\sum_{f\in S^i} l(c_f(S))$; namely, each player pays the total costs of all facilities she has chosen. In \emph{explicit} congestion games the possible actions of players (i.e., $S^i_1,...,S^i_{m_i}\subset F$) are given explicitly.

The problem $\textsc{Congestion}$ gets as an input an explicit congestion game and an $\epsilon_N$ and outputs an $\epsilon_N$-Nash equilibrium of the game.

\subsubsection{Polytensor Identical Interest Games}\label{sec:polytensor} 

An $n$-player \emph{polymatrix game} (see \cite{Yanovskaya68,Howson72}) is given by a tuple of two-player utility functions $(u^i_{i,j}(a_i,a_j))_{i,j\in [n]}$ and the utility of player $i\in [n]$ is the sum of these two-player utilities $u^i(a_i,a_{-i})=\sum_{j\in [n]} u^i_{i,j}(a_i,a_j)$. In the case where $u^i(a_i,a_j)=u^j(a_i,a_j)$ for every $i,j$ we denote this term by $u^{i,j}(a_i,a_j)$.% and the game is a potential game whose potential is $\sum_{i,j \in [n]} u^{i,j}(a_i,a_j)$. 

In a generalization of polymatrix games that we call \emph{$c$-polytensor games}, instead of two-player interactions (i.e., the local interaction can be described by a matrix), we allow $\le c$-player interactions (i.e., local interactions can be described by a tensor of order $\le c$). Namely, a $c$-polytensor game is given by a tuple $(u^i_{S}(a_S))_{i\in [n], S\subset [n], |S|=c, i\in S}$
where $a_S$ denotes the action profile of the players in $S$.
Similarly to polymatrix games the utility of a player is given by the sum of all her interactions: $u^i(a)=\sum_{S\subset [n], |S|=c, i\in S} u^i_{S}(a_S)$. In the case where $u^i(a_S)=u^j(a_S)$ for every $i,j\in S$ we denote this term by $u^S(a_S)$. % and this game is a potential game whose potential is $\sum_{S\subset [n], |S|=c} u^S(a_S)$. 
In case where the number of actions for each player is bounded by $m$ note that the representation size of a $c$-polytensor game is $O(n^{c+1}m^c)$. In particular, if $c$ is a constant and $m=poly(n)$ then $c$-polytensor games admit a succinct representation.

The problem $c-${\sc Polytensor-IdenticalInterest} gets as an input a $c$-polytensor identical interest  game and an $\epsilon_N$ and outputs an $\epsilon_N$-Nash equilibrium of the game.

Our primary purpose for studying the problem {\sc Polytensor-IdenticalInterest} is as an intermediary problem for our reductions, but its complexity may also be of independent game theoretic interest.

We argue that every $c$-polytensor identical interest game for a constant $c$ is also an explicit congestion game.
 The set of facilities is $\{a_S:S\subset [n],|S|=c,a_i\in A_i \text{ for every } i\in S \}$ where $A_i$ denotes player's $i$ action set. By choosing an action $a_i \in A_i$ player $i$ chooses the set of facilities $\{(a_i,a_{S'}) : S'\subset [n],|S'|=c-1 \}$; namely she chooses all facilities where her action is relevant (i.e., $i\in S$) and her action is $a_i$. In congestion games we should assign a common utility for all players that is a function of \emph{the number} of users of each facility. In this case we set the common utility of each facility $a_S$ to be $u^S(a_S)$ if it has $c$ users; and $0$ otherwise. Note that indeed only the ``true'' actions $a_S$ have non-zero utility and indeed after summing the utilities over all facilities we obtain exactly the common utility of the $c$-polytensor identical interest  game $\sum_{S\subset [n], |S|=c} u^S(a_S)$.
Hence, the class of $c$-polytensor identical interest  games is  essentially a subclass of explicit congestion games.\footnote{The converse is true (i.e., a congestion game is a $c$-polytensor identical interest  game) whenever for every facility the number of players that can potentially use it (i.e., have an action that includes this facility) is bounded by $c$.}
In particular, we have the following:
\begin{corollary}\label{fact:polytensor-vs-congestion} For every constant $c$ we have
$$c-\text{\sc Polytensor-IdenticalInterest} \le \text{\sc Congestion}.$$
\end{corollary}

\subsubsection{Nash equilibrium}
A mixed action profile $(x_i)_{i\in [n]}$ where $x_i\in \Delta(A_i)$ is an \emph{$\epsilon$-Nash equilibrium} if for every player $i\in [n]$ and every action $a_i\in A_i$ we have $u_i(x_i,x_{-i})\geq u_i(a_i,x_{-i})-\epsilon$. For $\epsilon=0$ such a profile is called a \emph{Nash equilibrium}.

\subsection{Optimization definitions}\label{sub:opt-def}

We consider the total search problem {\sc GD-FixedPoint} that naturally arise in the context of continuous optimization via gradient-based methods:

\begin{meta}[{\sc GD-FixedPoint}]\label{def:gdfp}\hfill

Given a smooth function $\phi:[0,1]^J \to [0,1]$ and an $\epsilon>0$ find a point $x\in [0,1]^J$ such that 
\begin{enumerate}
\item $\nabla_j \phi(x) \geq -\epsilon$ for every dimension $j\in J$ such that $x_j>0$.
\item $\nabla_j \phi(x) \leq \epsilon$ for every dimension $j\in J$ such that $x_j<1$.
\end{enumerate}
\end{meta}

Namely we are searching for a point $x$ such that a movement in the gradient direction will either be impossible (i.e., when $x_j=0$ the $j$-th coordinate cannot be decreased and when $x_j=1$ the $j$-th coordinate cannot be increased) or alternatively, there will be no improvement (approximately zero gradient).

So far we neither specified the form in which the function $\phi$ is given to the algorithm, nor provided any special properties of $\phi$. As we will see, variants of {\sc GD-FixedPoint} capture the class $\CCLS$ and the problem {\sc KKT}\footnote{\cite{FGHS20} also consider a problem that they call KKT which corresponds to the most general variant of {\sc GD-FixedPoint}. They leave as an open problem the computational complexity of \cite{CLS}'s (explicit) {\sc KKT}, which we settle in this paper.}  that have been defined in \cite{CLS}.

%\newtext{
In the first, most general, variant of {\sc GD-FixedPoint} 
$\phi$ is given by an arbitrary well-behaved arithmetic circuit with $\{+,-,\times,\times \zeta,\max,\min,>\}$ gates, where $\times \zeta$ denotes multiplication by a constant. $\max,\min$ gates are not smooth and $>$ gate is not even continuous, but it is promised%
\footnote{To make the problem total, the algorithm may either return a solution or a certificate that this promise is violated; however the $\PPAD \cap \PLS$-complete instances~\cite{FGHS20} always satisfy the promise, so it suffices to prove hardness for this special case.} that $\nabla \phi$ is defined and $\alpha$-Lipschitz in the $\ell_2 - \ell_2$ sense (the parameter $\alpha$ is given as an input to the algorithm). 
The {\em well-behaved} desideratum restricts that the number $\times$ gates on any input-output path is at most logarithmic (in the size of the circuit). This desideratum is important to avoid doubly-exponential numbers from repeated squaring. See~\cite{FGHS20} for a detailed discussion.

In particular, we focus on a special case of this variant where the domain of the potential is $[0,1]^2$ (i.e., $|J|=2$). We call this problem {\sc 2D-GD-FixedPoint}.
%} 
By~\cite{FGHS20}, this special case is as hard as the general high dimensional case:
\begin{theorem*}[\cite{FGHS20}]
{\sc 2D-GD-FixedPoint} is $\PPAD \cap \PLS$-complete.
\end{theorem*}

The second variant of {\sc GD-FixedPoint} is {\sc Con-GD-FixedPoint} where in addition to the $\alpha$-smoothness it deals with \emph{componentwise concave} potentials. Namely, for every $j\in J$ and every fixed $x_{-j}^* \in [0,1]^{J \setminus \{j\}}$ the function $\phi(x_j,x^*_{-j}):[0,1] \to [0,1]$ should be concave.
To guarantee a total variant of the problem we allow the algorithm to output either a gradient fixed point or an evidence for componentwise concavity violation that comes in the form of a triplet of points. We will see (Lemma \ref{lem:ccls=con}) that this problem is equivalent to {\sc CCLS}. 

The third variant of {\sc GD-FixedPoint} is {\sc Expl-GD-FixedPoint} considers instances where $\phi$ is a polynomial. Here the input is given by an explicitly description of monomials' coefficients. Note that this variant is a natural problem since $\phi$ is not described by a circuit.
We will see (Lemma \ref{lem:kkt=expl}) that this problem is equivalent to {\sc KKT}.

Finally the last, most restrictive, variant of {\sc GD-FixedPoint} is {\sc $c$-Deg-GD-FixedPoint} where $\phi$ is given explicitly and a sum of constant degree $c$ monomials, and each monomial is componentwise concave. This problem is a special case of {\sc Con-GD-FixedPoint} since each monomial is componentwise concave (and so is the sum) and it is also a special case of {\sc Expl-GD-FixedPoint} because monomials are given explicitly (and their total number does not exceed $|J|^c$).

\subsubsection{Relation to other problems in the literature}

%\Yakov{Reviewer 1 comment about this subsection: "inconsistencies in the names of problems". I didn't find inconsistencies.}

Now we relate the defined above variants of {\sc GD-FixedPoint} with the existing continuous optimization problems that have been defined in \cite{CLS}.  

The problem {\sc CCLS} that defines the class $\CCLS$ (see  \cite{CLS}) is defined as follows:
\begin{definition}[$\textsc{CCLS}$ \cite{CLS}]\label{def:CCLS}
{\color{white} none}

\textbf{Input:} A potential function $\phi:[0,1]^J \rightarrow [0,1]$ that is given 
an arithmetic circuit  and two constants $\oep$ and $\delta$.

\textbf{Output:} Any of the following:
\begin{itemize} 
\item a point $x\in [0,1]^J$ such that $\phi(x)\geq \phi((1-\delta)x_j,x_{-j})-\oep$ and  $\phi(x)\geq \phi((1-\delta)x_j + \delta,x_{-j})-\oep$ for every $j\in J$; 
\item a pair of points that exhibit a violation of Lipschitz condition; or
\item a triplet of points that exhibit a violation of componentwise concavity violation.
\end{itemize}
\end{definition}
\vspace{2mm}

The two terms $(1-\delta)x_j$ and $(1-\delta)x_j+\delta$ follow from the multiplicative deviation variant of {\sc CCLS} $x_j \to (1-\delta)x_j + \delta e_j$.

\begin{remark}\label{rem:ccls-def}
Our definition of the {\sc CCLS} problem is slightly more restrictive than the original definition in \cite{CLS}: (i) we require that the domain is a hypercube rather than the more general product of simplices; and (ii) we require that the function has bounded smoothness. 
Since we show that our variant is also $\PLS\cap \PPAD$-complete, it implies in particular that the definitions are equivalent. 
\end{remark}

The equivalence of the {\sc CCLS} problem and our formulation via the {\sc GD-FixedPoint} is stated in the following lemma.

\begin{lemma}\label{lem:ccls=con}
{\sc CCLS} is computationally equivalent to {\sc Con-GD-FixedPoint}.
\end{lemma}

%\Yakov{Reviewer 1 has the following comment: "Lemma 1: perhaps it is better to just do one direction: reduce con-GD-FixedPoint to CCLS. Removing the other 
%direction allows you to define CCLS in the usual way (without smoothness)." We address this issue in Remark \ref{rem:ccls-def}. Do we need to emphasize this remark more? Shall we simply follow his suggestion?}  

\begin{proof}
For the reductions in both direction if the output is a triplet of points violating concavity this output is valid for the reduced problem as well. So we restrict attention to the interesting case where the output is a gradient fixed point (in one direction of the reduction) or a solution for {\sc CCLS} (in the opposite direction reduction).

Given an instance of {\sc CCLS} $(\phi,\oep,\delta)$ we set\footnote{For this direction of the reduction the parameter $\delta$ turns out to be irrelevant.} $\epsilon = \oep$. Let $x$ be a  {\sc GD-FixedPoint-Con} solution for $(\phi,\epsilon)$. By concavity we have
\begin{align*}
&\phi((1-\delta)x_j,x_{-j}) \leq \phi(x)-\delta x_j \nabla_j \phi(x)\leq \phi(x)+\delta x_j \epsilon \leq \phi(x) + \epsilon \delta \leq \phi(x) + \oep \\
&\phi((1-\delta)x_j+\delta,x_{-j}) \leq \phi(x)+\delta(1- x_j) \nabla_j \phi(x)\leq \phi(x)+\delta (1-x_j) \epsilon \leq \phi(x) + \epsilon \delta \leq \phi(x) + \oep.
\end{align*}
 
%Conversely, given an instance of {\sc GD-FixedPoint-Con} $(\phi,\epsilon)$ we set $\oep=\frac{\epsilon^3}{64J\alpha^2}$ and $\delta=\frac{\epsilon}{4\alpha}$. Let $x$ be a  {\sc GD-FixedPoint-Con} solution for $(\phi,\oep,\delta)$. For every $j\in J$ with $x_j>0$ 

Conversely, given an instance $(\phi,\epsilon)$ of {\sc GD-FixedPoint-Expl} we set $\oep=\frac{\epsilon^3}{64J\alpha^2}$ and $\delta=\frac{\epsilon}{4\alpha}$. Let $x\in [0,1]^J$ be the {\sc CCLS} solution of $(\phi,\oep,\delta)$.
We $\delta$-round the point $x$ to the boundaries $\{0,1\}$ as follows.
For every $j\in J$ we set 
\begin{align*}
\ox_j=
\begin{cases}
0 & \text{if } x_j \in [0,\delta) \\
x_j & \text{if } x_j \in [\delta,1-\delta] \\
1 & \text{if } x_j \in (1-\delta,1]
\end{cases}
\end{align*}
We argue that $\ox$ is a solution for {\sc GD-FixedPoint}. For every dimension $j\in J$ we consider three cases.

\paragraph{Case 1: $x_j \in [0,\delta)$.} In such a case $\ox_j=0$ and hence we only need to prove that $\nabla_j \phi(\ox)\leq \epsilon$. Since $x$ and $\ox$ are $(\oep J)$-close, by the $\alpha$-Lipschitzness of the gradient we have $|\nabla_j (x)-\nabla_j(\ox)|\leq \oep J \alpha \leq \frac{\epsilon}{2}$. Therefore, it is sufficient to prove that $\nabla_j \phi(x)\leq \frac{\epsilon}{2}$. 

We know that $\phi((1-\delta) x_j+\delta,x_{-j})-\phi(x)\leq \oep$. By the mean value Theorem there exists $0\leq \epsilon' \leq \delta (1-x_j)$ such that $$\nabla_j \phi (x+\epsilon' e_j)= \frac{\phi((1-\delta) x_j+\delta,x_{-j})-\phi(x)}{\delta (1-x_j)}\leq \frac{\oep}{\delta (1-x_j)}.$$
By the Lipschitzness of the gradient we deduce that $\nabla_j \phi (x) \leq \frac{\oep}{\delta (1-x_j)} +\delta x_j \alpha \leq \frac{\oep}{\delta^2 } +\delta \alpha \leq \frac{\epsilon}{4}+\frac{\epsilon}{4}$. 

\paragraph{Case 2: $x_j \in (1-\delta,1]$.} We apply similar arguments to those of Case 1, but this time we consider the inequality $\phi(x)-\phi((1-\delta) x_j,x_{-j})\geq \oep$.

\paragraph{Case 3: $x_j \in [\delta,1-\delta]$.} We consider both inequalities $\phi((1-\delta) x_j+\delta,x_{-j})-\phi(x)\leq \oep$ and $\phi(x)-\phi((1-\delta) x_j,x_{-j})\geq \oep$ and apply the arguments of Cases 1 and 2.
\end{proof}

The {\sc KKT} problem that has been defined in \cite{CLS} is inspired by the following (multidimensional) formulation of Taylor's Theorem:

\begin{tay-theorem}[See \cite{CLS} Lemma 3.1]
For every function $\phi:\mathbb{R}^J \rightarrow \mathbb{R}$ with $\alpha$-Lipschitz gradients $\nabla \phi$ in the $\ell_2-\ell_2$ sense and for every $x_0,x\in \mathbb{R}^J$ we have $$|\phi(x)-\phi(x_0)-\nabla \phi (x_0) \cdot (x-x_0)|\leq \frac{\alpha}{2}||x-x_0||_2^2.$$ 
\end{tay-theorem}

The problem $\textsc{KKT}$ is defined as follows:

\begin{definition}[$\textsc{KKT}$ \cite{CLS}]\label{def:KKT}
{\color{white} none}

\textbf{Input:} A potential function $\phi:[0,1]^J \rightarrow [0,1]$ that is given by 
coefficients of monomials in $|J|$ variables\footnote{\label{foot:KKT} We follow the convention of \cite{CLS} and restrict the $\textsc{KKT}$ problem to explicit polynomials. To avoid confusions, the Karush–Kuhn–Tucker (KKT) conditions are valid for arbitrary potential functions.}, $\oep$ and $\kappa$.

\textbf{Output:} A point $x\in [0,1]^J$ such that $\phi(x)\geq \phi(y)-\frac{\alpha}{2}\oep^2-\kappa$ for every $y\in B(x,\oep)\cap [0,1]^J$, when $B(x,\oep)$ is the $\ell_2$-ball of radius $\oep$ around $x$.

\end{definition}
\vspace{2mm}

The following lemma shows the equivalence of the {\sc KKT} problem and the defined above {\sc GD-FixedPoint-Expl} variant of gradient-dynamic fixed point problem.

\begin{lemma}\label{lem:kkt=expl}
{\sc KKT} is computationally equivalent to {\sc Expl-GD-FixedPoint}.
\end{lemma}

\begin{proof}
Given an instance of {\sc KKT} $(\phi,\oep,\kappa)$ we set $\epsilon = \frac{\kappa}{J\oep}$ and we apply the {\sc GD-FixedPoint-Expl} algorithm for $(\phi,\epsilon)$. By Taylor's Theorem for every $y\in B(x,\oep)$ we have $$\phi(x) \geq \phi(y) - \nabla \phi \cdot (y-x) - \frac{\alpha}{2} \oep^2 \geq \phi(y) - \frac{\alpha}{2} \oep^2 - \epsilon J \oep \geq  \phi(y)-\frac{\alpha}{2}\oep^2-\kappa.$$

Conversely, given an instance $(\phi,\epsilon)$ of {\sc GD-FixedPoint-Expl} we set $\oep=\frac{\epsilon}{6J\alpha}$ and  $\kappa=\frac{\epsilon^2}{24J\alpha}$. Let $x\in [0,1]^J$ be the {\sc KKT} solution of $(\phi,\oep,\kappa)$.
We $\oep$-round the point $x$ to the boundaries $\{0,1\}$ as follows.
For every $j\in J$ we set 
\begin{align*}
\ox_j=
\begin{cases}
0 & \text{if } x_j \in [0,\oep) \\
x_j & \text{if } x_j \in [\oep,1-\oep] \\
1 & \text{if } x_j \in (1-\oep,1]
\end{cases}
\end{align*}
We argue that $\ox$ is a solution for {\sc GD-FixedPoint}. For every dimension $j\in J$ we consider three cases.

\paragraph{Case 1: $x_j \in [0,\oep)$.} In such a case $\ox_j=0$ and hence we only need to prove that $\nabla_j \phi(\ox)\leq \epsilon$. Since $x$ and $\ox$ are $(\oep J)$-close, by the $\alpha$-Lipschitzness of the gradient we have $|\nabla_j (x)-\nabla_j(\ox)|\leq \oep J \alpha \leq \frac{\epsilon}{2}$. Therefore, it is sufficient to prove that $\nabla_j \phi(x)\leq \frac{\epsilon}{2}$. 

Consider a point $y=x + \oep e_j$ which necessarily belongs to the hypercube. By the {\sc KKT} condition we know that $\phi(y)-\phi(x) \leq \frac{\alpha}{2}\oep^2 + \kappa$ and hence 
$\frac{\phi(y)-\phi(x)}{\oep} \leq \frac{\alpha}{2}\oep + \frac{\kappa}{\oep}$. By the mean value Theorem there exists $0 \leq \epsilon' \leq \oep$ such that $\nabla_j \phi(x+\epsilon' e_j) \leq \frac{\alpha}{2}\oep + \frac{\kappa}{\oep}$. By the Lipschitzness of the gradient we get $\nabla_j \phi(x) \leq \frac{3\alpha}{2}\oep + \frac{\kappa}{\oep} \leq \frac{\epsilon}{4} + \frac{\epsilon}{4}$. 

\paragraph{Case 2: $x_j \in (1-\oep,1]$.} We consider the point $y=x-\oep e_j$ and apply the symmetric arguments to those of Case 1.

\paragraph{Case 3: $x_j \in [\oep,1-\oep]$.} We consider both points $y=x\pm \oep e_j$ and apply the arguments of Cases 1 and 2.

\end{proof}

\subsection{Formal statement of our results}

Our main result is as follows:

\begin{theorem}\label{theo:main}
The problems {\sc Congestion, 5-Polytensor-IdenticalInterest, Con-GD-FixedPoint, KKT}, and {\sc Deg-5-GD-FixedPoint} are $\PPAD \cap \PLS$-complete.
\end{theorem}

As a corollary, we obtain the following equivalence of complexity classes.
\begin{corollary} $\CCLS = \PPAD \cap \PLS$.
\end{corollary}

\subsection{Structure of the remainder of the paper}
Our main technical contribution (Sections~\ref{sec:idea},\ref{sec:game},\ref{sec:analysis}) is the reduction 
\begin{gather}\label{eq:main-reduction}
\textsc{2D-GD-FixedPoint} \le \textsc{5-Polytensor-IdenticalInterest}.
\end{gather}

In Section \ref{sec:idea} we provide intuition and ideas of the main reduction. 
Section \ref{sec:game} describes the reduction (i.e., the construction of the potential game) given a $\textsc{2D-GD-FixedPoint}$ instance. 
Section \ref{sec:analysis} includes the equilibrium analysis which shows that every equilibrium corresponds to a $\textsc{2D-GD-FixedPoint}$ solution.

\eqref{eq:main-reduction} implies that \textsc{5-Polytensor-IdenticalInterest} and {\sc Congestion} are $\PPAD \cap \PLS$-complete.
The only non-obvious aspect of extending this hardness to variants of {\sc GD-FixedPoint}, is that \textsc{5-Polytensor-IdenticalInterest} naturally corresponds to a potential function over a product of simplices (where each simplex corresponds to the feasible mixed strategies of each player), whereas $\textsc{GD-FixedPoint}$ is defined with hypercube domain. In Section~\ref{sec:domain} we show how to handle non-hypercube domains.

\section{Main result: an informal technical overview}\label{sec:idea}
In this section we provide a more detailed overview of our reduction, complementing the highlights described in Section~\ref{sec:highlights}.

\subsection{The utility}

We find it useful for the reader to jump directly to the formula of the game's utility, explain the notations in the utility, and explain which role each utility term plays in the reduction. Formal definitions of these notations appear in Section \ref{sec:game}.
%\newtext{
We describe the game in the form of a sum of \emph{all} local interactions, which is equivalent (with respect to Nash equilibria analysis) to a polytensor identical-interest game; see Remark \ref{rem:identical-interest}.
%} 

\begin{align*}
\begin{aligned}
u=&-\sum_{p\in P,j\in J} (1+\1_{g=p}) [\ox^p_j-\oy_{j}-i(p,j)\epsilon_S]^2 &&\text{(imitation)}\\
&-\epsilon_C \sum_{p\in P,l\in \cL} (\epsilon_T+\1_{[g=p]}) 2^{-2d_l} w^p_{l} &&\text{(circuit)}\\
&+\epsilon_P \sum_{p\in P} \1_{[g=p]} (\phi^p-c^p) &&\text{(potential)}\\
&+\epsilon_G \sum_{p\in P,j\in J} \oy_{j} \1_{[g=p]} \triangle^p_j  &&\text{(gradient)}\\
&-\epsilon_T \sum_{p\in P, j\in J,k\in [K]} [\tx^p_{j,k} -\oy_{j}-i(p,j)\epsilon_S]^2 &&(x\text{-bit imitation)}\\
&-\epsilon_T \sum_{p\in P, j\in J,k\in [K]} (1+\1_{g=p})[\ox^p_{j} -\ty_{j,k}-i(p,j)\epsilon_S]^2 &&(y\text{-bit imitation)}\\
&-\epsilon_T\epsilon_G \sum_{p\in P, j\in J,k\in [K]} \ty_{j,k} \1_{[g=p]} \triangle^p_j  &&\text{(bit gradient)}\\
&+\epsilon_T \sum_{r\in R} m^x_{r} +\epsilon_T \sum_{j\in J} m^y_{j} &&\text{(veto)}.
\end{aligned}
\end{align*}

The parameters that appear in the utility satisfy $\epsilon_C \gg \epsilon_P \gg \epsilon_G \gg \epsilon_T$.

\paragraph{Imitation utility} The index $p\in P$ denotes a point in the sampling. The index $j\in J$ denotes a coordinate. The action of the guide player is denoted by $g$. The term $\ox^p_j$ denotes the $j$-th coordinate of the point $p$ (recall that this coordinate, as any other coordinate in the game, is chosen by $K$ bit players and one veto player). The term $\oy_j$ denotes the $j$-th coordinate of the point $p$. The term $i(p,j)\epsilon_S$ is the sampling shift of $y$ which the $x$ tries to imitate.

The quadratic imitation loss plays a central role in our reduction. This is a convenient imitation loss to work with due to two reasons. First, in a mixed Nash equilibrium analysis, even if both the $x^p_j$ team and the $y_j$ team are playing mixed strategies, we have a clean formula for the expected utility (see Fact~\ref{fact:var}) from which it is clear that the $x$ team tries to imitate the expectation of $\oy_j$. Second, the fact that it is a finite degree polynomial allows us to argue that this term is in fact $c$-uniform polymatrix game (See Proposition \ref{pro:hypergraphical}). Note that the imitation is the most significant term in the utility and it remains most significant even if the guide player does not choose the point $p$ (i.e., even if $\1_{[g=p]}=0$). Having said that, note that in case the imitation is almost perfect, for instance if $\ox^p_j$ is pure and is located $\epsilon \approx 2^{-K}$ close to its target $\Ex [\oy_j+i(p,j)\epsilon_S]$ the gain from making the imitation perfect improves the imitation loss only by $O(2^{-2K})$ which is negligible relative to the circuit, the potential, and the gradient utility terms.  

\paragraph{Circuit utility} 
The index $l\in \cL$ refers to a line (wire) in a circuit. Note that every point has its own circuit to compute $\phi$ and $(\triangle_j)_{j\in J}$. The term $d_l$ is the depth of the gadget. The term $w^p_l$ denotes the indicator of whether the binary action played by the player that is located at $l$-th line in the $p$-th circuit is \emph{wrong}; namely, it is not the correct output of the corresponding gate given the input players' actions.

By the exponential decay of the weights of wrong computations, the incentive of a gadget to match her predecessors is $4$ times larger than her gain from correct computations of her successors. This creates an incentive for all gadget players to perform correct computations in case the input to the circuit is pure.\footnote{It remains unclear what the circuit players will do in case where the input bits are mixed.} 

Note that in case the guide player chooses a point, its weight for correct computations is relatively high ($\epsilon_C$); whereas in case the guide player chooses a point with probability 0, the incentive for correct computation is tiny ($\epsilon_T$). This allows points that are located close to the boundary and are not chosen by the guide player to pass though the boundary because for such points imitation is more significant. The tiny incentive $\epsilon_T$ for correct computations appears there to reach updating of the circuit to the new input after $x^p_j$ has passed the boundary. 

\paragraph{Potential utility} 
The term $\phi^p$ denotes the calculated potential by the circuit of the point $p$. Note that only potentials that are chosen by the guide player with positive probability appear in the utility. We prove (see Lemma \ref{lem:pure-L-input}) that the guide player does not choose the problematic points with mixed input that are located on the $2^{-N_{in}}$ grid. Therefore, the computed potential in the utility is in fact meaningful correct potential rather than some unexpected potential that is a result of mixture of the predecessors.

\paragraph{Gradient utility}
The term $\triangle^p$ denotes the gradient\footnote{It is actually more convenient to calculate $\triangle_j (x)=\phi(x+2^{-N_{in}})-\phi(x)$ rather than the gradient whose approximation is 
$\frac{\phi(x+2^{-N_{in}})-\phi(x)}{2^{-N_{in}}}$. Namely, $\triangle_j \approx 2^{-N_{in}} \nabla_j$.} estimated by the circuit of the point $p$. Similarly to the potential utility, we note that only the values of circuits that are chosen by the guide player with positive probability are counted. The factor of $\oy_j$ incentivizes the $y$-team to increase (decrease) $\oy_j$ if $\triangle_j>0$ ($\triangle_j<0$).

\paragraph{$x$-bit imitation utility}
The index $k\in [K]$ refers to the binary bits that create the real numbers $\ox^p_j$. The term $\tx^p_{j,k}$ denotes the number $\ox^p_j$ when we replace the $k$-th bit to be the one that the $k$-th player choose (even if she was vetoed). The role of this term is to provide an incentive for vetoed bits to update their action to the one that veto player choose for them (recall that when a bit player is vetoed, her decision does not affect $\ox^p_j$).

\paragraph{$y$-bit imitation utility}
Similar to $x$-bit imitation.

\paragraph{Bit gradient utility}
This is a technical term that allows us to apply on the $y$-team similar arguments to those we apply for the $x^p$ teams. Note that the weight of this term is extremely low: $\epsilon_T \epsilon_G$.

\paragraph{Veto utility}
The index $r\in R$ refers to real numbers chosen by the $x$-teams (the real numbers that generate the points $x^p$).
The integer number $m^x_r \in [K+1]$ is chosen by the veto player of the $x_r$ team and it denotes the bit starting from which the veto player implies her veto ($m^x_r=K+1$ means no veto). This utility incentivizes the veto player to impose veto on less bits. This utility term is needed in order that after the bit players match their bits with the veto player, the veto on them would be canceled.

\subsection{A better-reply path}\label{sec:br-path} 
Another angle that provides intuition about the reduction is a description of a better-reply path that converges to a gradient decent fixed point in our game. 
Roughly speaking the $x$ and the $y$ teams almost continuously move  in the direction of the gradient until they reach a gradient decent fixed point. (The move is continuous up to a very small step size $\approx 2^{-K}$.) Below we describe the steps of this process. For simplicity we focus on a single coordinate; for multiple coordinates these steps simply happen simultaneously.  

\begin{enumerate}
\addtocounter{enumi}{-1}

\item In this discussion we consider the case where the initial state of the best-reply dynamic is already quite arranged\footnote{The assumption that the initial state of the dynamic is arranged substantially simplifies the arguments. For an arbitrary mixed action profile it is more involved to identify the player that has an incentive to deviate. These cases are treated in the formal proof; see the discussion at the beginning of Section \ref{sec:analysis}.}: the $x$-team and the $y$-team players are playing pure; all the $x$-teams perfectly imitate the corresponding $y$ shifts; no veto is imposed by the veto player at any team; all the $x$-teams and the $y$ team are located far from the $2^{-N_{in}}$-grid and we denote $x^p\in [a2^{-N_{in}},(a+1)2^{-N_{in}})$ for some integer $a\in [2^{N_{in}}]$;  and finally we assume that in this region $\nabla \phi(a2^{-N_{in}})>0$. 

\item Any circuit player that is located on the input gadgets observes a bit player who plays pure and whom she tries to imitate. Therefore, this input circuit player strictly prefers to match her bit, so the inputs to the circuit is $a2^{-N_{in}}$ for every circuit. We proceed inductively with the depth of the circuit to deduce that all the circuits correctly compute $\phi(a2^{-N_{in}})$ and correctly compute $\nabla (a2^{-N_{in}})$.

\item The guide player sees identical samples in terms of the imitation loss and in terms of the computed potential and might play mix at this point.

\item The $y$-team will be better off by deviating to $y+2^{-K}$: this causes an imitation loss of $|P|2^{-2K}$ ($2^{-2K}$ for each one of the $|P|$ $x$-teams) but increases the gradient utility by $\epsilon_G2^{-K} \nabla (a2^{-N_{in}})\gg |P| 2^{-2K}$. 
If the least significant bit of $y$ is $0$ it is simply flipped to $1$; otherwise consider the suffix $01...1$ beginning with the least significant $0$: the $y$-vetor player prefers to impose a $10...0$ veto over the existing $01...1$.

\item After the bit players corresponding to the suffix of $y$ are vetoed, their action affects only the $y$-bit imitation term. In this term their bit is aggregated with the veto player's assignments for the other bits. Therefore, every vetoed bit player prefers to match the action of the veto player. 

\item After all vetoed bit players have updated their bits to $100...0$ the veto player prefers to cancel her veto on these players because of the veto utility.

\item After the $y$ team has moved by $2^{-K}$ the $x$-teams prefer to move by $2^{-K}$ too in order to perfectly imitate their shifts of the $y$-team. Note that such a move does not affect the first $N_{in}$ bits and hence does not affect the circuit utility which continues to be perfect. Namely, each $x$-team separately updates the actions according to Steps (3),(4), and (5).

\item The $x$ and $y$ team proceed moving in the positive direction by repeatedly applying Steps (3)-(6) until one of these teams reaches the boundary $(a+1)2^{-N_{in}}$. 

\item Let $x_{-1},x_0,x_1$ denote three $x$-samples, trying to imitate $y$ with negative, zero, and positive shift, respectively. The first team that will reach the boundary is the $x_{1}$ team. At this point it will stop increasing $x_{1}$ because such a change will cause wrong computations of the circuit. Meanwhile teams $x_{0}$ and $x_{-1}$ proceed updating their actions and proceed imitating their shifts of $y$ perfectly. 

\item Now the guide player is no longer indifferent between all samples and she prefers to avoid choosing $x_{1}$ because this sample has non-perfect imitation while the computed potential at $x_{1}$ is the same as in $x_{0}$ and in $x_{-1}$.

%Recall that by Step (3) the guide player does not choose this team. By definition of the circuit utility, in such a case the weight of this utility is tiny in comparison with the imitation utility. Therefore, this corner team will again apply Step (7) and will pass the $(a+1)2^{-N}$ boundary.

\item Once the guide player does not choose $x_{1}$ the weight of the circuit in this team becomes negligible and then the $x_{1}$-team prefers to increase the value $x_{1}$. In particular the $N_{in}$ most significant bits become $(a+1)2^{-N_{in}}$.

\item The circuit players of the $x_{1}$ team update their actions according to Step (2).

\item The computed potential $\phi((a+1)2^{-N_{in}})$ at $x_{1}$ team is now \emph{higher}  than the computed potential $\phi(a2^{-N_{in}})$ at $x_{0}$ and $x_{-1}$ teams and hence the guide player switches to choose team $x_{1}$.

\item All $x$-teams and the $y$-team proceed to increase their values by $2^{-K}$ by repeatedly applying Steps (3)-(6) until both other teams $x_{0}$ and $x_{-1}$ also pass the boundary from $a2^{-N_{in}}$ to $(a+1) 2^{-N_{in}}$. Now we got back to the initial state at Step (0) but now all the $x$-teams and the $y$-team are located in the next grid cell $[(a+1)2^{-N_{in}},(a+2)2^{-N_{in}})$.

\item Players repeatedly apply Steps (1)-(13) until they reach either the boundary ($0$ or $1$) or they reach a point $x$ with approximately equal potential values at $a2^{-N_{in}}$ and at $(a+1)2^{-N_{in}}$. In the former case the procedure terminates are a boundary point 0 (1) with negative (positive) gradient. In the latter case the procedure terminates at an interior point with approximately zero gradient.
  
\end{enumerate}

\section{Main Result: Constructing a hard game}\label{sec:game}
In this section we present the formal construction for our main result; we analyze the construction in Section~\ref{sec:analysis}.

\subsection{Proof's notations}
Our reduction (the game and its equilibrium analysis) involves quite a few notations. For reader's convenience we summarize all the notations their short verbal description and a reference for their definition in the following table. %\Yakov{Irrelevant notations (hopefully) has been erased from the table.}

\input{t1.tex}

\subsection{Reduction parameters}\label{sec:constants}

The parameters $\epsilon$ and $\alpha$ are inputs to the problem. $|J|=2$ denotes the dimension of potential's domain.

We define four parameters $D,N_{in},N_{out}$ and $K$ of polynomial magnitude. The depth of the modified circuit (which will be described in Section \ref{sec:circuit}) is $D$. The parameter $N_{in}$ is chosen to satisfy $2^{-N_{in}}=(\frac{\epsilon}{6\alpha})^3$.
%\newtext{
The precision of the circuit computation is given by $$N_{out} := \left(N_{in} + D\right) \cdot 2^{\text{$\times$-gate-depth of arithmetic circuit}}.$$ It is chosen so that for any $N_{in}$-bit input we can compute the output exactly. By the well-behaved desideratum (see Section~\ref{sub:opt-def}), 
$2^{\text{$\times$-gate-depth of arithmetic circuit}}$ is polynomial in the input size. 
%}
Finally, the parameter $K$ is chosen to satisfy the exponential Inequality~\eqref{eq:constants} below.

The exponentially small parameters are $\epsilon_R,\lambda,\epsilon_S,\oep_C,\epsilon_C,\epsilon_P,\epsilon_G,$ and $\epsilon_T$. These parameters are chosen to satisfy the following inequalities.  The relation between the parameters in our reduction are summarized below where the sign $f \gg g$ means that $f > g \cdot \polylog(g))$ for some sufficiently large $\polylog$.

%\Aviad{TBD - update this:}
\begin{align}\label{eq:constants}
\begin{aligned}
&\epsilon,\alpha^{-1} \gg 2^{-N_{in}}\gg \epsilon_R \gg 2^{-1.5N_{in}} \gg \lambda \gg \epsilon_S  \gg \oep_C \gg \oep_C^2 \gg \epsilon_C \gg \epsilon_C 2^{-2D}
\gg \epsilon^2_C 2^{-2D} \gg \epsilon_P \gg \epsilon_G \gg \\
&\gg \lambda \epsilon_G \gg 2^{-K} \gg 2^{-2K} \gg \epsilon_T.
\end{aligned} 
\end{align}

\subsection{Collective Choice of a Real Number}\label{sec:real}

We introduce several useful notations. Given a binary string $(x_k)_{k\in [K]}$ we denote by 
\begin{align}\label{eq:x<}
x_{<k} := \sum_{k'<k} 2^{-k'}x_{k'}\in [0,1]
\end{align}
the value of the bits that are more significant than $k$. Similarly, we define $x_{\leq k},x_{\geq k}$ and $x_{>k}$ where the summation is over the bits $\{k'\leq k \}, \{k' \geq k\}$ and $\{k'>k\}$ correspondingly. Same notations are valid for the case where $x\in [0,1]$ is a real value: we simply apply it on the binary representation of $x$.

As will be described in Section \ref{sec:circuit}, the circuit's inputs are numbers in $[0,1]$ in binary representation with precision $2^{-N_{in}}$ (i.e., of length $N_{in}$). The players in our game will (collectively) choose real numbers with much higher precision $2^{-K}$ (i.e., binary strings of length $K$) for $K\gg N_{in}$. The first $N_{in}$ bits of this string (i.e., the most significant bits of this string) will serve as an input to the circuit. 

Each real number in our game is chosen by a team of $K+1$ players: $K$ \emph{bit players} and a single \emph{veto player}. The bit players simply choose binary actions $x_k\in \{0,1\}$. Their profile of choices induces a binary string $(x_k)_{k\in [K]}$. The veto player chooses a coordinate $m\in [K+1]$ and imposes a veto on a suffix $(x_k)_{k\geq m}$. Such a choice replaces the suffix of the binary string $(x_k)_{k\geq m}$ by another string that also is chosen by the veto player. For each $m$, veto player has a constant size menu (specifically 6) of strings for the replacement of the suffix.  

Formally, the veto player chooses a pair $(m,t)$ where $m\in [K+1]$ and $t\in \{\vzero, 00\vone, 0\vone, 1\vzero, 11\vzero,\vone\}$. 
The $K$-bit string that is determined by the collective choices of the bit players and the veto player is denoted by $\ox^b=(\ox^b_k)_{k\in [K]}$ (the superscript $^b$ distinguishes the {\bf b}it string from the associated real number $\ox$).
The $m-1$ initial elements of $\ox^b$ are determined by the bit players; i.e., $\ox^b_k=x_{k}$ for every $1\leq k < m$. The suffix $(\ox^b_{k})_{k\geq m}$ is determined by $t$ and is defined to be the binary string 
$00...0$, $0011...1$, $011...1$, $100...0$, $1100...0$, or $11...1$ according to the type of veto ($\vzero$, $00\vone$, $0\vone$, $1\vzero$, $11\vzero$, or $\vone$ correspondingly).

The binary string $\ox^b$ straightforwardly induces a real value $\ox\in [0,1]$, by $\ox= (\ox^b)_{\geq 1}$.
Note that
\begin{align}\label{eq:veto-value}
\begin{aligned}
\ox:=x_{<m} + v(m,t) \ \text{ where } \ v(m,t):=\begin{cases}
0 &\text{if } t=\vzero \\
2^{-m-1}-2^{-K} &\text{if } t=00\vone \\
2^{-m}-2^{-K} &\text{if } t=0\vone \\
2^{-m} &\text{if } t=1\vzero \\
2^{-m}+2^{-m-1} &\text{if } t=11\vzero \\
2^{-m+1}-2^{-K} &\text{if } t=\vone.
\end{cases}
\end{aligned}
\end{align}
The real value that ignores veto's choice only in the $k$-th bit and instead follows the choice of the bit player is denoted by
\begin{align}\label{eq:tx}
\tx_{k}:=(\ox^b)_{<k}+2^{-k}x_{k}+(\ox^b)_{>k}.
\end{align}

\subsection{Sampling}\label{sec:sampling}
%\Yakov{New Section.}
We recall that $J=\{1,2\}$ are the dimensions of the potential, and we denote $I=\{-1,0,1\}$. The point $y$ is simply chosen by $2$ teams of players: $y_j$-team for $j=1,2$, where each team chooses a real number as in Section~\ref{sec:real}. For the $x$ group the players choose $|I|^2=9$ points as follows.
We have $|I||J|=6$ teams of players denoted by $x_{i,j}$ for $i\in I,j\in J$ each team chooses a real number. These teams generate the $9$ points $x^p=x^{i_1,i_2}:=(x_{i_1,1},x_{i_2,2})$ for $i_1,i_2\in I$. We denote by $R=I\times J$ the set of indexes for the $x$-teams. A generic index of a team is denoted by $r\in R$. We denote by $P=I^2$ the set of indexes for the $x$-points, which will be called \emph{$x$-samples}. A generic index of a sample is denoted by $p\in P$. For $r\in R$ we denote by $P(r)\subset P$ the samples in which the $x_r$-team participates. Namely, for $r=(i,1)$ we have $P(r)=\{(i,-1),(i,0),(i,1)\}$ and for $r=(i,2)$ we have $P(r)=\{(-1,i),(0,i),(1,i)\}$. 

Samples' incentives will be designed to imitate different shifts of $y$. Concretely, the $x_{i,j}$ team tries to imitate $y_j+i\epsilon_S$. The location of the points in case of perfect imitation of all teams is demonstrated in Figure \ref{fig:sampling}.

\begin{figure}
\begin{center}
\begin{tikzpicture}
\draw[->] (0,0)--(10,0);
\draw[->] (0,0)--(0,7);
\foreach \x in {-1,0,1}
	\foreach \y in {-1,0,1}
		{\filldraw (5+2*\x,4+2*\y) circle(0.1);
		\node[right] at (5+2*\x,4+2*\y) {$x^{\x,\y}$}; 
		}
\foreach \x in {-1,0,1}
	{\draw (5+2*\x,-0.1)--(5+2*\x,0.1);
	\node[below] at (5+2*\x,-0.1) {$x_{\x,1}$};
	}	
\foreach \y in {-1,0,1}
		{\draw (-0.1,4+2*\y)--(0.1,4+2*\y);
		\node[left] at (-0.1,4+2*\y) {$x_{\y,2}$}; 
		}
\node[right] at (5.6,3.85) {$=y$};
\node[below] at (5.6,-0.1) {$=y_1$};
\node[left] at (-0.8,4) {$y_2=$};
%\node[right] at (5,4) {$x^{0,0}=y$};

\end{tikzpicture}
\end{center}
\caption{The figure demonstrates the samples in case where all the $x$-samples perfectly imitate their target of the shifted $y$ point.}\label{fig:sampling}
\end{figure}
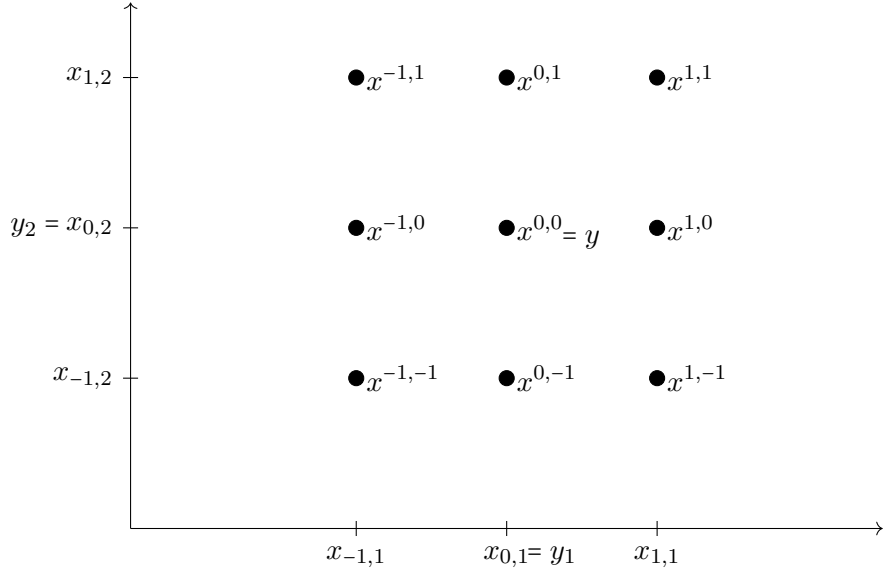

\subsection{The circuit}\label{sec:circuit} 

We recall that the input to the problem is $(\phi,\epsilon)$ where the potential $\phi$ is given by an \emph{arithmetic} circuit whose Lipschitz constant for the gradient (in the $\ell_2-\ell_2$ norms) is $\alpha$ and $
\alpha$ might be at most exponentially large in the input size. We set $N_{in}$ to satisfy  %\Yakov{TO DO. Check this inequality again.}
\begin{align}\label{eq:N and ep-r}
2^{-N_{in}}=(\frac{\epsilon}{6\alpha})^3 \text{ and we set } \epsilon_R=(\frac{\epsilon}{6\alpha})^4=2^{-\frac{4}{3}N_{in}}. 
\end{align}

%We define $\phi':[0,1-2^{-N}]^J\ra [0,1]$ by $\phi'(x):=\phi(\frac{1}{1-2^{-N}}x)$.
We translate $\phi$ to a \emph{boolean} circuit $C'$ where the inputs  of $C'$ in each coordinate are given by binary stings of length $N_{in}$. 
Note that the input in each coordinate is some $x_j \in \{0,2^{-N_{in}},\dots, 1-2^{-N_{in}})$. %In particular, the boundary $x_j=1-2^{-N}$ of the boolean circuit $C'$ corresponds to the boundary $x_j=1$ of $\phi$. 
The circuit computes $\phi(x)$ with precision $2^{-N_{out}}$; in particular, the output is a binary string of length $N_{out}$. We denote  by $[0,1)^2_{2^{-N_{in}}}$ the domain of the circuit $C'$, which is the $2$-dimensional $2^{-N_{in}}$-grid. For every sample $x$ on the grid we denote by 
\begin{align}\label{eq:triangle}
\triangle_j (x):=\phi(x+ 2^{-N_{in}}e_j)-\phi(x)
%\triangle_j (x):=\begin{cases}
%\phi(x+ 2^{-N_{in}}e_j)-\phi(x) \text{ if } x_j \neq 1-2^{-N_{in}} \\
%\phi(x)-\phi'(x- 2^{-N_{in}}e_j) \text{ if } x_j = 1-2^{-N_{in}}
%\end{cases}
\end{align}
the difference in $\phi$ between two adjacent grid points in the direction of $j$. 
Our reduction finds a point $x$ on the grid such that
\begin{align}\label{eq:delta-reduction}
\begin{aligned}
&\triangle_j (x) \leq \epsilon_R \text{ for every } j\in J \text{ such that } x_j\neq 1-2^{-N_{in}} \\
&\triangle_j (x) \geq -\epsilon_R \text{ for every } j\in J \text{ such that } x_j\neq 0
\end{aligned}
\end{align}
The following lemma shows that this property is sufficient for the reduction.
\begin{lemma}\label{lem:sufficient}
If $x\in [0,1)^J_{2^{-N_{in}}}$ satisfies the property in Equation \eqref{eq:delta-reduction} then the point $\frac{1}{1-2^{-N_{in}}}x$ is a solution for $\textsc{2D-GD-FixedPoint}(\phi,\epsilon)$.
\end{lemma}

\begin{proof}
By the Mean Value Theorem  in the $j$-th dimension, we derive that for some $x'=x + \delta e_j$ for $0\leq \delta\leq 2^{-N_{in}}$ we have 
$2^{N_{in}} \triangle_j (x) = \nabla_j \phi(x') $. Note that $x'$ is $2^{-N_{in}}$-close to $x$ which it turn is $2\cdot 2^{-N_{in}}$-close to $\frac{1}{1-2^{-N_{in}}}x$.
By the $\alpha$-Lipschitzness of the gradient we deduce that $$|\nabla_j \phi(\frac{1}{1-2^{-N_{in}}}x) - 2^{N_{in}} \triangle_j (x)| \leq 3\cdot 2^{-N_{in}}\alpha.$$

If $\frac{1}{1-2^{-N_{in}}}x_j \neq 1$ which is equivalent to $x_j\neq 1-2^{-N_{in}}$ from the property in Equation \eqref{eq:delta-reduction} we deduce that $\nabla_j \phi(\frac{1}{1-2^{-N_{in}}}x)\leq 2^{N_{in}} \epsilon_R + 3\cdot 2^{-N_{in}}\alpha \leq \frac{\epsilon}{2}+\frac{\epsilon}{2}$.

If $\frac{1}{1-2^{-N_{in}}}x_j \neq 0$ which is equivalent to $x_j\neq 0$ from the property in Equation \eqref{eq:delta-reduction} we deduce that $\nabla_j \phi(\frac{1}{1-2^{-N_{in}}}x)\geq -2^{N_{in}}\epsilon_R - 3\cdot 2^{-N_{in}}\alpha \geq -\frac{\epsilon}{2}-\frac{\epsilon}{2}$.
\end{proof}

We denote by 
\begin{align}\label{eq:lambda}
\lambda:=\alpha 2^{-2N_{in}},
\end{align}
an upper bound on the scaled Lipschitz constant of the $\Delta_j$'s when we move to a $2^{-N_{in}}$-close point.
Note that  $\lambda<2^{-\frac{5}{3}N_{in}}\ll 2^{-1.5N_{in}} \ll \epsilon_R$. 

\subsubsection{Modifying the circuit's computations}\label{sec:mod-cir}
We are given a boolean circuit $C'$ with boolean gates that computes a potential function $\phi'(x)\in [0,1]$ for every $x\in[0,1]^J_{2^{-N_{in}}}$. From $C'$ we can construct (in polynomial time) a circuit $C$ that for every $x\in[0,1)^J_{2^{-N_{in}}}$ outputs $\phi(x)$ as well as a pair $(\triangle_j(x))_{j\in J}$.

We denote by $\cL$ the set of lines (wires) in the circuit $C$. The set of input lines is denoted by $\cL_{in}\subset \cL$. 
The set of output lines is denoted by $\cL_{out}\subset \cL$.
We denote also $\cL_{out}=\{(\phi_{k})_{k\in [N_{out}]} \cup \{ (\triangle_{j,k})_{j\in J,k\in [N_{out}]}\}$ to distinguish the $\phi$ output from the $\triangle$ outputs. The depth of a wire in the circuit is denoted by $d_l$. The depth of the circuit is denoted by  $D\geq d_l$ for every $l\in \cL$.
Without loss of generality we may assume that every gate in $C$ has a fan-in and fan-out 2 (otherwise one can replace larger fan-out by a binary tree structure with equality gates).  

\subsubsection{Replicating the circuit}\label{sec:rep}
Our reduction includes $|I|^2=9$ samples $(x^p)_{p\in P}$. For each sample we have a separate circuit $C^p$; whose lines are denoted by $l^p$ for $l \in \cL$. %When $p$ is clear from context we simply use $C,\cL,l$. 
Each such circuit computes $(\phi(x^p),(\triangle_j(x^p))_{j\in J})$. 

The inputs to these circuits come from the first $N_{in}$ bits of the corresponding real numbers $\ox_r$; see Section \ref{sec:real} for the definition of $\ox$; see Section \ref{sec:sampling} for the correspondence of between samples $x^p$ and the real numbers $x_r$. 
We construct a player for every wire $\l^p$ of every circuit; we will abuse notation and also denote the player by $l^p$. 
We say that the computation of an input player $l^p\in \cL_{in}$ is \emph{correct} if its binary action matches the bit $\ox^b_r$ of the corresponding team that chooses $x_r$.
We say that the computation of a non-input player $l^p$ for $l\notin \cL_{in}$ is \emph{correct} if its binary action matches the boolean operation on the binary actions of the inputs to the gate.
We denote by $w^p_l$ the indicator of a {\em wrong} computation.

The actions of the output players $(\phi^p_{k})_{k\in [N_{out}]}$ define a value for the potential 
$$\phi^p:=\sum_{k\in [N_{out}]}2^{-k}\phi^p_{k}\in [0,1].$$ Similarly, we define the value of $\triangle^p_j$ for every $j\in J$ with respect to the corresponding output bits $(\triangle^p_{j,k})_{k\in [N_{out}]}$.

\subsection{Players and Actions}\label{sec:players-actions}
We summarize here all the players and their actions that we have described. 
For clarity of notations we typically identify players' names with the action they are taking.

We have $|R|=6$ teams of real-number players $(x_r)_{r\in R}$. Each team chooses a real number. Each team consists of $K$ bit players and a single veto player. The bit players are denoted by $x_{r,k}$ and their actions are $x_{r,k}\in \{0,1\}$. The $x_r$-veto player has $6(K+1)$ actions that are all the pairs $(m^x_{r},t^x_r)$ where $m^x_{r}\in [K+1]$ and $t^x_r \in \{\vzero, 00\vone, 0\vone, 1\vzero, 11\vzero,\vone\}$.

Similarly, we have $|J|=2$ teams of players $(y_j)_{j\in J}$. Each team chooses a real number. Each team consists of $K$ bit players and a single veto player. The bit players are denoted by $y_{j,k}$ and their actions are $y_{j,k}\in \{0,1\}$. The $y_j$-veto player has $6(K+1)$ actions that are denoted by $(m^y_{j},t^y_j)$.

Each sample $x^p$ (which is generated by the real numbers $(x_r)$) has its own circuit players that are denoted by $(l^p)_{l\in \cL}$. Each circuit player has binary actions.

Finally, we have a single \emph{guide} player who chooses a sample in $P$ and her action is denoted by $g\in P$.

\subsection{Utility}\label{sec:utility-formal}

Given a point $p=(i_1,i_2)$ we denote $i(p,j)=i_j$. We recall that the imitation target of $x^p_j=(x^{i_1,i_2})_j$ is $y_j+i_j\epsilon_S$. 
We also recall that this coordinate is chosen by the $x_{i_j,j}$-team and the real value (after taking veto's action into account) is denoted by $\ox^p_j$.

The common utility of all players is given in the following equation where the names of each utility term is indicated on the right hand side.
\begin{align}\label{eq:u}
\begin{aligned}
u=&-\sum_{p\in P,j\in J} (1+\1_{g=p}) [\ox^p_j-\oy_{j}-i(p,j)\epsilon_S]^2 &&\text{(imitation)}\\
&-\epsilon_C \sum_{p\in P,l\in \cL} (\epsilon_T+\1_{[g=p]}) 2^{-2d_l} w^p_{l} &&\text{(circuit)}\\
&+\epsilon_P \sum_{p\in P} \1_{[g=p]} \phi^p &&\text{(potential)}\\
&+\epsilon_G \sum_{p\in P,j\in J} \oy_{j} \1_{[g=p]} \triangle^p_j  &&\text{(gradient)}\\
&-\epsilon_T \sum_{p\in P, j\in J,k\in [K]} [\tx^p_{j,k} -\oy_{j}-i(p,j)\epsilon_S]^2 &&(x\text{-bit imitation)}\\
&-\epsilon_T \sum_{p\in P, j\in J,k\in [K]} (1+\1_{g=p})[\ox^p_{j} -\ty_{j,k}-i(p,j)\epsilon_S]^2 &&(y\text{-bit imitation)}\\
&-\epsilon_T\epsilon_G \sum_{p\in P, j\in J,k\in [K]} \ty_{j,k} \1_{[g=p]} \triangle^p_j  &&\text{(bit gradient)}\\
&+\epsilon_T \sum_{r\in R} m^x_{r} +\epsilon_T \sum_{j\in J} m^y_{j} &&\text{(veto)}.
\end{aligned}
\end{align}

%\newtext{
The utility is presented as the sum of \emph{all} local interactions; see Remark \ref{rem:identical-interest}. 

\begin{proposition}\label{pro:hypergraphical}
The identical interest game $u$ is strategically equivalent to a $5$-polytensor identical interest game.
\end{proposition}
%}

\begin{proof}
We write each of the terms $\ox^p_{j},\oy_{j},\tx^p_{j,k},\ty_{j,k}$ as a sum of $K$ terms $(\ox^b_{r,k})$ or $(\oy^b_{j,k})$, and each of these terms depends on actions of only two players (veto player and bit player).
We write each of the terms $\phi^p,\triangle^p_{j}$ as the sum of $N_{out}$ terms $\phi^p_k,\triangle^p_{j,k}$, and each of these terms is determined by the action of a single player. Then we open the brackets of $[\cdot ]^2$ and of all multiplication operations. We obtain a representation of the utility as a sum of polynomial number ($O(2|P| K^2)$) of terms. Each term is determined by the actions of at most five players.\footnote{The maximal dependence on the actions of five players is obtained, for instance, in terms of the form $\1_{[g=p]}\ox^b_{r,k}\oy^b_{j,k}$ which depends on the actions of two veto players, two bit players, and the guide player.}
\end{proof}

\section{Main Result: Equilibrium Analysis}\label{sec:analysis}

The informal discussion of our construction mostly relied on the underlying assumption that the behavior of players in the game is pure. For instance, we have used terms as ``$x_{i,j}$ is located somewhere''. However, in a mixed Nash equilibrium analysis all bit players as well as the veto player might play mixed strategies which creates a \emph{distribution} over the reals. In order that the above intuition will be applicable we need to prove that $x_{i,j}$ is either pure or that the distribution of its realizations is extremely concentrated.

The intuition for the validity of the above pure equilibrium analysis is the following. We first use the convenient structure of the square loss imitation utility to show that indeed ``many'' $x_{i,j}$ teams have concentrated values. The pureness of these teams provokes pureness on other players too. For example, if all inputs to the circuit are pure, then all gadget players in that circuit are playing pure. 
%The ``pureness of the system'' that we prove turns out to be sufficient to apply the above intuitions.
With this high level intuition in mind, in this section we present the formal proof, which is exceptionally delicate. 
At each step of the proof, we exhibit more players playing (approximately) pure strategies; along the way we also slowly obtain increasingly precise surrogates (see $\oox_{i,j}$ and $x^c_{i,j}$) for ``the location of $x_{i,j}$''.

\subsection{Rewriting the Expected Utility}\label{sec:rewriting}

The square-loss form of the imitation utility terms allows us to provide a clean formula of the expected common utility (even if all players are playing mixed strategies). We provide this formula in this subsection.

We denote by capital letters the mixed strategies of the players.
Henceforth, we will refer to a mixed Nash equilibrium as a tuple $(X_{r,k},(M^x_{r},T^x_{r}),L^p,G,Y_{j,k},(M^y_{j},T^y_{j}))$ of independent random variables.\footnote{For the actions of the veto players the pair $(M^y_{j},T^y_{j})$ is independent of the other random variables. Obviously $M^y_{j}$ might be correlated with $T^y_{j}$ because those are chosen by the same player. Similarly, for veto players in the $x$-team.}
All real-valued random variables (such as $u$, $\phi^p$, or $\ox^p_{j}$) will be denoted by capital latter ($U$, $\oX^p_{j}$, and $\Phi^p$) For the guide player we denote $g_p=\Prob [G=p]$. We denote by $\Delta_j$ the random variable that is equal $\Delta^p_j$ with probability $g^p$ (we recall that $\Delta^p_j$ is determined by the circuit output players and hence is determined by $(l^p)_{l\in \cL}$). 

\begin{fact}\label{fact:var}
For any three independent random variables $A,X,Y$ we have
\begin{align*}
\Ex [A(X-Y)^2]=\Ex[A]\left[ \left( \Ex[X]-\Ex[Y] \right)^2 + \Var [X] + \Var [Y] \right].
\end{align*}
\end{fact}

\begin{proof}
\begin{align*}
\Ex [A(X-Y)^2] & = \Ex[A] \Ex[(X-Y)^2]=\Ex[A] \left[ \Ex[X^2]-2\Ex[XY]+\Ex[Y^2] \right] \\
& = \Ex[A] \left[ \Var[X] + \Ex[X]^2- 2\Ex[X] \Ex[Y]+\Var[Y] + \Ex[Y]^2 \right]\\
& =\Ex[A]\left[ \left( \Ex[X]-\Ex[Y] \right)^2 + \Var [X] + \Var [Y] \right]
\end{align*}
The first equation follows from the independence of $A$ and $(X,Y)$. The third equation follows from the definition of the variance and by the independence of $X$ and $Y$.
\end{proof}

By applying Fact~\ref{fact:var} on the imitation terms the $x$-bit and $y$-bit imitation terms of Equation~\eqref{eq:u}, and by the fact that for independent random expectation commutes with multiplication, we can rewrite the expected utility in any profile of mixed strategies as:

\begin{align}\label{eq:eu}
\begin{aligned}
\Ex[U]=&-\sum_{p\in P, j\in J} (1+g^p) \left[ \left( \Ex[\oX^p_{j}]-\Ex[\oY_j] -i(p,j)\epsilon_S \right)^2 + \Var [\oX^p_{j}] + \Var [\oY_j] \right] 
&&\text{(imitation)} \\
&-\epsilon_C \sum_{p\in P,l\in \cL} (\epsilon_T+g_i) 2^{-2d(c_{i,l})} \Prob [W_{p,l}] &&\text{(circuit)}\\
&+\epsilon_P \sum_{p\in P} g^p  \Ex[\Phi^p]    &&\text{(potential)}\\
&+\epsilon_G \sum_{j\in J} \Ex[\oY_{j}] \Ex[\Delta_j]  &&\text{(gradient)}\\
&-\epsilon_T \sum_{p\in P, j\in J,k\in [K]} 
 \left( \Ex[\tX^p_{j,k}]-\Ex[\oY_j]-i(p,j)\epsilon_S \right)^2 + \Var [\tX^p_{j,k}] + \Var [\oY_j]  &&(x\text{-bit imitation)}\\
&-\epsilon_T  \sum_{p\in P, j\in J,k\in [K]} (1+g^p) \left[ 
 \left( \Ex[\oX^p_{j}]-\Ex[\tY_{j,k}]-i(p,j)\epsilon_S \right)^2 + \Var [\oX^p_{j}] + \Var [\tY_{j,k}] \right]  &&(y\text{-bit imitation)}\\
&+\epsilon_T\epsilon_G \sum_{j\in J,k\in [K]} \Ex[\tY_{j,k}] \Ex[\Delta_j]    &&\text{(bit gradient)}\\
&+\epsilon_T \sum_{r\in R} \Ex[M^x_{r}] + \epsilon_T \sum_{j\in J}\Ex[M^y_{j}] &&\text{(veto)}.
\end{aligned}
\end{align}

\subsection{Interior Equilibrium Analysis}\label{sec:int}

We use a standard trick (see e.g., \cite{DGP}): We first reduce $\textsc{2D-GD-FixedPoint}$ to the problem of finding a \emph{well supported} $\epsilon'_N$-Nash equilibrium. Thereafter, we can use the fact that the problem of finding a well-supported $\epsilon'_N$-Nash equilibrium can be reduced to the problems of finding $\epsilon_N$-Nash equilibrium; see \cite{DGP}. For clarity of presentation, in this section and in Section \ref{sec:boundary} we argue about exact equilibria. However, all our arguments remain valid also for the case of well-supported $\epsilon'_N$-Nash equilibrium for $\epsilon'_N\ll \epsilon_T^2$.% and therefore the term \emph{well-supported $\epsilon'_N$-Nash equilibrium} is simply replaced by the short term \emph{equilibrium}.

For clarity of presentation we describe first an equilibrium analysis under the assumption that all imitation targets of the samples $x^p$ for $p\in P$ and are located in the interior of $[0,1]$. Formally, this Section states and proves all lemmata and propositions under the following assumption.
\begin{assumption}\label{as:hy}
For every $r\in R$ we have $0\leq \hy_r<1$; see Equation \eqref{eq:hy-def} for the definition of $\hy_r$.
\end{assumption}
%\begin{assumption}\label{as:xj}
%For every $j\in J$ we have $\hx_j + \frac{\epsilon_G \Ex[\Delta_j] }{2P+2}>0$; see Section \ref{sec:y-team} for the motivation of this precise assumption.
%\end{assumption}
%Assumption \ref{as:xj} is indeed a boundary condition that is similar to stating $\hx_j\geq 0$ because $\epsilon_G$ is a small constant.

Thereafter, in Section \ref{sec:boundary} we show that, with an assistance of several complementary lemmata, very similar analysis can be applied in the general case; i.e., without imposing Assumption \ref{as:hy}. The fact that \emph{the proof} of a certain lemma in this section relies on Assumption \ref{as:hy} is indicated in the statement of the lemma by ``Under Assumption \ref{as:hy}''. For lemmata that indirectly rely on Assumption \ref{as:hy} (i.e., are using previous lemmata that have used Assumption \ref{as:hy}) we omit the dependence on Assumption \ref{as:hy} from the statement of the lemma. The reason for that is because in Section \ref{sec:boundary} we provide complementary lemmata to those that actually rely on Assumption \ref{as:hy}. The lemmata that indirectly reply on Assumption \ref{as:hy} need no adjustments in their proofs. 

\subsubsection{Perfect Imitation for some $x$-teams}

For every team $x_r=x_{i,j}$ in the sampling structure we define the \emph{imitation target} of the $x_r$-team by:
\begin{align}\label{eq:hy-def}
\hy_r:=
\Ex [\oY_j]+i\epsilon_S.
\end{align}
Note that $-\epsilon_S\leq \hy_r\leq 1+\epsilon_S$ where $\epsilon_S$ is the sample shift size. 
Recall that the imitating team $x_r$ always chooses a number $0\leq \ox_r \leq 1-2^{-K}$. It is possible that the imitation target $\hy_r$ will be out of $[0,1-2^{-K}]$. Our Assumption \ref{as:hy} excludes the possibility of $\hy_r\in [1,1+\epsilon_S]$ and the possibility of $\hy_r\in [-\epsilon_S,0)$. This possibility is relegated to Section \ref{sec:boundary}.

By slightly rewriting of the utility in Equation~\eqref{eq:eu} we will see (Equation~\eqref{eq:uij}) that indeed the $x_r$-team tries to imitate the scalar $\hy_r \in [0,1)$. We first define different notions of imitation for the $x_r$-team.

\begin{definition}[Perfectly, strongly perfectly, approximately perfect, and mild imitation] \label{def:perfect-imitation}\hfill

We say that the $x_{r}$-team \emph{perfectly imitates} $\hy_{r}$ if $\Ex[(\oX_{r}-\hy_r)^2]= O(2^{-2K})$.
We say that the $x_{r}$-team \emph{strongly perfectly imitates} $\hy_{r}$ if, in addition, $\Prob[M^x_{r}>N_{in}]=1$ and $\Prob[X_{r,k}=y^b_{r,k}]=1$ for every $k\in [N_{in}]$; i.e., the veto player does not impose veto on the first $N_{in}$ bits and the first $N_{in}$ bit players purely match the binary representation of $\hy_{r}$. We say that the $x_{r}$-team \emph{approximately perfectly imitates} $\hy_{r}$ if $\Ex[(\oX_{r}-\hy_r)^2]\leq O(\epsilon_G)$. Finally, we say that the $x_{r}$-team \emph{mildly imitates} $\hy_{r}$ if $\Ex[(\oX_{r}-\hy_r)^2]\leq O(\oep_C)$.

We say that \emph{(strongly) perfect imitation happens in a sample} $p=(i_1,i_2)\in P$ if both teams $x_{i_j,j}$ for $j=1,2$ (strongly) perfectly imitate their target $\hy_{i_j,j}$. 
\end{definition}

Roughly speaking, our goal is to prove that all $x_r$-teams imitate well their target $\hy_r$. As a first step toward this goal we succeed to prove it under some conditions. Specifically, we show the following.

\begin{enumerate}
\item If $\oox_r$ (see Equation \ref{eq:oox}) is sufficiently far from the $2^{-N_{in}-1}$-grid then strongly perfect imitation occurs (see the $\oox$-far-imitation Lemma (Lemma~\ref{lem:grid-far})).
\item If a team is not chosen by the guide player then the imitation is perfect (see the No-guide-imitation Lemma (Lemma~\ref{lem:no-krentel->perfect})).
\end{enumerate} 

Roughly speaking, we obtain these results inductively starting from the most significant bit in a decreasing order. We first prove that the first bit of $\oX_r$ is pure and matches the first bit of $\hy_r$. Thereafter, we proceed to the second bit, et cetera. We start with a rather weak version of ``pureness of the $k$-th bit'' argument of the Vetoed-mix-significant-bits Lemma (Lemma \ref{lem:no-mix}), which states that a bit that is sufficiently significant and is vetoed with probability 0 cannot be mixed. This allows us to define the useful scalar $\oox_r$ that captures the ``location of the $x_r$ team'' (even if $\oX_r$ is random).
We show in the $\oox$-utility-approximation Lemma (Lemma \ref{lem:oox}) that this scalar indeed captures well the imitation loss of the $x_r$-team. The key lemmata of this section are the Correct-bit Lemmata \ref{lem:imitation-No-circuit} and \ref{lem:imitation-No-guide}. These lemmata argue that if all more significant bits behave well (i.e., match the imitation target) so does the current bit. The inductive application of these Correct-bit Lemmata appear in Imitation Lemmata \ref{lem:grid-far} and \ref{lem:no-krentel->perfect} (correspondingly).

We start our formal discussion by denoting 
\begin{align}\label{eq:gr}
g_r:=\sum_{p\in P(r)} g^p
\end{align}
the total mass the guide player assigns to the $x_r$ team. By Equation \eqref{eq:eu} we can write
\begin{align}\label{eq:uij}
\begin{aligned}
\Ex[U^x_{r}]=&- (3+g_r) \left[ \left( \Ex[\oX_r]- \hy_r \right)^2 + \Var [\oX_r] \right] &&\text{(imitation)}\\
&-\epsilon_C \sum_{p\in P(r),l\in \cL_{in}} (\epsilon_T+g_i) \Prob [W_{p,l}] &&\text{(circuit)}\\
&-\epsilon_T \sum_{k\in [K]} 
 \left( \Ex[\tX_{r,k}]-\hy_{r} \right)^2 + \Var [\tX_{r,k}]  &&\text{($x$-bit imitation)}\\
&-\epsilon_T \Ex [U^y_{\text{bit}}(\oX_{r})] &&\text{($y$-bit imitation)}\\
&+\epsilon_T \Ex[M^x_{r}] &&\text{(veto)},
\end{aligned}
\end{align}
where $U^y_{\text{bit}}(\oX_{r})$ denotes the $y$-bit imitation term and highlights the fact that this utility term depends on the actions of the $x_{r}$-team only through $\oX_{r}$. Note that the potential and the bit-potential terms are not affected by the $x_{r}$-team and hence are excluded from $\Ex[U^x_{r}]$. Also, $\Var[\oY_j]$ is not affected by the $x_{r}$-team and is omitted from $\Ex[U^x_{r}]$ for simplicity. 

All our lemmata refer to the behaviour in (an arbitrary) equilibrium. Our first lemma in the analysis states that sufficiently significant bits that are not vetoed are playing pure.
Formally, we recall that $ \oep_C \gg \sqrt{\epsilon_C}$. Let 
\begin{align}\label{eq:kc}
k_C:=\lfloor-\log \oep_C\rfloor.
\end{align}
Roughly speaking, the influence of every bit $1\leq k \leq k_C$ on $\ox$ is more significant than the circuit utility.  

\begin{lemma}[Vetoed-mix-significant-bits Lemma]\label{lem:no-mix} 

For every $r\in R$, and every $k < \min\{\supp(M^x_{r}) \cup \{k_C\}\}$ the bit player $x_{r,k}$ is playing pure. 
\end{lemma}

\paragraph{Informal proof sketch}
We consider the imitation utility only. If, by way of contradiction there exists a bit player who is not vetoed and is playing mixed, we deduce by the indifference principle that $\hy_{r}$ is located approximately at the center in between the two realizations ($0$ and $1$) of the $k$-th bit and the distance between $\oX_{r}$ and $\hy_{r}$ is approximately $2^{-k-1}$. In the relevant interval that contains $\hy_{r}$, veto player has a grid of options of distance $2^{-k-1}$ from each other (see Equation \eqref{eq:veto-value}). Therefore, one of veto's types will imitate  $\hy_r$ better (i.e., the distance will reduce to $2^{-k-2}$), so veto strictly prefers to impose this veto on the $k$-th bit.

\begin{proof}
For simplicity of notations we omit the index $r$, which remains fixed along the entire proof of this lemma. 
Assume by way of contradiction that there exists a sufficiently significant bit $k<k_C$ that is not vetoed (because $k<\min\{\supp(M^x)\}$) and mixes between $0$ and $1$. 
We denote by $x_{<k}= \Ex [\oX_{<k}]$ the expected contribution to $\Ex[\oX_r]$ of the bits that are more significant than $k$.
We denote by $x_{>k}= \Ex [\oX_{>k}]$ the expected contribution to $\Ex[\oX_r]$ of the bits that are less significant than $k$.

By the indifference principle we have $\Ex[U^x]=\Ex[U^x|X_{k}=0]=\Ex[U^x|X_{k}=1] $. By Equation \eqref{eq:uij} we can write 
\begin{align*}
\Ex[U^x|X_{k}=0] &= - (3+g_r) \left[ ( x_{<k} + x_{>k} - \hy )^2 + \Var [\oX_{-k}] \right] \pm \tO(\epsilon_C)\\
\Ex[U^x|X_{k}=1] &= - (3+g_r) \left[ \left( x_{<k} +2^{-k} + x_{>k}- \hy \right)^2 + \Var [\oX_{-k}] \right] \pm \tO(\epsilon_C).
\end{align*}
from this we deduce that
\begin{align*}
&\left( x_{<k} + x_{>k}- \hy \right)^2 = \left( 2^{-k}+ x_{<k} + x_{>k} - \hy \right)^2 \pm \tO(\epsilon_C) \Rightarrow \\
& 2^{-k} \cdot |2^{-k} + 2x_{<k} + 2x_{>k}-2\hy|=\tO(\epsilon_C) \Rightarrow \\
& \hy=x_{<k} + x_{>k} +2^{-k-1} \pm 2^{k+1}\tO(\epsilon_C) \\
 & \;\;\;\; =x_{<k} + x_{>k}+2^{-k-1} \pm \tO(\oep_C) \Rightarrow \\
&\Ex[U^x]=-(3+g_r)[2^{-2(k+1)}+\Var[\oX_{-k}]]\pm \tO(\epsilon_C).
\end{align*}
 
Since $x_{>k}\in [0,2^{-k})$ we know that $\hy \in [x_{<k}+2^{-k-1},x_{<k}+2^{-k}+ 2^{-k-1}]$ and so one of the numbers $x_{<k} + 2^{-k-1}-2^{-K}, x_{<k} + 2^{-k}, x_{<k} + 2^{-k} + 2^{-k-1}$ is $2^{-k-2}$-close to $\hy$. If the veto player will choose the corresponding veto action ($00\vone$, $1\vzero$, $11\vzero$, see Equation \eqref{eq:veto-value}) on the $k$-th bit, she will improve the common utility to at least $\Ex[U^x]=-(3+g_r)[2^{-2(k+2)}+\Var[\oX_{>k}]]\pm \tO(\oep_C)$, which is a contradiction. This is indeed an improvement because $2^{-2(k+2)}<2^{-2(k+1)}+\tO(\epsilon_C)$ and $\Var[\oX_{>k}]\leq \Var[\oX_{>k}]+\Var [\oX_{<k}] = \Var[\oX_{-k}]$.
\end{proof}

The Vetoed-mix-significant-bits Lemma (Lemma \ref{lem:no-mix}) allows us to assign a real number $\oox_{r}$ for each random variable $\oX_{r}$ that captures well (up to $\oep_C$) the common number represented by  the $x_{r}$-team. 
We denote $m:=\min \{\supp (M^x_{r}) \}$ and we recall that $v(m,t)$ denotes the suffix value of the veto action $(m,t)$ (see Equation \eqref{eq:veto-value}). We define
\begin{align}\label{eq:oox}
\oox_{r}:=\begin{cases}
(\ox_{r})_m^< + v(m,t) &\text{if } m<k_C  \text{ and } t=\arg\max_{t'} \{\Prob[(M^x_{r},T^x_{r})=(m,t')]\} \\
\Ex[\oX_{r}] &\text{if } m\geq k_C. 
\end{cases}
\end{align}
The real number $\oox_{r}$ chooses the (deterministic) number that is obtained when the veto player imposes a veto on the most significant bit in her support. If several veto types are imposed on this most significant bit we choose the most probable one.\footnote{The tie-braking rule in case we have two most probable types is not important.} But, if this bit it too insignificant (i.e., $m\geq k_C$) $\oox_{r}$ is set to be the expectation. Note that in the first case $\ox^<_{r,m}$ is pure by the Vetoed-mix-significant-bits Lemma (Lemma \ref{lem:no-mix}), and hence it is not a random variable. 

\begin{lemma}[$\oox$-utility-approximation Lemma]\label{lem:oox}
For every $r\in R$ we have
\begin{align}\label{eq:uoox}
\begin{aligned}
\Ex[U^x_{r}] & =-(3+g_r) \Big((\Ex[\oX_{r}]-\hy_{r})^2 + \Var[\oX_r]\Big)\\
% \ \ \ \text{ and equivalently, } \\ &
& = -(3+g_r)\Big((\oox_{r}-\hy_{r})^2\pm O(\oep_C^2)\Big) .
\end{aligned}
\end{align} 
\end{lemma}

\begin{proof}
%The two conditions in equation \eqref{eq:uoox} are equivalent simply because the imitation utility captures up to $\tO(\epsilon_C)\ll \oep_C^2$ the utility.
The first line follows from Eq.~\eqref{eq:uij} simply because all terms except the imitation affect the utility by at most $\tO(\epsilon_C)\ll \oep_C^2$.
%We ignore all utility terms other than the imitation one. These terms affect the utility by at most $\tO(\epsilon_C)=O(\oep_C^2)$.
%We use the representation of Equation \eqref{eq:uij} for the imitation utility. 

We now prove the second line.
If $m<k_C$ then the indifference principle for the veto player and the fact that for such a choice $\Var[\oX_{r}]=0$ (i.e., $\oX_{r}$ is deterministic) implies  the lemma by Equation \eqref{eq:uij}. If $m\geq k_C$ then by the Vetoed-mix-significant-bits Lemma (Lemma \ref{lem:no-mix}) $\oX^+_{r,k_C}$ is deterministic
and hence all realizations of $\oX_{r}$ are $2\times 2^{-k_C}=2\oep_C$-close to $\Ex [\oX_{r}]$ and so $\Var[\oX_{r}]\leq 4\oep_C^2$. Therefore, $(3+g_r)(\oox_{r}-\hy_{r})^2=(3+g_r)(\Ex[\oX_{r}]-\hy_{r})^2$ approximates up to $O(\oep_C^2)$ the imitation loss. 
\end{proof}

We denote by $(\hy^b_{r,k})_{k\in [K]}$ the binary representation of $\hy_{r}$.
Simply speaking, the following lemma states that for coordinates $k>N_{in}$ (those that do not serve as an input for the circuit) if all more significant bits have matched the binary representation of $\hy_{r}$ and if the veto player didn't impose veto on these more significant bits, then the $k$-th bit player also matches the binary representation of $\hy-r$.

\begin{lemma}(Correct-bit-no-circuit Lemma)\label{lem:imitation-No-circuit}
Under Assumption \ref{as:hy}, for every $r\in R$, and every $N_{in}<k<K-3$, if the following conditions hold: 
\begin{description}
\item[Bit-$k$ is not vetoed:] $\Prob[M^x_{r}<k]=0$; and 
\item[More significant bits agree with $\hy_r$:] $\Prob[X_{r,k'}=\hy^b_{r,k'}]=1$ for every $k'<k$,
\end{description}
 then we have at least one of the following conclusions:
 \begin{description}
 \item[$\hy_r$ is indifferent about Bit-$k$:] $(\hy_{r})_{\geq k} \in [2^{-k}-2^{-K},2^{-k}]$; or 
 \item[Bit-$k$ also agrees with $\hy_r$:] $\Prob[X_{r,k}=\hy^b_{r,k}]=1$. 
 \end{description}
\end{lemma}
Note that the case where $(\hy_{r})_{\geq k} \in [2^{-k}-2^{-K},2^{-k}]$ corresponds to the case where $\hy$ has two different binary representations. Both representations approximate $\hy_{r}$ up to $2^{-K}$, but the $k$-th bit of these representations differs. So the lemma essentially states that the $x_{r,k}$ matches the representation, unless there are multiple representation with different $k$-th bit.

\paragraph{Informal proof sketch} The idea is to state that if the bit player $x_{r,k}$ plays the wrong bit with positive probability then the veto player will improve the common utility if she vetoes this bit. Once the $k$-th bit player is vetoed with probability 1, the bit player affects the utility only through the bit imitation. Therefore, she prefers to match the correct bit. The formal arguments are quite involved because the lemma should treat the case where both the $x_{r,k}$ player and the veto player are playing mixed strategies.

\begin{proof}
For simplicity of notations we omit the index $r$ in the notations, which remains fixed for the entire proof of this lemma.
The bits $(\oX^b_{k'})_{k'\geq N_{in}}$ do not affect the circuit utility. Moreover, we will consider deviations of the veto player to $m^x=k$ which again does not affect the circuit utility. Therefore, the circuit utility can be ignored.
By the assumption that $x_{k'}=\hy_{k'}$ for $k'<k$ the contribution of these bits to $\Ex[\oX]$ and to $\hy$ is identical. Similarly, for the expression $\Ex[\tX]$. For simplicity of notations we denote $X=\oX_{\geq k} \in [0,2^{-k+1}-2^{-K}]$ and $y=\hy_{\geq k}\in [0,2^{-k+1}]$. Therefore we can rewrite the expected common utility (see Equation \eqref{eq:uij}) as 
\begin{align}\label{eq:uijk}
\begin{aligned}
\Ex[U^x_k]=\text{\footnotemark}&- 
(3+g_r) \left[ \left( \Ex[X]- y \right)^2 + \Var [X] \right] &&\text{(imitation)}\\
&-\epsilon_T \left[ \left( \Ex[\tX_{\geq k}]-y \right)^2 + \Var [\tX_{\geq k}] \right] &&\text{($x_k$-bit imitation)}\\
&-\epsilon_T \Ex [U^x_{-k \text{ bit}}(X)] &&\text{($x_{-k}$-bit imitation)}\\
&-\epsilon_T \Ex [U^y_{\text{bit}}(X)] &&\text{($y$-bit imitation)}\\
&+\epsilon_T \Ex[M^x] &&\text{(veto)}.
\end{aligned}
\end{align}

\footnotetext{Omitting the circuit term because it is not affected by less (or equally) significant bits than $k$.}

To analyze the case where the $k$-th bit is not vetoed, it suffices to simplify to 
\addtocounter{footnote}{-1}
\begin{equation}\label{eq:uijk-simple}
\Ex[U^x_k]=\text{\footnotemark}- 
(3+g_r) \left[ \left( \Ex[X]- y \right)^2 + \Var [X] \right] \pm \tO(\epsilon_T).
\end{equation}

\begin{figure}\label{fig:y-cases}
\caption{Cases of $\hy_r$ for proof of Imitation Lemma~\ref{lem:imitation-No-circuit}}
\begin{tikzpicture}

\def\myx{0}
\def\myy{0}
\def\W{16}
\def\smallx{1.2}
\def\largex{4}
\def\linewidtha{0.4}
\def\linewidthg{0.5}
\def\scalesize{1.1}
\def\scalesizea{1.2}
%\node (z) at (\myx+18,\myy-0.5) {$1({\color{b2}A})=\alpha$, $0({\color{b2}A})=1-\alpha$};

\draw [line width = \linewidtha mm] (\myx ,\myy - 0.5) -- (\myx + \W, \myy -0.5) ;

\draw [dashed,line width = \linewidtha mm] (\myx , \myy+ 0 ) -- (\myx ,\myy - 1); 
\draw [dashed,line width = \linewidtha mm] (\myx + 8-\smallx, \myy+ 0 ) -- (\myx + 8-\smallx,\myy - 1); 
\draw [dashed,line width = \linewidtha mm] (\myx + 8, \myy+ 0 ) -- (\myx + 8,\myy - 1); 
\draw [dashed,line width = \linewidtha mm] (\myx + \W - \largex, \myy + 0 ) -- (\myx + \W - \largex, \myy - 1); 
\draw [dashed,line width = \linewidtha mm] (\myx + \W, \myy + 0 ) -- (\myx + \W, \myy - 1); 

\node[scale = \scalesize](z) at (\myx, \myy-1.3) {$0$};
\node[scale = \scalesize](z) at (\myx+8-\smallx, \myy+0.34) {$2^{-k}-2^{-K}$};
\node[scale = \scalesize](z) at (\myx+8, \myy-1.3) {$2^{-k}$};
\node[scale = \scalesize](z) at (\myx+\W-\largex, \myy-1.3) {$2^{-k}+2^{-k-1}$};
\node[scale = \scalesize](z) at (\myx+\W, \myy-1.3) {$2^{-k+1}$};

\node[scale = \scalesize](z) at (\myx + 4-0.5*\smallx, \myy-0.2) {\color{blue}(a)};
\node[scale = \scalesize](z) at (\myx + 8-0.5*\smallx, \myy-0.2) {\color{blue}(b)};
\node[scale = \scalesize](z) at (\myx + 12-0.5*\largex, \myy-0.2) {\color{blue}(c)};
\node[scale = \scalesize](z) at (\myx + 16-0.5*\largex, \myy-0.2) {\color{blue}(d)};

\end{tikzpicture}

\end{figure}
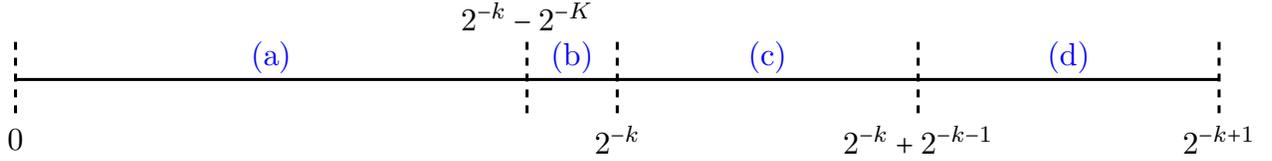

We now consider a few different cases of $\hy_r$, as depicted in Figure~\ref{fig:y-cases}.
Case (a),  where $y<2^{-k}-2^{-K}$, will follow by arguments symmetric to Cases (c) and (d) which are described below.
In Case (b), where $2^{-k}-2^{-K} \leq y \leq 2^{-k}$ we already have the ``$\hy_r$ is indifferent'' conclusion. 
For the rest of the proof, we will consider the Cases (c) and (d), in which $y>2^{-k}$  and $\hy^b_{k}=1$.  We need to prove that $\Prob[X_k=1]=1$. 

We start with considering the six types of the vetos $(k,\vzero),(k,11\vone),(k,0\vone),(k,1\vzero),(k,11\vzero),(k,\vone)$ and find which is optimal (among the six options) in the possible regions of $y$.
Each of the veto types ($\vzero,00\vone,0\vone,1\vzero,11\vzero,\vone$) with $m^x=k$ induces a deterministic value for $X$; i.e.~the variance term is zero.  By Equations \eqref{eq:veto-value} and \eqref{eq:uijk-simple} the corresponding utilities are given by 

\begin{align}\label{eq:m1}
\begin{aligned}
\Ex[U^x_k|M^x=k,T^x=t]=
\begin{cases}
-(3+g_r)(0-y)^2 &\text{if } t=\vzero \\
-(3+g_r)(2^{-k-1}-2^{-K}-y)^2 &\text{if } t=00\vone \\
-(3+g_r)(2^{-k}-2^{-K}-y)^2 &\text{if } t=0\vone \\
-(3+g_r)(2^{-k}-y)^2 &\text{if } t=1\vzero \\
-(3+g_r)(2^{-k}+2^{-k-1}-y)^2 &\text{if } t=11\vzero \\
-(3+g_r)(2^{-k+1}-2^{-K}-y)^2 &\text{if } t=\vone \\
%\text{Three symmetric cases} &\text{if } t\in \{\vzero, 0\vone, 00\vone\}
\end{cases}
\pm \tO(\epsilon_T).
\end{aligned}
\end{align}

We denote $\Ex[U^x_k]=v$. We assume by way of contradiction that $\Prob[X_k=0]>0$ and we consider two cases.

\paragraph{Case (c): $2^{-k}\leq y \leq 2^{-k}+2^{-k-1}$.} By Equation \eqref{eq:m1} the two relevant vetos with $m^x=k$ are either $(k,1\vzero)$ or $(k,11\vzero)$.

By the indifference principle, since $\Prob[X_k=0]>0$, we know that $\Ex[U^x_k|X_k=0]=v$. Since $v$ is also the utility of the veto player in equilibrium, we know that $\Ex[U^x_k|(M^x,T^x)=(k,1\vzero)]\leq v$ and $\Ex[U^x_k|(M^x,T^x)=(k,1\vzero)]\leq v$. We use the shorter notations $\Ex[U^x_k|(M^x,T^x)=(k,1\vzero)]=\Ex[U^x_k|1\vzero]$ and $\Ex[U^x_k|(M^x,T^x)=(k,11\vzero)]=\Ex[U^x_k|11\vzero]$. We show the following inequalities
\begin{align}
&\Ex[U^x_k|X_k=0,M^x>k]<\Ex[U^x_k|1\vzero] \leq v \label{eq:2.1} \\
&\Ex[U^x_k|X_k=0,M^x=k,T^x=1\vzero]<\Ex[U^x_k|1\vzero] \leq v \text{ and }\label{eq:2.2}\\
&\Ex[U^x_k|X_k=0,M^x=k,T^x=11\vzero]<\Ex[U^x_k|1\vzero] \leq v\label{eq:2.3}
\end{align}
which lead to a contradiction because $v=\Ex[U^x|X_k=0]$ is a convex combination the LHS of Equations \eqref{eq:2.1}, \eqref{eq:2.2} and \eqref{eq:2.3} (recall that by lemma premise, $\Prob[M^x_{r}<k]=0$). 

For Inequality~\eqref{eq:2.1} we focus on the imitation term only. For $\Ex[U^x_k|1\vzero]$ the imitation term is
$-(3+g_r)(y-2^{-k})^2$ while the imitation term of $\Ex[U^x_k|X_k=0,M^x>k]$ is at most $-(3+g_r)(y-(2^{-k}-2^{-K}))^2$ because if $X_k=0$ then $X\leq 2^{-k}-2^{-K}$. 
The difference between these two imitation terms is at least $\Omega(2^{-2K})\gg \tO(\epsilon_T)$ and hence $\Ex[U^x_k|1\vzero]$ is higher. 
%In the potential term, again it is easy to see that $\Ex[U^x|1\vzero]$ is better. Finally in all the remaining terms the utility is bounded by $\tO(2^{-35})<< 2^{-2K}$.

To show Inequality~\eqref{eq:2.2} we first observe that it is equivalent to 
\begin{align*}
\Ex[U^x_k|X_k=0,(M^x,T^x)=(k,1\vzero)]<\Ex[U^x_k|X_k=1,(M^x,T^x)=(k,1\vzero)]
\end{align*}
because $\Ex[U^x_k|1\vzero]$ is a convex combination of these two terms. 
Conditional on the $k$-th bit being vetoed, all utility terms of $U^x_k$ are identical for $X_k=0$ and $X_k=1$ except the $x_k$-bit imitation term. For $\Ex[U^x_k|X_k=0,(M^x,T^x)=(k,1\vzero)]$ the bit imitation term is $\epsilon_T y^2$ while for $\Ex[U^x_k|X_k=1,(M^x,T^x)=(k,1\vzero)]$ the bit imitation term is $\epsilon_T (y-2^{-k})^2$. Namely the latter is strictly better.

To show Inequality~\eqref{eq:2.3} we first observe that it is equivalent to 
\begin{align*}
\Ex[U^x_k|X_k=0,(M^x,T^x)=(k,11\vzero)]<\Ex[U^x_k|X_k=1,(M^x,T^x)=(k,11\vzero)]
\end{align*}
because $\Ex[U^x_k|11\vzero]$ is a convex combination of these two terms. 
%Conditional on the $k$-th bit being vetoed, all utility terms of $U^x_k$ are identical for $X_k=0$ and $X_k=1$ except the $x_k$-bit imitation term.
Similarly to the arguments for Inequality~\eqref{eq:2.2}, we shall focus on the $x_k$-bit imitation only. 
For $\Ex[U^x_k|X_k=0,(M^x,T^x)=(k,11\vzero)]$ the bit imitation term is $\epsilon_T (y-2^{-k-1})^2$ while for $\Ex[U^x_k|X_k=1,(M^x,T^x)=(k,1\vzero)]$ the bit imitation term is $\epsilon_T (y-(2^{-k}+2^{-k-1}))^2$. Since $y>2^{-k}$, $y$ is closer to $2^{-k}+2^{-k-1}$ than to $2^{-k-1}$; namely the latter $x_k$-bit imitation is better.

\paragraph{Case (d): $2^{-k}-2^{-k-1} \leq y \leq 2^{-k}+2^{-k-1}$.} By Equation \eqref{eq:m1} the two relevant vetos with $m^x=k$ are either $(k,11\vzero)$ or $(k,\vone)$.

By the indifference principle, since $\Prob[X_k=0]>0$, we know that $\Ex[U^x_k|X_k=0]=v$. Since $v$ is also the utility of the veto player in equilibrium, we know that $\Ex[U^x_k|(M^x,T^x)=(k,11\vzero)]\leq v$ and $\Ex[U^x_k|(M^x,T^x)=(k,\vone)]\leq v$. We use the shorter notations $\Ex[U^x_k|(M^x,T^x)=(k,11\vzero)]=\Ex[U^x_k|11\vzero]$ and $\Ex[U^x_k|(M^x,T^x)=(k,\vone)]=\Ex[U^x_k|11\vone]$. We show the following inequalities
\begin{align}
&\Ex[U^x_k|X_k=0,M^x>k]<\Ex[U^x_k|11\vzero] \leq v \label{eq:3.1} \\
&\Ex[U^x_k|X_k=0,M^x=k,T^x=11\vzero]<\Ex[U^x_k|11\vzero] \leq v \text{ and }\label{eq:3.2}\\
&\Ex[U^x_k|X_k=0,M^x=k,T^x=\vone]<\Ex[U^x_k|1\vone] \leq v\label{eq:3.3}
\end{align}
which lead to a contradiction because $v=\Ex[U^x|X_k=0]$ is a convex combination the tree left-hard sides of Equations \eqref{eq:3.1}, \eqref{eq:3.2} and \eqref{eq:3.3}.

For Inequality~\eqref{eq:3.1} we focus on the imitation term only. For $\Ex[U^x_k|11\vzero]$ the imitation term is
$-(3+g_r)(y-2^{-k}-2^{-k-1})^2$ while the imitation term of $\Ex[U^x_k|X_k=0,M^x>k]$ is at most $-(3+g_r)(y-(2^{-k}-2^{-K}))^2$ because if $X_k=0$ then $X\leq 2^{-k}-2^{-K}$. 
The difference between these two imitation terms is at least $\Omega(2^{-2k})\gg \tO(\epsilon_T)$ and hence $\Ex[U^x_k|11\vzero]$ is higher. 
%In the potential term, again it is easy to see that $\Ex[U^x|1\vzero]$ is better. Finally in all the remaining terms the utility is bounded by $\tO(2^{-35})<< 2^{-2K}$.

Inequality \eqref{eq:3.2} is identical to \eqref{eq:2.3} and the arguments of its proof are valid for Case (d) as well. 

To show Inequality~\eqref{eq:3.3} we first observe that it is equivalent to 
\begin{align*}
\Ex[U^x_k|X_k=0,(M^x,T^x)=(k,\vone)]<\Ex[U^x_k|X_k=1,(M^x,T^x)=(k,\vone)]
\end{align*}
because $\Ex[U^x_k|11\vzero]$ is a convex combination of these two terms. 
%Conditional on the $k$-th bit being vetoed, all utility terms of $U^x_k$ are identical for $X_k=0$ and $X_k=1$ except the $x_k$-bit imitation term.
Again, we shall focus on the $x_k$-bit imitation only. 
For $\Ex[U^x_k|X_k=0,(M^x,T^x)=(k,\vone)]$ the bit imitation term is $\epsilon_T (y-(2^{-k}-2^{-K}))^2$ while for $\Ex[U^x_k|X_k=1,(M^x,T^x)=(k,\vone)]$ the bit imitation term is $\epsilon_T (y-(2^{-k+1}-2^{-K}))^2$. Since $y>2^{-k}+2^{-k+1}-2^{-K}$, $y$ is closer to $2^{-k+1}-2^{-K}$ than to $2^{-k}-2^{-K}$; namely the latter $x_k$-bit imitation is better.  
\end{proof}

\begin{lemma}[$\oox$-far-imitation Lemma]\label{lem:grid-far}
Under Assumption \ref{as:hy}, if $\oox_{r}$ is $\oep_C$-far from the $2^{-N_{in}-1}$-grid then team-$x_{r}$ strongly perfectly imitates $\hy_{r}$.
\end{lemma}

\begin{proof}
For simplicity of notations we omit the indices $r$ in the notations, which remain fixed throughout the proof of this lemma.
First, we know that the $\Prob[M^x\leq N_{in}]=0$, because such a veto results in $\oox$ being $2^{-K}$-close to the $2^{-N_{in}-1}$-grid; see Equation \eqref{eq:oox}. By the Vetoed-mix-significant-bits Lemma (Lemma \ref{lem:no-mix}) all bit players $(X_k)_{k\in [N_{in}]}$ are pure and hence can be denoted as $(x_k)_{k\in [N_{in}]}$. 
We argue that $x_k=\hy^b_k$ for every $k\in [N_{in}]$, namely that these pure bits match the binary representation of $\hy$. The set of actions $\{(k,\vzero),(k,\vone):k\in [N_{in}+1]\}$ are able to move $\oox$ to (both) closest $2^{-N_{in}}$-grid points (or $2^{-K}$-close to it in case the veto is $(k,\vone)$). If by way of contradiction there exists a bit $x_k\neq \hy^b_k$ then the veto player will improve the utility by moving $\oox$ to the closest $2^{-N_{in}}$-grid point in the direction of $\hy$ (i.e., if $\oox<a2^{-N_{in}}\leq \hy$, veto player should impose a veto of $\vone$ on the corresponding bit; if $\oox>a2^{-N_{in}}\geq \hy$, veto should impose a veto of $\vzero$ on the corresponding bit). Such a deviation improves $(\oox-y)^2$ by at least $\oep_C^2$ and hence by the $\oox$-utility-approximation Lemma (Lemma \ref{lem:oox}) strictly improves the total utility.\footnote{Note that the possible loss in the circuit computation is negligible with respect to $\oep_C^2$.}  So far we established: the $N_{in}$ first bits match the binary representation of $\hy$ and the veto player does not impose a veto on these bits. Hence, it suffices to prove that team-$x$ perfectly imitates $\hy$ (and ``strongly'' would follow from previous sentence).

Assume by contradiction that the $x$-team does not perfectly imitate $\hy$. We will use the Correct-bit-no-circuit Lemma (Lemma~\ref{lem:imitation-No-circuit}) to argue by induction over $k \in [N_{in}+1,K-4]$ that $X_{k}=\hy^b_{k}$   w.p.~$1$, and also veto player does not veto the $k$-th bit. But this would imply that team-$x$ {\em does} perfectly imitate $\hy$ - hence a contradiction.

The Correct-bit-no-circuit Lemma (Lemma \ref{lem:imitation-No-circuit}) allows for two conclusion. In the first conclusion, $2^{-k}-2^{-K}\leq (\hy)_{\geq  k} \leq 2^{-k}$ or equivalently $\ox_{<k} + 2^{-k}-2^{-K}\leq \hy \leq \ox_{<k} + 2^{-k} $. Thus veto's action $(N_{in}+1,1\vzero)$ ensures an imitation loss of $O(2^{-2K})$ and hence in equilibrium we must have $(\Ex[\oX]-\hy)^2\leq O(2^{-2K})$  because otherwise the veto player will prefer to impose such a veto (note that such a veto has no effect on the first $N_{in}$ bits and hence does not affect the circuit utility). But then team-$x$ perfectly imitates $\hy$ - a contradition.

Assume by induction that the conditions of the Correct-bit-no-circuit Lemma (Lemma \ref{lem:imitation-No-circuit}) hold for the $k$-th bit. Consider the second conclusion of the Correct-bit-no-circuit Lemma (Lemma \ref{lem:imitation-No-circuit}), i.e.~we know that $X_{k}=\hy^b_{k}$  w.p.~1. We will argue that also the $k$-th bit player is not vetoed ($\Prob[M^x=k]=0$), hence the conditions of the Correct-bit-no-circuit Lemma (Lemma \ref{lem:imitation-No-circuit}) hold also for the $(k+1)$-th bit. To see this, note that if the veto player vetos the $k$-th bit (i.e.~$M^x = k$ --- recall that $M^x \ge k$ by our inductive hypothesis), it's dominant strategy is to pick a veto $\ox^b_{k}=\hy^b_{k}$ that correctly matches the imitation target. The last sentence follows from Equation~\eqref{eq:m1} and the fact that we're in the case $\hy \notin [\ox_{<k} + 2^{-k}-2^{-K}, \ox_{<k} + 2^{-k}]$.
Therefore, veto player's assignment to the $k$-th bit agrees with the bit's player assignment. Hence, due to the $\epsilon_T$ incentive from (veto utility), veto player's action that vetoes the $k$-th bit is dominated by the analogous action that begins the veto with the $k+1$-th bit.
\end{proof}

The following lemma is similar to the Correct-bit-no-circuit Lemma (Lemma \ref{lem:imitation-No-circuit}). It states an identical statement: if all more significant bits have matched the binary representation of $\hy_{i,j}$ and if the veto player didn't impose veto on these more significant bits, then the $k$-th bit player also matches the binary representation of $\hy$. The difference is in the conditions under which we prove it. The Correct-bit-no-circuit Lemma (Lemma \ref{lem:imitation-No-circuit}) stated it for bits $k>N_{in}$. In the following Lemma we state that it is true also for bits $k\leq N_{in}$ in case where the guide player does not choose sample $i$.

\begin{lemma}(Correct-bit-no-guide Lemma)\label{lem:imitation-No-guide}
Under Assumption \ref{as:hy}, for every $r\in R$, and every $k<K-3$ if an equilibrium satisfies $g_r=0$, $\Prob[M^x_{r}<k]=0$ and $\Prob[X_{r,k'}=\hy^b_{r,k'}]=1$ for every $k'<k$ then either $2^{-k}-2^{=K}\leq (\hy_{r})_{\geq k} \leq 2^{-k}$ or $\Prob[X_{r,k}=\hy^b_{r,k}]=1$.
\end{lemma}

The proof is very similar to the proof of the Correct-bit-no-circuit Lemma (Lemma \ref{lem:imitation-No-circuit}) once we observe that the circuit utility has very low weight of $\epsilon_T$ in case $g_r=0$. In the Correct-bit-no-circuit Lemma (Lemma \ref{lem:imitation-No-circuit}) this term simply did not appear because the relevant bit $k$ was not part of the circuit's input.

\begin{proof}
We omit the indices $i,j$ which remain fixed.
As in the Correct-bit-no-circuit Lemma (Lemma \ref{lem:imitation-No-circuit}), we observe that the bits $k'<k$ cancel out in the terms $(\oX-\hy)^2$ and in $(\tX-\hy)^2$.
We denote $X=\oX_{\geq k} \in [0,2^{-k+1}-2^{K}]$ and $y=\hy_{\geq k}\in [0,2^{-k+1}]$. We can rewrite the expected common utility (see Equation \eqref{eq:uij}) similarly to how we did in the Correct-bit-no-circuit Lemma (Lemma \ref{lem:imitation-No-circuit}); see the corresponding Equation \eqref{eq:uijk}.
\begin{align}\label{eq:uijkg}
\begin{aligned}
\Ex[U^x_k]=&-3 \left[ \left( \Ex[X]- y \right)^2 + \Var [X] \right] &&\text{(imitation)}\\
&-\epsilon_C \epsilon_T \sum_{p\in P(r), l\in \cL_{in}}  \Prob [W_{p,l}] &&\text{(circuit)}\\
&-\epsilon_T \left[ \left( \Ex[\tX_{\geq k}]-y \right)^2 + \Var [\tX_{\geq k}] \right] &&\text{($x_k$-bit imitation)}\\
&-\epsilon_T \Ex [U^x_{-k \text{ bit}}(X)] &&\text{($x_{-k}$-bit imitation)}\\
&-\epsilon_T \Ex [U^y_{\text{bit}}(X)] &&\text{($y$-bit imitation)}\\
&+\epsilon_T \Ex[M^x] &&\text{(veto)}.\\
\end{aligned}
\end{align}
The only difference is in the second term, but its weigh is only $\epsilon_C \epsilon_T$. From here we can repeat the arguments of the Correct-bit-no-circuit Lemma (Lemma \ref{lem:imitation-No-circuit}) one by one and verify that they hold for this utility too. The only places in the proof of the Correct-bit-no-circuit Lemma (Lemma \ref{lem:imitation-No-circuit}) that use $\epsilon_T$-low utility terms is in the considerations of the bit player is conditional on being vetoed. But in case she is vetoed she does not affect the circuit utility.
\end{proof}

As a corollary from the Correct-bit-no-guide Lemma (Lemma \ref{lem:imitation-No-guide}) we get the following lemma.% (we recall the definition of perfect imitation in Definition \ref{def:perfect-imitation}).

\begin{lemma}[No-guide-imitation Lemma]\label{lem:no-krentel->perfect}
Under Assumption \ref{as:hy}, for every $r\in R$ such that $g_r=0$ the $x_{r}$-team perfectly imitates $\hy_{r}$.
\end{lemma}

\begin{proof}
The proof uses very similar arguments to those we have applied in (the end of) the $\oox$-far-imitation Lemma (Lemma \ref{lem:grid-far}). The only differences are that we initialize the arguments from the first bit (rather than from the $(N_{in}+1)$-th bit), and the  corresponding Correct-bit lemma that we utilize is the Correct-bit-no-guide Lemma \ref{lem:imitation-No-guide} (rather than the Correct-bit-no-circuit Lemma \ref{lem:imitation-No-circuit}).

The conditions for the Correct-bit-no-guide Lemma (Lemma \ref{lem:imitation-No-guide}) are satisfied for the bit $k=1$. If $\frac{1}{2}-2^{-K}\leq (\hy_{r})_{\geq k} \leq \frac{1}{2}$ then the action $(1,1\vzero)$ ensures a utility of $O(2^{-2K})$, see Equation \eqref{eq:uijkg} and hence in equilibrium we must have $(\Ex[\oX]-\hy)^2\leq 2^{-2K}$ because otherwise the veto player will prefer to impose such a veto. In this case the proof is complete. In the second option of the Correct-bit-no-guide Lemma (Lemma \ref{lem:imitation-No-guide}) we know that $x_1=y^b_1$ w.p.~1. Moreover, one can verify (see Equation~\eqref{eq:m1}) that the veto player applies with positive probability only vetoes that match $y^b_1$ (if any): vetoing the most significant bit is dominated since it would agree with the corresponding bit player's action anyway. Therefore $\Prob[M^x=1]=0$. We proceed with the argument inductively until we reach a bit $k$ for which we have $2^{-k}-2^{-K}\leq (\hy_r)_{\geq k} \leq 2^{-k}$. For such a bit we have that the imitation loss is at most $O(2^{-2K})$. 
\end{proof}

\subsubsection{Mild imitation of all $x$-teams}

Now our goal is to prove that \emph{all} $x_r$-teams mildly imitate their target.
To obtain this goal we consider the interaction from the perspective of the guide player (see Equation \ref{eq:ug}). We show in the Perfect-imitation-sample Lemma (Lemma \ref{lem:perfect-im-exist}) the existence of a prefect imitating sample (see Definition \ref{def:perfect-imitation}). This Lemma 
utilizes the shifted structure of the sampling  
by arguing that the guide player cannot be indifferent between all samples if all of them are located far from their imitation target. Once the guide player has the option to choose a prefect imitating sample we deduce that \emph{all} samples imitate well their target because if this is not the case, the guide player will not choose them and then we can apply the No-guide-imitation Lemma (Lemma \ref{lem:no-krentel->perfect}) to argue that the imitation at this sample is perfect.

We proceed with the formal description of the analysis by considering the interaction from the perspective of the guide player. By Equation \eqref{eq:eu} the utility terms that she affects are given by

\begin{align*}
\begin{aligned}
\Ex[U^g]=&-\sum_{p\in P, j\in J} g^p \left[ \left( \Ex[\oX^p_{j}]-\Ex[\oY_j] -i(p,j)\epsilon_S \right)^2 + \Var [\oX^p_{j}] + \Var [\oY_j] \right] 
&&\text{(imitation)} \\
&-\epsilon_C \sum_{p\in P,l\in \cL} g^p 2^{-2d(c_{i,l})} \Prob [W_{p,l}] &&\text{(circuit)}\\
&+\epsilon_P \sum_{p\in P} g^p  \Ex[\Phi^p]    &&\text{(potential)}\\
&\pm g^p \tO(\epsilon_G),
\end{aligned}
\end{align*}
where the last term $\pm g^p \tO(\epsilon_G)$ includes the gradient, the $y$-bit imitation, and the bit potential utilities. We denote by

\begin{align}\label{eq:ug}
\begin{aligned}
u^{g,p}:=&-\sum_{j\in J} \left[ \left( \Ex[\oX^p_{j}]-\Ex[\oY_j] -i(p,j)\epsilon_S \right)^2 + \Var [\oX^p_{j}] + \Var [\oY_j] \right] 
&&\text{(imitation)} \\
&-\epsilon_C \sum_{l\in \cL} 2^{-2d(c_{i,l})} \Prob [W_{p,l}] &&\text{(circuit)}\\
&+\epsilon_P   \Ex[\Phi^p]    &&\text{(potential)}\\
&\pm \tO(\epsilon_G),
\end{aligned}
\end{align}
the utility of the guide player from choosing the action $g=p$.
Note that the guide player puts a weight of $0$ on a sample whose sum of imitation, circuit, and potential utilities is worse by $\tilde{\omega}(\epsilon_G)$ than some other sample.
 
The following lemma states that there exists a sample where all teams perfectly imitate the corresponding $\hy$.

\begin{lemma}[Perfect-imitation-sample Lemma]\label{lem:perfect-im-exist}
Under Assumption \ref{as:hy}, there exists a sample with perfect imitation (see Definition \ref{def:perfect-imitation}).
\end{lemma}

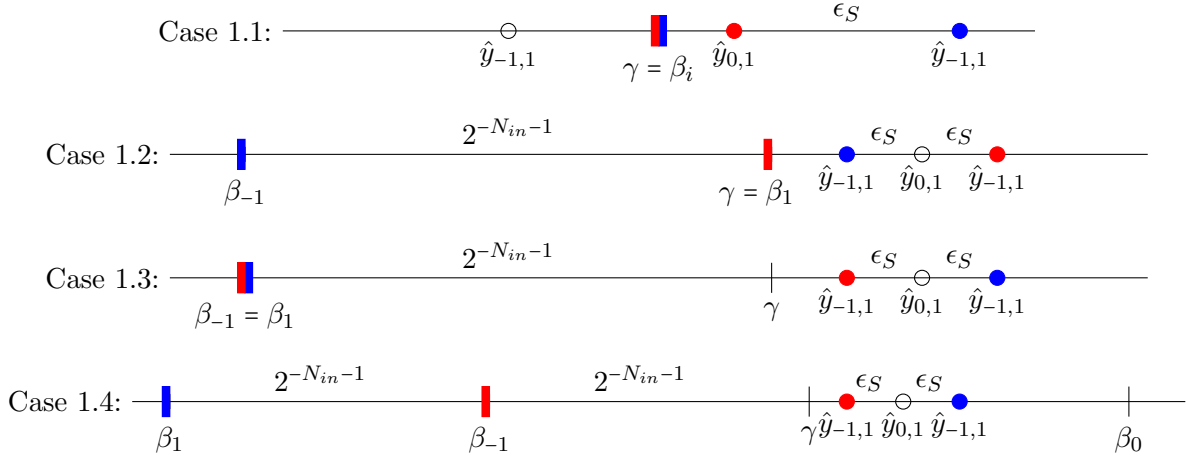
\begin{figure}
\caption{The choice of pairs $(\hy_{i,1},\beta_i)$ in Cases (1.1)-(1.4) in the proof of Lemma \ref{lem:perfect-im-exist}. The pair with larger distance $d_i$ appears in blue. The pair with smaller distance $d_i$ appears in red.}
\begin{center}
\begin{tikzpicture}
\node[left] at (0,0) {Case 1.1:};
\draw (0,0)--(10,0);
\draw (5,-0.1)--(5,0.1);
\node[below] at (5,-0.2) {$\gamma=\beta_i$};
\draw (3,0) circle(0.1);
\filldraw[red] (6,0) circle(0.1);
\filldraw[blue] (9,0) circle(0.1);
\node[below] at (3,0) {$\hy_{-1,1}$};
\node[below] at (6,0) {$\hy_{0,1}$};
\node[below] at (9,0) {$\hy_{-1,1}$};
\node[above] at (7.5,0) {$\epsilon_S$};

\filldraw[blue] (5,-0.2) rectangle (5.1,0.2);
\filldraw[red] (5,-0.2) rectangle (4.9,0.2);

\end{tikzpicture}

\vspace*{2mm}

\begin{tikzpicture}
\node[left] at (-3,0) {Case 1.2:};
\draw (-3,0)--(10,0);
\draw (5,-0.1)--(5,0.1);
\node[below] at (4.8,-0.2) {$\gamma=\beta_1$};
\filldraw[blue] (6,0) circle(0.1);
\draw (7,0) circle(0.1);
\filldraw[red] (8,0) circle(0.1);
\node[below] at (6,0) {$\hy_{-1,1}$};
\node[below] at (7,0) {$\hy_{0,1}$};
\node[below] at (8,0) {$\hy_{-1,1}$};
\node[above] at (7.5,0) {$\epsilon_S$};
\node[above] at (6.5,0) {$\epsilon_S$};
\draw (-2,-0.1)--(-2,0.1);
\node[above] at (1.5,0) {$2^{-N_{in}-1}$};
\node[below] at (-2,-0.2) {$\beta_{-1}$};
\filldraw[red] (5,-0.2) rectangle (4.9,0.2);
\filldraw[blue] (-2,-0.2) rectangle (-2.1,0.2);
\end{tikzpicture}

\vspace*{2mm}

\begin{tikzpicture}
\node[left] at (-3,0) {Case 1.3:};
\draw (-3,0)--(10,0);
\draw (5,-0.2)--(5,0.2);
\node[below] at (5,-0.2) {$\gamma$};
\filldraw[red] (6,0) circle(0.1);
\draw (7,0) circle(0.1);
\filldraw[blue] (8,0) circle(0.1);
\node[below] at (6,0) {$\hy_{-1,1}$};
\node[below] at (7,0) {$\hy_{0,1}$};
\node[below] at (8,0) {$\hy_{-1,1}$};
\node[above] at (7.5,0) {$\epsilon_S$};
\node[above] at (6.5,0) {$\epsilon_S$};
\draw (-2,-0.1)--(-2,0.1);
\node[above] at (1.5,0) {$2^{-N_{in}-1}$};
\node[below] at (-2,-0.2) {$\beta_{-1}=\beta_1$};
\filldraw[blue] (-2,-0.2) rectangle (-1.9,0.2);
\filldraw[red] (-2,-0.2) rectangle (-2.1,0.2);
\end{tikzpicture}

\vspace*{2mm}

\begin{tikzpicture}
\node[left] at (-4,0) {Case 1.4:};
\draw (-4,0)--(10,0);
\draw (5,-0.2)--(5,0.2);
\node[below] at (5,-0.2) {$\gamma$};

\filldraw[red] (5.5,0) circle(0.1);
\draw (6.25,0) circle(0.1);
\filldraw[blue] (7,0) circle(0.1);
\node[above] at (-1.5,0) {$2^{-N_{in}-1}$};
\node[above] at (2.75,0) {$2^{-N_{in}-1}$};

\node[below] at (5.5,0) {$\hy_{-1,1}$};
\node[below] at (6.25,0) {$\hy_{0,1}$};
\node[below] at (7,0) {$\hy_{-1,1}$};

\node[above] at (5.8,0) {$\epsilon_S$};
\node[above] at (6.6,0) {$\epsilon_S$};
\draw (-3.5,-0.1)--(-3.5,0.1);
\node[below] at (-3.5,-0.2) {$\beta_1$};
\filldraw[blue] (-3.5,-0.2) rectangle (-3.6,0.2);

\draw (0.75,-0.1)--(0.75,0.1);
\node[below] at (0.75,-0.2) {$\beta_{-1}$};
\filldraw[red] (0.75,-0.2) rectangle (0.65,0.2);

\draw (9.25,-0.2)--(9.25,0.2);
\node[below] at (9.25,-0.2) {$\beta_0$};
\end{tikzpicture}
\end{center}
\end{figure}
 
\begin{proof}
We prove that for both dimensions $j=1,2$ there exists $i=i(j)$ such that $x_{i,j}$ perfectly imitates $\hy_{i,j}$ and then the sample $(i(1),i(2))$ is a sample with perfect imitation.

Assume by way of contradiction that this is not the case for $j=1$ (the arguments for $j=2$ are identical). Namely, all teams $x_{i,1}$ for $i\in I$ do not perfectly imitate $\hy_{i,1}$. By the $\oox$-far Lemma (Lemma \ref{lem:grid-far}) this implies that $\oox_{i,1}$ is $\oep_C$-close to the $2^{-N_{in}-1}$-grid for every\footnote{This is the only place in the proof where we use Assumption \ref{as:hy}. Indeed, if $\hy_{i,1}\notin [0,1)$ we have not proved yet that non-perfect imitation implies that $\oox_{i,1}$ must be located close to the grid.\label{foot:assumption-use}} $i\in I$.  We denote by $\beta_i$ the $2^{-N_{in}-1}$-grid point to which $\oox_{i,1}$ is close.

Let $d_i:=|\beta_i-\hy_{i,1}|$. We argue that there exist two indexes $\oi,\ui$ such that $d_{\oi}-d_{\ui} \geq \epsilon_S$. To prove it we consider several cases. 

\paragraph{Case 1: $\hy_{0,1}$ is $2\epsilon_S$-close to the $2^{-N_{in}-1}$-grid.} We denote this grid point by $\gamma$ and we consider three sub-cases.

\paragraph{Case 1.1: $\beta_i=\gamma$ for every $i\in I$.} Among the three points $\hy_{-1,1}=\hy_{0,1}-\epsilon_S$, $\hy_{0,1}$, and $\hy_{1,1}=\hy_{0,1}+\epsilon_S$ at least two are located from the same side of $\beta$ (i.e., both belong to $[0,\beta_i]$ or both belong to $[\beta_i,1]$). By choosing $\oi$ to be the farther point and $\ui$ to be closer point we obtain $d_{\oi}-d_{\ui} = \epsilon_S$.

\paragraph{Case 1.2: For some $i,i'$, $\beta_i = \gamma$ and $\beta_{i'} \ne \gamma$.} By choosing $\oi=i$ and $\ui=i'$ we obtain $d_{\ui}-d_{\oi} \geq 2^{-N_{in}-1} - 2 \epsilon_S \geq \epsilon_S$.

\paragraph{Case 1.3: $\beta_i \neq \gamma$ for every $i\in I$, and there exists a pair $i,i'$ with $\beta_i=\beta_{i'}$.}  We consider the two points $\hy_{i,1}$ and $\hy_{i',1}$ and set $\oi$ ($\ui$) to be the index of the point that is farther from (closer to) $\beta_i$. Notice that both $\hy_{i,1}$ and $\hy_{i',1}$ are located $3\epsilon_S$-close to $\gamma$ and hence both are located from the same side of $\beta_i=\beta_{i'}$ (recall that $\gamma$ and $\beta_i$ are $2^{-N_{in}}$-far from each other and $2^{-N_{in}}> 3\epsilon_S$). The shifts are of size $\epsilon_S$, so we get $d_{\oi}-d_{\ui} \geq \epsilon_S$.

\paragraph{Case 1.4: $\beta_i \neq \gamma$ for every $i\in I$, but the $\beta_i$'s are all different.}
Consider the locations of the three points $\beta_{-1},\beta_0,\beta_1$ with respect to the location of $\hy_{0,1}$. Two of them are located from the same side of $\hy_{0,1}$ and are at least $2^{-N_{in}-1}$ apart. Let $\oi$ be the index of the farther $\beta_i$, and $\ui$ let be the index of the closer $\beta_i$. We get that $d_{\oi}-d_{\ui}\geq |\beta_{\oi} - \hy_{0,1}|-\epsilon_S - |\beta_{\ui} - \hy_{0,1}| - \epsilon_S \geq 2^{-N_{in}-1}-2\epsilon_S \geq \epsilon_S$.

\paragraph{Case 2: $\hy_{0,1}$ is $2\epsilon_S$-far the $2^{-N_{in}-1}$-grid.} 
%This case is treated similarly to Case 1.3 
We consider two sub-cases.

\paragraph{Case 2.1: There exists a pair $i,i'$ with $\beta_i=\beta_{i'}$.} This case is treated similarly to Case 1.3.

\paragraph{Case 2.2: All $\beta_i$'s are different.} This case is treated similarly to Case 1.4. \\

Once we know the existence of $\ui,\oi$ with $d_{\oi}-d_{\ui}  \geq \epsilon_S$ we deduce that $|\oox_{\oi,1}-\hy_{\oi,1}|-|\oox_{\ui,1}-\hy_{\ui,1}|\geq \epsilon_S-2\oep_C \geq \frac{1}{2}\epsilon_S$ because $\oox_{i,1}$ is $\oep_C$-close to $\beta_i$. This implies that $(\oox_{\oi,1}-\hy_{\oi,1})^2-(\oox_{\ui,1}-\hy_{\ui,1})^2\geq  \frac{1}{4}\epsilon^2_S \gg \oep^2_C$. By the $\oox$-utility-approximation Lemma (Lemma \ref{lem:oox}) the imitation loss of a sample $(\oi,i_2)$ is $\Omega(\epsilon_S^2)$ greater than the imitation loss of a sample $(\ui,i_2)$ for every $i_2 \in I$. Therefore, the guide player chooses each of the samples $\{(\oi,i_2):i_2 \in I\}$  with probability 0 (because the sample $(\ui,i_2)$ is superior in the $j=1$ coordinate and is identical in the $j=2$ coordinate). Hence $g_{\oi,1}=0$, and by the No-guide-imitation Lemma (Lemma \ref{lem:imitation-No-guide}) we deduce that the team $x_{\oi,1}$ perfectly imitates $\hy_{\oi,1}$ which leads to a contradiction.
\end{proof}
 
The following lemma states that all $x_r$-teams perform reasonable imitation; namely, up to $\oep_C$
\begin{lemma}[Mild-imitation Lemma]\label{lem:mild-imitation}
For every $r\in R$ we have
$(\oox_{r}-\hy_{r})^2\leq \oep_C^2$.
\end{lemma}
The lemma is stated using the real values $(\oox_{r})$, which by the $\oox$-utility-approximation Lemma (Lemma \ref{lem:oox}) captures up to $O(\oep_C^2)$ both the common utility and the imitation term (see Equation \eqref{eq:uoox}).

\begin{proof}
By the Perfect-imitation-sample Lemma (Lemma \ref{lem:perfect-im-exist}) there exists $p^*=(i^*_1,i^*_2)\in P$ such that both teams $x_{i^*_1,1}$ and $x_{i^*_2,2}$ perform perfect imitation.
The total imitation loss (of both teams) is $-\tO(2^{-K})$. Hence the utility of the guide player from choosing the sample $p^*$ (i.e., $u^{g,p^*}$) is at least $-\tO(\epsilon_C)$ (see Equation \eqref{eq:ug}). If, by way of contradiction, there exists an $\oox_{r}$ with $(\oox_{r}-\hy_{r})^2 > \oep_C^2$ then the guide player will choose every sample $p\in P(r)$ with zero probability (because her imitation loss from the $r$-th real number exceeds $\oep_C^2$ while she has a better sample $p^*$). Namely, we get that $g_r=0$ and by the No-guide-imitation Lemma (Lemma \ref{lem:no-krentel->perfect}) this implies that  $x_{r}$ performs perfect imitation which contradicts the fact that $(\oox_{r}-\hy_{r})^2 > \oep_C^2$. 
%because the utility terms other than the imitation are bounded by $\tO(\epsilon_C)\ll \oep_C^2$.
\end{proof}

The Mild-imitation Lemma (Lemma \ref{lem:mild-imitation}) allows us to rephrase the $\oox$-far-imitation Lemma (Lemma \ref{lem:grid-far}).

\begin{corollary}[Far-Perfect Corollary]\label{cor:far->perfect}
If $\hy_{r}$ is $2\oep_C$-far from the $2^{-N_{in}-1}$-grid then the $x_{r}$-team strongly perfectly imitates $\hy_{r}$.
\end{corollary}

The corollary simply follows from the fact that if $\hy_r$ is $2\oep_C$-far from the $2^{-N_{in}-1}$-grid then, by the Mild-imitation Lemma (Lemma \ref{lem:mild-imitation}), $\oox_r$ is $\oep_C$-far from the $2^{-N_{in}-1}$-grid.

So far we have assigned a real value $\oox_r$ to a random variable $\oX_r$ that captures ``the location'' of $\oX_r$. Now it will be convenient to refer to (the seemingly more natural) number $\Ex [\oX_r]$ as the location of $\oX_r$. The following lemma states that these numbers are indeed close to each other.

\begin{lemma}\label{lem:ex-oox}
For every $r\in R$ we have $|\oox_r-\Ex [\oX_r]|\leq 2\oep_C$.
\end{lemma}

\begin{proof}
By the $\oox$-utility-approximation Lemma (Lemma \ref{lem:oox}) $(\oox_r-\hy_r)^2$ captures the imitation loss.
On the other hand, by Fact~\ref{fact:var} the imitation loss can  also be written as $(\Ex[\oX_r]-\hy_r)^2 + \Var[\oX_r]$. By the Mild-imitation Lemma (Lemma \ref{lem:mild-imitation}) we know that $(\oox_r-\hy_r)^2\leq \oep_C^2$ and also $(\Ex[\oX_r]-\hy_r)^2\leq \oep_C^2$ (because the variance is positive) which implies that $|\oox_r-\Ex [\oX_r]|\leq 2\oep_C$. 
\end{proof}

Lemma \ref{lem:ex-oox} allows us to replace in all the analysis so far the term $\oox_r$ by $\Ex[\oX_r]$. In particular, the $\oox$-utility-approximation Lemma (Lemma \ref{lem:oox}) and the Mild-imitation Lemma (Lemma \ref{lem:mild-imitation})
% and Corollary \ref{cor:far->perfect} 
are valid when we replace $\oox_r$ by $\Ex[\oX_r]$, and in particular we get the following corollary:

\begin{corollary}[Mild-imitation Corollary]\label{cor:mild}
For every $r\in R$ the $x_r$-team mildly imitates $\hy_r$.
\end{corollary}

\subsubsection{Perfect Imitation of the $y$-teams}\label{sec:y-team}

In this subsection our goal is to show that the $y_j$-teams perfectly imitate the real number $\hx_j + \frac{\epsilon_G\Ex [\Delta_j]}{20}$ which combines the imitation target $x_j$ and the desired shift in the gradient direction $\triangle_j$. See the Perfect-$y$-imitation Lemma (Lemma \ref{lem:y-team}). To show it we consider the interaction from the perspective of a $y_j$-team and we slightly rewrite it (see Equations \eqref{eq:uy*} and \eqref{eq:uy}). We observe that the Correct-bit-no-guide Lemma (Lemma \ref{lem:imitation-No-guide}) is applicable for the $y_j$-team (without any modifications).

We start our formal discussion with isolating the utility terms of Equation \eqref{eq:eu} that depend on the choice of the $y_j$-team.

\begin{align*}
\begin{aligned}
\Ex[U^y_{j}]=&-\sum_{p\in  P} (1+g^p) \left[ \left( \Ex[\oX^p_j]-\Ex[\oY_j]-i(p,j)\epsilon_S \right)^2 + \Var [\oY_j] \right] 
&&\text{(imitation)} \\
&+\epsilon_G \Ex[\oY_{j}] \Ex[\Delta_j]  &&\text{(gradient)}\\
&-\epsilon_T \sum_{p\in P,k\in [K]}  (1+g^p) \left[\left( \Ex[\oX^p_j]-\Ex[\tY_{j,k}] - i(p,j) \epsilon_S \right)^2 + \Var [\tY_{j,k}] \right]  &&(y\text{-bit imitation)}\\
&-\epsilon_T \Ex [U^x_{\text{bit}}(\oY_j)] &&(x\text{-bit imitation)}\\
&+\epsilon_T\epsilon_G \sum_{k\in [K]} \Ex[\tY_{j,k}] \Ex[\Delta_j] &&\text{(bit gradient)} \\
&+\epsilon_T \Ex[M^y_j] &&\text{(veto)}.
\end{aligned}
\end{align*}

\begin{definition}\label{def:xj}
We define the variable $\hX_j$ that equals $\Ex[\oX^p_{j}]-i(p,j)\epsilon_S$ with probability $\frac{1+g^p}{10}$ for every $p\in P$, and we denote $\hx_j = \Ex[\hX_j]$.
\end{definition}
With this notation we can rewrite
\begin{align}\label{eq:uy*}
\begin{aligned}
\Ex[U^y_j]=\text{\footnotemark}&-10 \left[ \left(\Ex[\oY_j]-\hx_j \right)^2 + \Var [\oY_j] \right] 
&&\text{(imitation)} \\
&+\epsilon_G \Ex[\oY_{j}] \Ex[\Delta_j]  &&\text{(potential)}\\
&-\epsilon_T \sum_{p\in P,k\in [K]}  10 \left[\left( \Ex[\tY_{j,k}] - \hx_j \right)^2 + \Var [\tY_{j,k}] \right]  &&(y\text{-bit imitation)}\\
&-\epsilon_T \Ex [U^x_{\text{bit}}(\oY_j)] &&(x\text{-bit imitation)}\\
&+\epsilon_T\epsilon_G \sum_{k\in [K]} \Ex[\tY_{j,k}] \Ex[\Delta_j] &&\text{(bit gradient)} \\
&+\epsilon_T \Ex[M^y_j] &&\text{(veto)}.
\end{aligned}
\end{align}
\footnotetext{Omitting the term $10\Var [\hX_j]$ which doesn't depend on $y$-players from both: the imitation utility and the $y$-bit imitation utility.} 
By merging the imitation utility with the gradient utility and merging the $y$-bit utility with the bit gradient utility we obtain
\begin{align}\label{eq:uy}
\begin{aligned}
\Ex[U^y_{j}]=\text{\footnotemark}&-10 \left[ \left(\Ex[\oY_j]-\hx_j -\frac{\epsilon_G\Ex[\Delta_j]}{20} \right)^2 + \Var [\oY_j] \right] \\
&-10 \epsilon_T  \left[ \sum_{k\in [K]}  \left( \Ex[\tY_{j,k}] -\hx_j -\frac{\epsilon_G\Ex[\Delta_j]}{20} \right)^2 + \Var [\tY_{j,k}] \right]  \\
&-\epsilon_T \Ex [U^x_{\text{bit}}(\oY_j)] \\
&+\epsilon_T \Ex[M^y_j].
\end{aligned}
\end{align}
\footnotetext{Up to terms that depend on $\hx$ and $\Ex[\Delta_j]$ but not on $\oY$ and $\tY$. Note that those are not affected by the $y_j$-team. %Simple calculations show that relevant terms (those that multiply $\Ex[\hY]$ and those that multiply $\Ex[\tY]$) are identical. 
}
%With an abuse of terminology, we refer to the first term in Equation \eqref{eq:uy} as the imitation utility of team $y_j$ and we refer to the second term in Equation \eqref{eq:uy} as the potential utility of team $y_j$. 
Using the representation of Equation \eqref{eq:uy} for the common utility of the $y_j$-team we deduce the following lemma.

\begin{lemma}[Perfect-$y$-imitation Lemma]\label{lem:y-team}
For every $j\in J$, the $y_{j}$-team perfectly imitates $\hx_j + \frac{\epsilon_G \Ex[\Delta_j] }{20}$.
\end{lemma}

\begin{proof}[Proof Sketch]
The arguments of the proof are very similar to those of the No-guide-imitation Lemma (Lemma \ref{lem:no-krentel->perfect}). 
The first term in Equation \eqref{eq:uy} plays the role of the imitation utility.
The second term in Equation \eqref{eq:uy} plays the role of the bit imitation utility.
The factor $3+g_r$ is replaced by the factor $10$ which does not affect the arguments; the real number $\hy_r$ is replaced by the real number $\hx_j + \frac{\epsilon_G \Ex[\Delta_j] }{20}$ which again does not affect the arguments.

Note also that we argued the validity of the lemma \emph{without} imposing Assumption \ref{as:hy}. This is because the imitation target $\hx_j + \frac{\epsilon_G \Ex[\Delta_j] }{20}$ necessarily belongs to $[0,1)$.
\end{proof}

\subsubsection{Mutual imitation}

In this subsection we combine our observation regarding the imitation of the $x_r$-teams and the $y_j$-teams from which we deduce a stronger bound on the mutual imitation of the groups. The Far-Perfect Corollary (Corollary \ref{cor:far->perfect}) and the Mild-imitation Corollary (Corollary \ref{cor:mild}) allow us to argue that in each dimension there exists at most one $x_r$-team that does not perfectly imitate $\hy_r$. We call such a team \emph{problematic} and we denote its corresponding imitation error by $E_j$; see Definition \ref{def:error}. Our observations about perfect imitation of groups allow us to provide a formula for the imitation error $E_j$. This formula improves the imitation bound: we obtain that all groups approximately perfectly imitate their target (this is stronger than mild imitation from the Mild-imitation Corollary (Corollary \ref{cor:mild})).

Since the shifts of the samples are of size $\epsilon_S$ and $2^{-N_{in}} \gg \epsilon_S \gg \oep_C$ we know that for every $j\in J$ there exists at most one $i\in I$ such that $\hy_{i,j}$ lands $\oep_C$-close to the $2^{-N_{in}-1}$-grid. Therefore, by the Mild-imitation Corollary (Corollary \ref{cor:mild}) and the Far-Perfect Corollary (Corollary \ref{cor:far->perfect}) we deduce that among the $x_{i,j}$ teams for $i\in I$ there exists at most one team that does not perfectly imitate $\hy_{i,j}$. This observation leads us to the following definitions and notations.

\begin{definition}[Problematic teams and imitation error]\label{def:error}\hfill

For every $j\in J$ the unique $i$ such that $\ox_{i,j}$ does not perfectly imitate $\hy_{i,j}$ is called a \emph{problematic} team and is denoted by $i^*_j$. 
If all $\hy_{i,j}$ are $\oep_C$-far from the $2^{-N_{in}-1}$-grid then we say that \emph{a problematic team does not exist}. 
We denote by $E_j:=\Ex[\oX_{i^*_j,j}]-\hy_{i^*_j,j}$ the \emph{imitation error} of the problematic sample.\footnote{Note that we define the imitation error as a \emph{signed} difference. The sign will play a substantial role latter.} If a problematic team does not exist we set $E_j=0$.
A sample $p=(i_1,i_2)$ is called \emph{problematic} if at least one of the teams $x_{i_1,1}$ or $x_{i_2,2}$ is problematic. 
\end{definition}

%Our next argument repeats the arguments of Lemma \ref{lem:perfect->zero-gradient}, but this time it takes into account the fact that a single team in each coordinate might be problematic. 

\begin{lemma}[Error-formula Lemma]\label{lem:error-gradient}
Under Assumption \ref{as:hy}, for every $j\in J$ we have $$E_j=-\frac{\epsilon_G}{6+2g_{r^*(j)}} \Ex[\Delta_j]\pm \tO(2^{-K}).$$
\end{lemma}

\begin{proof}
We fix $j\in J$. 
We note that for every $i\in I$ we have
\begin{align}\label{eq:ii1}
\begin{aligned}
\Ex[\oX_{i,j}]&=\Ex[\oY_{j}]+i\epsilon_S \pm O(2^{-K}) &&\text{ for every } i\neq i^*_1 \\
\Ex[\oX_{i,j}]&=\Ex[\oY_{j}]+i\epsilon_S +E_j &&\text{ for } i=i^*_1.
\end{aligned}
\end{align}
The equations in the first line follow from perfect imitation. The equality in the second line follows from the definition of $E_j$. 
By the fact that the $y_{j}$-team perfectly imitates $\hx_j + \frac{\epsilon_G \Ex[\Delta_j] }{20}$ (Perfect-$y$-imitation Lemma \ref{lem:y-team}) and the definition of $\hx_{j}$ we have
\begin{align}\label{eq:ii2}
\Ex[\oY_{j}]=\sum_{i\in I} \frac{3+g_{i,j}}{10} \left( \Ex[\oX_{i,j}] -i\epsilon_S \right) + \frac{\epsilon_G \Ex[\Delta_j] }{20} \pm O(2^{-K}) 
\end{align}
By plugging the Equations of \eqref{eq:ii1} for every $i\in I$ in \eqref{eq:ii2} the terms $i\epsilon_S$ as well as term $\Ex[\oY_{j}]$ cancel out and we obtain
\begin{align*}
0=\frac{3+g_{r^*(j)}}{10}E_j +\frac{\epsilon_G \Ex[\Delta_j] }{20} \pm O(2^{-K}),
\end{align*}
and the lemma follows.
\end{proof}

The Error-formula Lemma (Lemma \ref{lem:error-gradient}) has two important consequences. First, the Error-formula Lemma significantly improves the imitation bound of the Mild-imitation Corollary (Corollary \ref{cor:mild}). While the Mild-imitation Corollary states that the imitation is bounded by $\tO(\oep_C^2)$, in the Error-formula Lemma \ref{lem:error-gradient} we see that the imitation is bounded by $\tO(\epsilon_G^2)$. In particular, so far the driving force of most arguments we applied was the imitation term. The Error-formula Lemma shows that this term is in fact negligible in every equilibrium since it is bounded by $\tO(\epsilon_G^2)$. Circuit's computations and the potential utility, which have the weights of $\epsilon_C\gg \epsilon_G^2$ and $\epsilon_P\gg \epsilon_G^2$ will be the driving forces of the upcoming arguments. The second important consequence is that $E_j$ and $\Ex[\Delta_j]$ have opposite signs. Namely, whenever $\Ex[\Delta_j]>0$ the team $x_{r^*(j)}$ \emph{under-estimates} the target $\hy_{r^*(j)}$. Whenever $\Ex[\Delta_j]<0$ the team $x_{r^*(j)}$ \emph{over-estimates} the target $\hy_{r^*(j)}$. This consequence will play a significant role in the final argument; see Lemma \ref{lem:final}.

The following corollary follows immediately from the Error-formula Lemma (Lemma \ref{lem:error-gradient}).

\begin{corollary}[Perfect-dimension Corollary]\label{cor:perfect->zero-gradient}
For every $j\in J$ if the $x_{i,j}$-teams perfectly imitate $\hy_{i,j}$ for every $i\in I$ then we have $|\Ex[\Delta_j]|\leq \tO(\frac{2^{-K}}{\epsilon_G})$. 
\end{corollary}

\begin{proof}
Perfect imitation implies $E_j \leq O(2^{-K})$ and hence $|\Ex[\Delta_j]|\leq \tO(\frac{2^{-K}}{\epsilon_G})$.
\end{proof}

In order that the Perfect-dimension Corollary will be applicable for the reduction we need to have two additional properties of an equilibrium which have not been proved yet.

\begin{enumerate}
\item All $x_{r}$-teams should perfectly imitate $\hy_{r}$. So far, we know it only for $x_{r}$-teams such that $\hy_{r}$ is sufficiently far from the $2^{-N_{in}-1}$-grid (see the Far-Perfect Corollary \ref{cor:far->perfect}) and for teams such that $g_r=0$ (the No-guide-imitation Lemma \ref{lem:no-krentel->perfect}).  
\item The term $|\Ex[\Delta_j]|$ should be equal to (or approximate) the actual terms $\triangle_j (z)$ for some point $z$ that is near to the points $\{(x^p),y\}$. So far, we have proved neither that these points are actually close to each other nor that the computations in the circuits are done correctly. 
%Also approximate well the gradient $\nabla_j(x)$. In particular we should prove that $\phi_{j\pm (x)}\approx \Ex [\Phi_{j\pm}]$ for some relevant point $x$;\footnote{The region where the relevant point $x$ is located should follow from the desired item (1) above.} i.e., computations in the circuit are done correctly. 
\end{enumerate}

In the remainder of the proof we show that these two properties are indeed satisfied in every equilibrium.

\subsubsection{Correct circuit computations}

In this subsection we show several results on the correctness of computations inside the circuits (i.e., by the circuit players). We start with several definitions of correct computation. 
We prove the existence of two samples where the computations are correct. The Pure-input Lemma (Lemma \ref{lem:pure-L-input}) allows us to define (yet another) notion that captures the ``location of the team $\ox_r$'': $x^c_r$, see Definition \ref{def:xcr}, which looks at the actions of the circuit's input players rather than on the group $x_r$. Then our goal is to show that samples that are chosen by the guide player have approximately correct computation (see Definition \ref{def:cor-comp}). To do so, we formulate (yet another) Correct-bit Lemma \ref{lem:imitation-matched-input} and apply it inductively in the Strong-perfect-imitation Lemma (Lemma \ref{lem:same-bin}). 

We recall that $W^p_l$ denotes the indicator of wrong computation of the player $l$ in the circuit $C^p$ is wrong. 

\begin{definition}\label{def:cor-comp}
We say that \emph{the computations at the sample $p\in P$ are correct} if $\Pr[W^p_l=1]=0$ for every $l\in \cL$. 
The computations are \emph{weakly correct} if $\sum_{l\in \cL} \Pr[W^p_l=1]\leq \frac{\epsilon_P}{\epsilon_C}$.
\end{definition}

Note that correct computations is a demanding definition for correctness of circuit's computations. In particular, it requires that all the relevant bits $\ox^b_r$ in both dimensions will be pure because otherwise the input players $l^p$ for $l\in \cL_{in}$ have a positive probability of mismatching $\ox^b_r$. Moreover, it requires that all circuit players $l^p$ (including the outputs $\phi^p_k$ and $\triangle^p_{j,k}$) will play a pure and correct action.

We say that a sample $p\in P$ is \emph{perfect} if strongly perfect imitation occurs at the sample $p$ (see Definition \ref{def:perfect-imitation}) and the computations at the sample $p$ are correct (see Definition \ref{def:cor-comp} above).

\begin{lemma}[Perfect-imitation+computation-sample Lemma]\label{lem:perfect-sample}
There exists a perfect sample $p\in P$.
\end{lemma}

\begin{proof}
For every coordinate $j\in J$ there exists at most one problematic team $x_{i,j}$ (by the Far-Perfect Corollary~\ref{cor:far->perfect}) and hence two non-problematic teams; i.e., teams $x_r$ where the real number $\hy_r$ is $\oep_C$-far from the $2^{-N_{in}-1}$-grid. In such a non-problematic team by the Far-Perfect Corollary (Corollary \ref{cor:far->perfect}), the imitation in all these teams is strongly perfect, and in particular all bits $\ox^b_{r,k}$ for $k\in [N_{in}]$ are pure.
We let $p\in P$ be a point with both coordinates be non-problematic (we have at least four such points). We show that  correct computation happens in the sample $p$.

Each circuit player $l^p$ interacts with her predecessors and with at most two successors.
We proceed inductively on the depth of the circuit to argue that $\Prob[W^p_{l}=1]=0$. 
Consider a player $l^p$ that is located at depth $d=d_l$.
We argue that by choosing  the action that performs correct computations of the predecessors, player  $l^p$ ensures a utility of $\geq -2\epsilon_C (\epsilon_T+g^p) 2^{-2(d+1)}$, see Equation \eqref{eq:eu}. 
If this player is not an output player, her maximal loss is twice the penalty for wrong computation of her successors who are located in depth of at least $d+1$.
For an output player her affect on $\Ex[\Phi]$ is bounded by $g^p$  and hence her utility from playing the correct computation bit is at least $-\epsilon_P g^p \gg -2\epsilon_C (\epsilon_T+g^p) 2^{-2(d+1)}$. On the other hand, if player $l^p$ chooses the wrong bit her utility will be at most $-\epsilon_C (\epsilon_T+g^p) 2^{-2d}+O(g^p \epsilon_P)$. Since $2^{-2d} > 2^{-2d+1}+O(\frac{\epsilon_P}{\epsilon_C})$ the result follows. 
\end{proof}

The following lemma states that the input players of every circuit $C^p$ that is chosen by the guide player are pure. Moreover, for circuits $C^{p'}$ of other samples that share a coordinate with $p$, the input bits corresponding to this coordinate agree.
\begin{lemma}[Pure-input Lemma]\label{lem:pure-L-input}
Let $g^p=g^{i_1,i_2}>0$. Then, all circuit-input players $l^p$ for $l\in \cL_{in}$ are playing pure. Moreover, for every $p' = (i'_1,i_2)$ and every $l\in \cL_{in}$ that corresponds to an input from $x_{i_2,2}$ the circuit-input player $l^{p'}$ is playing pure and $l^{p'}=l^p$  (i.e., their pure actions match). 
Similarly, for every $p' = (i_1,i'_2)$ and every $l\in \cL_{in}$ that corresponds to an input from $x_{i_1,1}$ the circuit-input player $l^{p'}$ is playing pure and $l^{p'}=l^p$.
\end{lemma}

\begin{proof}
We first argue that the circuit loss of the $i$-th sample cannot exceed $\epsilon_P$:
There exists of a perfect sample (see the Perfect-imitation+computation-sample Lemma \ref{lem:perfect-sample}) with $O(2^{-2K})$-imitation loss and no loss at all in the circuit computation; combined with the fact  that $g^p>0$ we have that if circuit loss of the $i$-th sample was larger, the guide player would strictly prefer the perfect sample.
 
The bound on the total circuit loss implies weakly correct computations, because each computation mistake is multiplied by a factor of $\epsilon_C$. 
In particular, the computation loss of each input line $\Prob [W^p_{l}=1] $ for $l\in \cL_{in}$ should not exceed $\frac{\epsilon_P}{\epsilon_C}$. Namely, the corresponding bit $\oX_{r,k}$ should match the action $L^p$ with probability $1-\frac{\epsilon_P}{\epsilon_C}$. This can happen only if $\oX_{r,k}$ is a $Bernoulli(q)$ random variable with $q\in [0,\frac{\epsilon_P}{\epsilon_C}]\cup [1-\frac{\epsilon_P}{\epsilon_C},1]$. In case $q\in [0,\frac{\epsilon_P}{\epsilon_C}]$ the unique best-reply of player $l^p$ is $0$. In case $q\in [1-\frac{\epsilon_P}{\epsilon_C},1]$ the unique best-reply of player $l^p$ is $1$.

To see the moreover part, we 
use the fact that $\oX_{r,k}$ is almost pure (i.e., a $Bernoulli(q)$ random variable). Therefore, similarly to the argument for $l^p$, the unique best reply of the input-circuit player $l^{p'}$ is to also match the closest pure realization of $\oX_{r,k}$.
\end{proof}

The Pure-input Lemma (Lemma \ref{lem:pure-L-input}) provides us another useful notion of ``location'' of the team $x_r$ (so far we had $\oox_r$ and $\Ex[\oX_r]$ to capture it). 
\begin{definition}\label{def:xcr}
For $r\in R$ such that there exists $p\in P(r)$ with $g^p>0$ we define $x^c_r \in [0,1]_{2^{-N_{in}}}$ to be the (pure) real number of precision $2^{-N_{in}}$ which is induced by the input players $(l^p)_{l\in \cL_{in}}$ that correspond to the input of $x_r$.
For $p$ such that $g^p>0$ we denote by $x^{c,p}$ the (pure) point of precision $2^{-N_{in}}$ which is induced by the input players $(l^p)_{l\in \cL_{in}}$
\end{definition}
 Note that by the moreover part of the Pure-input Lemma (Lemma \ref{lem:pure-L-input}) this definition is well defined (i.e., does not depend on the choice of $p\in P(r)$). Although $x^c_r$ is defined through the actions of the \emph{circuit} players (rather than by actions of the $x_r$-team players) this notion captures well the ``location'' of $x_r$ since the input bits $\oX^b_{r,k}$ for $k\in [N_{in}]$ match the bits of $x^c_r$ with high probability.

We present, yet another, analogue of a Correct-bit lemma.
This time, the lemma refers to bits $k\leq N_{in}$ that \emph{serve} as an input for the circuits (unlike the Correct-bit-no-circuit Lemma \ref{lem:imitation-No-circuit}), the lemma is valid for every choice of the guide player (unlike the Correct-bit-no-guide Lemma \ref{lem:imitation-No-guide}), but it assumes that all corresponding circuit input players are playing purely some string $x^c_r$ and that the imitation target $\hy_r$ is located very close to the segment $[x^c_r, x^c_r+2^{-N_{in}})$ which corresponds to the set of real numbers whose binary representation is $x^c_r$.

\begin{lemma}[Correct-bit-matched-input Lemma]\label{lem:imitation-matched-input}
Let $r\in R$, $1\leq k\leq N_{in}$. Suppose that all of the following hold:
\begin{description}
%\item[Circuit input players agree:] $\Prob[L^p=x^c_{r,k}]=1$ for every\footnote{The formulation of the Lemma indirectly assumes that $g_r>0$ because only for such teams the notion of $x^c_r$ is ensured to be well-defined.} $p\in P(r)$, where $l^p\in \cL_{in}$ is the input wire corresponding to the $k$-th bit of $\ox_r$; and
\item[$\hy_r$ agrees:] $\hy_r \in \Big(x^c_r - \tO (\sqrt{\epsilon_P}), x^c_r +2^{-N_{in}}+ \tO (\sqrt{\epsilon_P})\Big)$; and
\item[No veto:] $\Prob[M^x_{r}<k]=0$; and
\item[More significant bits agree:] $\Prob[X_{r,k'}=x^c_{r,k}]=1$ for every $k'<k$.
\end{description}
Then also $\Prob[X_{r,k}=x^c_{r,k}]=1$.
\end{lemma}

\begin{proof}
Intuitively, the fact that the corresponding input player in the circuit plays the imitating bit, only increases the incentives of the bit player to match it. Formally, one may verify that all the arguments in the proof of the Correct-bit-no-circuit Lemma (Lemma \ref{lem:imitation-No-circuit}) are valid for this case; instead of arguing that the circuit utility is irrelevant (as in the Correct-bit-no-circuit Lemma \ref{lem:imitation-No-circuit}) we argue that it is relevant but has better utility when $x_{r,k}=x^c_{r,k}$. Note that the possible imitation loss of $\tO(\epsilon_P)$ is negligible with respect to the gain from imitating correctly the circuit inputs $\epsilon_C$.
\end{proof}

The following lemma states that whenever $\hy_r$ has the same binary representation (up to the $N_{in}$-th bit) as $x^c_r$, the imitation of $x_r$ is perfect and there are no computation errors in the $x_r$ inputs of the circuit.

\begin{lemma}[$\hy$-imitation Lemma]\label{lem:same-bin}
For every $r\in R$, if\footnote{This lemma indirectly assumes that $g^p>0$ for some $p\in P(r)$ because the notion $x^c_r$ is defined only for such indices $r\in R$.} $\hy_r \in \Big(x^c_r - \tO (\sqrt{\epsilon_P}), x^c_r +2^{-N_{in}}+ \tO (\sqrt{\epsilon_P})\Big)$  then $\Ex[\oX_r] \in [x^c_r , x^c_r +2^{-N_{in}})$. Moreover, if $\hy_r \in [x^c_r , x^c_r + 2^{-N_{in}}]$ then the $x_r$ team perfectly imitates $\hy_r$. 
\end{lemma}

\begin{proof}
At a high level, we apply the Correct-bit-matched-input Lemma (Lemma \ref{lem:imitation-matched-input}) inductively on the bits $k=1,2,...,N_{in}$ to get that $(\ox^b_r,k)_{k\in [N_{in}]}$ matches $x^c_r$ bit-by-bit. This ensures that $\Ex[\oX_r] \in [x^c_r , x^c_r +2^{-N_{in}})$ irrespective of the bits $N_{in}+1,...,K$. To get the moreover part we apply the Correct-bit-no-circuit Lemma (Lemma \ref{lem:imitation-No-circuit}) inductively on the bits $k=N_{in}+1,...,K-3$. 

\paragraph{Base case:}
For $k=1$ the ``No veto'' and ``More significant bits agree'' conditions of the Correct-bit-matched Lemma (Lemma \ref{lem:imitation-matched-input}) are satisfied trivially. Therefore $X^b_{r,1}=x^c_{r,1}$ with probability 1. 

\paragraph{Inductive step ($k \le N_{in}$):} Now we apply the same argument as we did in the proof of the $\oox$-far-imitation Lemma (Lemma \ref{lem:grid-far}). If the player vetoes the bit $k=1$, she does it by a veto that matches the bit $x^c_{r,1}$ in the first bit of the veto (see Equation \eqref{eq:veto-value}). Since the bit player $x_{r,1}$ also matches $x^c_{r,1}$ the veto player will be better of by switching to the corresponding veto on bit $k=2$; i.e., the following deviations are profitable $(1,\vzero) \to (2,\vzero)$, $(1,00\vone) \to (2,0\vone)$, $(1,0\vone) \to (2,\vone)$, $(1,1\vzero) \to (2,\vzero)$, $(1,11\vzero) \to (2,1\vzero)$, and $(1,\vone) \to (2,\vone)$. Such a change does not affect any terms of the utility except for the veto utility which is improved by $\epsilon_T$. Now the conditions for bit $k=2$ of the Correct-bit-matched Lemma (Lemma \ref{lem:imitation-matched-input}) are satisfied and we may proceed to the next bit. We proceed so inductively till bit $N_{in}$. 

\paragraph{Inductive step ($N_{in} < k \le K-3$):} To prove the moreover part we proceed similarly by an inductive application of the Correct-bit-no-circuit Lemma (Lemma \ref{lem:imitation-No-circuit}) combined with no-veto argument as above.
\end{proof}

\subsubsection{Completing the proof}
Let $p=(i_1,i_2)$ be a sample with $g^p>0$. We argue that $|\triangle_j(x^{c,p})|\leq \epsilon_R$ for both $j=1,2$ (see Definition \ref{def:xcr} for $x^{c,p}$). This is sufficient for the reduction by Lemma \ref{lem:sufficient}.
\begin{lemma}[Final Lemma]\label{lem:final}
If $|\triangle_j(x^{c,p})|\geq \epsilon_R$ then all teams $x_r$ for $r\in R$ perfectly imitate their target. 
\end{lemma}

\begin{proof}
Assume by way of contradiction that $x_r$ does not perfectly imitate $\hy_r$. Assume without loss of generality that $r=(i,1)$. By the no-guide-imitation Lemma (Lemma \ref{lem:imitation-No-guide}) we know that $g_r>0$ and in particular for $p=(i,i_2)$ we have $g^{(i,i_2)}>0$ for some $i_2\in I$. We show that the guide player will improve her utility by choosing a different sample $p'=(i',i_2)$ for $i'\neq i$ instead of $p$. We analyse each of the three significant terms in the guide player utility (see Equation \eqref{eq:ug}) and show that it is better for the sample $p'$ (with a strict inequality in at least one term).

\paragraph{Imitation utility.} The imitation utility in dimension $j=2$ remains unchanged. For the dimension $j=1$ note that $\hy_{i,1}$ is $\oep_C$-close to the $2^{-N_{in}-1}$-grid by the Far-Perfect Corollary (Corollary \ref{cor:far->perfect}). Therefore, $\hy_{i',1}$ is $\Omega(\epsilon_S)$-far from the $2^{-N_{in}-1}$-grid (because $\hy_{i',1}$ is $\epsilon_S$-far from $\hy_{i,1}$). Thus, again by the Far-Perfect Corollary, we deduce that the imitation of the team $x_{i',1}$ is strongly perfect. This strictly improves upon the non-perfect imitation of the $x_{i,1}$-team.

\paragraph{Circuit utility.} For every $l\in \cL_{in}$ that corresponds to an input of the dimension $j=2$ the term $W^p_l$ remains unchanged because by the Pure-input Lemma (Lemma \ref{lem:pure-L-input}) the circuit's input players are playing pure and identically in the circuits of the samples $p$ and $p'$ circuits. Whereas, the mixed strategy of the $x_{i_2,2}$-teams remains unchanged.

For every $l\in \cL_{in}$ that corresponds to an input of the dimension $j=1$ we have $\mathbb{P}[W^{p'}_l]=0$. Hence these terms of the circuit utility are (weakly) superior at $p'$.

Finally, in both circuits, $C^p$ and $C^{p'}$, since the input bits are pure, by an inductive argument on the depth of the wires we deduce that $\mathbb{P}[W^{p}_l]=\mathbb{P}[W^{p'}_l]=0$ for every $l^p,l^{p'}\notin \cL_{in}$ and we have an indifference.

\paragraph{Potential utility.}
Note that $\phi^p=\phi(x^{c,p})$ and $\phi^{p'}=\phi(x^{c,p'})$; i.e., the actual potentials are computed at $p$ and $p'$ because the computations in both circuits ($C^p$ and $C^{p'}$) are done correctly.
We present the analysis for the case of 
$\triangle_1(x^{c,p})>\epsilon_R$. The arguments for $\triangle_1(x^{c,p})<-\epsilon_R$ are symmetric. We denote $r=(i,1)$ and $r'=(i',1)$. We consider three cases: 

\paragraph{Case 1: $x^c_r<x^c_{r'}$.} 
Since the potential is computed correctly in both points $x^{c,p}$ and $x^{c,p'}$ (see Pure-Input-Lemma \ref{lem:pure-L-input}) by definition of $\triangle$ we have $\phi(x^{c,p'})= \phi(x^{c,p})+ \triangle_1(x^{c,p}) >\phi(x^{c,p'})$.

\paragraph{Case 2: $x^c_r=x^c_{r'}$.} The potential remains unchanged. Note that here the gradient term remains unchanged too.

\paragraph{Case 3: $x^c_r>x^c_{r'}$.} 
We show that  this case is impossible by showing that every one of the following two sub-cases is impossible.

\paragraph{Case 3.1: $\Ex [\oX_r]\geq x^c_r + 2^{-N_{in}-1}$.} The distance $|\Ex [\oX_{r}] - \Ex [\oX_{r'}]|$ is bounded by $O(\epsilon_S + \oep_C)$ because both groups $x_{r}$ and $x_{r'}$ imitate up to $2\oep_C$ their targets $\hy_r$ and $\hy_{r'}$ (see the Mild-imitation Corollary (Corollary \ref{cor:mild})) and the distance of $\hy_r$ from $\hy_{r'}$ is at most $2\epsilon_S$ by the definition of the shifts. Therefore, $\Ex [\oX_{r'}]\geq x^c_r + 2^{-N_{in}-1} - O(\epsilon_S + \oep_C)$.
So we deduce that $x^c_{r}\leq x^c_{r'}$ (because the first $N_{in}$ bits of the $x_{r'}$ team are at least $x^c_r$). This contradicts the assumption of Case 3.

\paragraph{Case 3.2: $\Ex [\oX_r] \leq  x^c_r + 2^{-N_{in}-1}$.} We make two sub-claims
\begin{enumerate}
\item $x^c_r \leq \Ex [\oX_r]$.
\item $\Ex [\oX_r] \leq \hy_r$.
\end{enumerate}
By combining these two sub-claims we deduce that $x^c_r \leq  \hy_r$ and in particular 
$\hy_r \in [x^c_r,x^c_r +2^{-N_{in}})$.
But in this case, by the $\hy$-imitation Lemma (Lemma \ref{lem:same-bin}),  $x_r$-team strongly perfectly imitates $\hy_r$. 
It remains to prove the two sub-claims.

\paragraph{Sub-claim 1: $x^c_r \leq \Ex [\oX_r]$} Since $g^p>0$ we know that the bit of $\ox_{r,k}$ matches $x^c_{r,k}$ for every $k\in [N_{in}]$  with probability of at least $\frac{\epsilon_P}{\oep_C}$ (otherwise, a perfect sample whose existence has been proved in the Perfect-imitation+computation (Lemma \ref{lem:perfect-sample}) will be superior for the guide player). Therefore, $\Ex(\oX_r)\geq x^c_r - \tO(\frac{\epsilon_P}{\oep_C})$. Also, $|\Ex(\oX_r)-\hy_r|\leq O(\sqrt{\epsilon_P})$ again because otherwise the perfect sample would be superior. Therefore, we have $\hy_r \geq x^c_r - O(\sqrt{\epsilon_P})$. 
By the $\hy$-imitation Lemma (Lemma \ref{lem:same-bin}) we get that $x^c_r \leq \Ex [\oX_r]$.

\paragraph{Sub-claim 2: $\Ex [\oX_r] \leq \hy_r$} Recall that we have assumed that $\triangle(x^{c,p})>\epsilon_R$. By the $\lambda$-Lipschitness of $\triangle_j$ and the fact that $\triangle_j$ is computed correctly (see the Pure-Input Lemma \ref{lem:pure-L-input}) we deduce that 
\begin{gather}\label{eq:sub-claim-2}
\Ex[\Delta_1] = \triangle(x^{c,p})\pm O(\lambda) \ge \epsilon_R - O(\lambda) .
\end{gather}
%By the we deduce that 
Therefore,
\begin{align*}
\Ex[\oX_r]-\hy_r = E_1 & = \frac{-\epsilon_G}{6+g_r}\Ex[\Delta_1] \pm \tO(2^{-K}) && \text{(Error-formula Lemma (Lemma \ref{lem:error-gradient}))}\\
&\le \frac{-\epsilon_G}{6+g_r}\Big( \epsilon_R - O(\lambda) \Big) \pm \tO(2^{-K}) && \text{(Eq.~\eqref{eq:sub-claim-2})}\\
&\leq 0.
\end{align*}
\end{proof}

The Final Lemma \ref{lem:final} combined with the Perfect-dimension Corollary (Corollary \ref{cor:perfect->zero-gradient}) complete the proof because if $|\triangle_j(x^{c,p})|\leq \epsilon_R$ we are done. And if $|\triangle_j(x^{c,p})|\geq \epsilon_R$ we have perfect imitation for all $x_{i,j}$ groups for $i\in I$ and by the Perfect-dimension Corollary we have 
$|\triangle_j(x^{c,p})|\leq \tO(\frac{2^{-K}}{\epsilon_G})\ll \epsilon_R$.

\subsection{Equilibrium Analysis: Including Boundaries}\label{sec:boundary}
In this section we sketch the equilibrium analysis without imposing Assumptions \ref{as:hy}.
%We rephrase and sketch the proofs of all lemmata and Propositions in Section \ref{sec:int} that reply on the interior assumptions Assumptions \ref{as:y-far} and \ref{as:x-far}. 

We start this section with two lemmata that demonstrate the behavior of the $x_r$ teams in cases where the corresponding imitation target $\hy_r$ is located from the wrong side of the boundary. 

\begin{lemma}\label{lem:hy-boundary}
If $1-2^{-N_{in}-1} \leq \hy_r < 1$ then the $x_r$ team strongly perfectly imitates $\hy_r$. If $\hy_r\geq 1$ then the $x_r$-team strongly perfectly imitates $1-2^{-K}$.

Similarly, if $0 \leq \hy_r < 2^{-N_{in}-1}$ then the $x_r$ team strongly perfectly imitates $\hy_r$. If $\hy_r<0$ then the $x_r$-team strongly perfectly imitates $0$.
\end{lemma} 

\begin{proof}
We prove the first part of the lemma. The second part is proved symmetrically.

Since $1-2^{-N_{in}-1} \leq \hy_r$ for every $k$-th bit player $x_{r,k}$ for $k\in [N_{in}]$  it is easy to see that their action $x_{r,k}=1$ is strictly better than $x_{r,k}=0$ \emph{for every} probability of them being vetoed. This is because the imitation loss and the bit imitation losses after playing $x_{r,k}=0$ is at least $\Omega(2^{-2N_{in}})\gg \epsilon_C$. From here we can deduce that the veto player has no incentive to veto any of them; i.e., $\Prob [M^x_r>N_{in}]=1$.

For the case of $\hy_r < 1$ we can inductively repeat Imitation the Correct-bit-no-circuit Lemma (Lemma \ref{lem:imitation-No-circuit}) to deduce that perfect imitation occurs. For the case of $\hy_r \geq 1$ it remains dominant strategy for the bit players to play $x_{r,k}=1$ (irrespective of their probability for being vetoed).
\end{proof}

%Recall that $I$ is the index of the largest cross. 
\begin{lemma}\label{lem:hy-one-violation}
For every $j\in J$ and every $i\in \{-1,0\}$ we have $\hy_{i,j} < 1$.  Moreover, if $\hy_{(1,j)}\geq 1$ then both other teams $x_{i,j}$ for $i\in \{-1,0\}$ perfectly imitate $\hy_{i,j}$.

Similarly, for every $j\in J$ and every $i\in \{0,1\}$ we have $\hy_{i,j} >0$.  Moreover, if $\hy_{(-1,j)}\leq 0$ then both other teams $x_{i,j}$ for $i\in \{0,1\}$ perfectly imitate $\hy_{i,j}$.
\end{lemma}

\begin{proof}
We prove the first part of the lemma. The second part is proved symmetrically.

Note that it is sufficient to prove that $\hy_{0,j}<1$ because $\hy_{0,j}>\hy_{-1,j}$. Assume by way of contradiction that $\hy_{0,j}\geq 1$. This implies (by definition of the sampling shifts) that $\hy_{1,j}\geq 1+\epsilon_S$.

We denote $E_{i,j}=\hy_{i,j}-\Ex[\oX_{i,j}]$ the imitation error of the $x_{i,j}$-team. By Lemma~\ref{lem:hy-boundary} we know that if $\hy_{i,j}\leq 1$ then we have perfect imitation. So the only possible case not to have perfect imitation is if $\hy_{i,j}>1$, and in such a case $E_{i,j}=\hy_{i,j}-\mathbb{E}[\oX_{i,j}]\geq \hy_{i,j}-1\geq 0$. Note also that $E_{1,j}\geq E_{0,j}+\epsilon_S \geq \epsilon_S$.

The term $\Ex[Y_j]-\hx_j$ (see Definition \ref{def:xj}) can be equivalently viewed as the expectation of $E_{i,j}$ when we draw $i$ from $\{-1,0,1\}$ with corresponding probabilities $\frac{3+g_{i,j}}{10}$. So by the above observation we deduce that $\Ex[Y_j]-\hx_j \geq \frac{3}{10}\epsilon_S$ (because it is the expectation of a positive random variable that gets a value of $\epsilon_S$ with probability of at least $\frac{3}{10}$). This contradicts\footnote{The Perfect-$y$-imitation Lemma (Lemma \ref{lem:y-team}) is applicable here because its proof does not rely on Assumption \ref{as:hy}.} the Perfect-$y$-imitation Lemma (Lemma \ref{lem:y-team}) which implies that $(\Ex[Y_j]-\hx_j\pm \tO(\epsilon_G))^2 \leq O(2^{-2K})$ and, in particular, $\Ex[Y_j]$ is $\tO(\epsilon_G)$-close to $\hx_j$.

The moreover part is a straightforward consequence of Lemma \ref{lem:hy-boundary} because if $\hy_{1,j}\geq 1$ and $\hy_{0,j}<1$ we can bound $1-2^{-N_{in}-1}<y_{i,j}<1$ for $i\in \{0,1\}$.
\end{proof}

With Lemmata \ref{lem:hy-boundary}, and \ref{lem:hy-one-violation} the modifications that are needed to the equilibrium analysis of Section \ref{sec:int} are minor. We briefly overview the equilibrium analysis and mention which (cosmetic) changes are required. 

The Vetoed-mix-significant-bits Lemma (Lemma \ref{lem:no-mix}) and the $\oox$-utility-approximation Lemma (\ref{lem:oox}) do not rely on Assumption \ref{as:hy}.
The analogue of the Correct-bit-no-circuit Lemma (Lemma \ref{lem:imitation-No-circuit}), the $\oox$-far-imitation Lemma (Lemma \ref{lem:grid-far}), the Correct-bit-no-guide Lemma (Lemma \ref{lem:imitation-No-guide}), and the No-guide-imitation Lemma (Lemma \ref{lem:no-krentel->perfect})  for the boundary case is Lemma \ref{lem:hy-boundary}. In the boundary case the situation is simpler and positive results on perfect imitation can be deduced independently of guide player's mixed strategy. The only place we rely on Assumption \ref{as:hy} in the proof of the Perfect-imitation-sample Lemma (Lemma \ref{lem:perfect-im-exist})  is when we state that every non-perfectly imitating $x_r$ satisfies $\oox_r$ is located close to the grid; see Footnote \ref{foot:assumption-use}. This property remains true in the boundary case (Lemma \ref{lem:hy-boundary}).
The Mild-imitation Lemma (Lemma \ref{lem:mild-imitation}), \ref{lem:ex-oox}, and the Perfect-$y$-imitation Lemma (Lemma \ref{lem:y-team}) do not rely on Assumption \ref{as:hy}. 

Given $j\in J$ we extend the definitions of a problematic team $r^*(j)$ and of an imitation error $E_j$ in Definition~\ref{def:error} to the boundary cases. In case $\hy_{i,j}>1$ for some $i\in I$ we know by Lemma \ref{lem:hy-one-violation} that the only team that might have an imitation error  is $r=(1,j)$. This observation leads us to the following definition.

\begin{definition}[Problematic teams and imitation error (boundary version)]\label{def:error-boundary}
In case $\hy_{i,j}>1$ we define the {\em problematic team} to be $r^*=(1,j)$. Correspondingly, the {\em imitation error} is defined by $E_j=(1-2^{-K})-\hy_{1,j}$.

In case $\hy_{i,j}<0$ we define the {\em problematic team} to be $r^*=(-1,j)$. Correspondingly, the {\em imitation error} is defined by $E_j=0-\hy_{-1,j}$.

In case $\hy_r \in [0,1]$ we stick to the original definitions of problematic team and imitation error. See Definition \ref{def:error}.
\end{definition}

We proceed with the overview of the equilibrium analysis for remaining lemmata. 
The property that is used in the proof of the Error-formula Lemma (Lemma \ref{lem:error-gradient}) is the fact that in each coordinate we have a single problematic team. By Lemma \ref{lem:hy-one-violation} and definition \ref{def:error-boundary} this is true also for the boundary case. The Perfect-imitation+computation-sample Lemma (Lemma \ref{lem:perfect-sample}), the Pure-input Lemma (Lemma \ref{lem:pure-L-input}), 
and the $\hy$-imitation Lemma (Lemma \ref{lem:same-bin}), 
do not rely on Assumption \ref{as:hy}, while Lemma \ref{lem:hy-boundary} is complementary to the Correct-bit-matched-input Lemma (Lemma \ref{lem:imitation-matched-input}). 
The Final Lemma (Lemma \ref{lem:final}) does not rely on Assumption \ref{as:hy}.
This completes the reduction of $\textsc{2D-GD-FixedPoint}$ to $\textsc{5-Pt-Potential}$.

\section{Reducing Polytensor Games via Domain's Modification}\label{sec:domain}

To complete the proof of Theorem \ref{theo:main} we need to show that {\sc 5-Polytensor IdenticalInterest} can be polynomially reduced to $\textsc{Deg-5-GD-FixedPoint}$. 

\begin{proposition}\label{pro:inKKT}
$c-${\sc Polytensor-IdenticalInterest} can be polynomially reduced to\\ $\textsc{Deg-5-GD-FixedPoint}$.  
\end{proposition}

It have been observed in \cite{CLS} that there is a one-to-one correspondence of  local maxima of the potential function defined on the \emph{product of simplices} (i.e., the mixed action profiles in the game) and mixed Nash equilibria in the congestion games. However, since we have restricted the domain of $\textsc{GD-FixedPoint}$ to be the \emph{hypercube} (rather than product of simplices), this observation is insufficient, and mild additional arguments are required. The proof below provides these additional arguments.

\begin{proof}

We recall that $n$ denotes the number of players.
We denote by $A_i$ the number actions  of player $i\in [n]$. We denote $A=\sum_{i\in [n]} A_i$.
The mixed strategies of player $i$ are denoted by $x_i\in \Delta(A_i)\subset [0,1]^{A_i}$ where $x_{i,k}$ denotes the probability of player $i$ for choosing her $k$-th action.
The potential of the game is a multilinear function $U:\times_{i\in [n]} \Delta(A_i) \rightarrow [-1,1]$. The same multilinear function $U$ can be viewed as a function with larger domain $U:[0,2]^{A_1}\times ... \times [0,2]^{A_n} = [0,2]^{A} \rightarrow [-2^A,2^A]$.
For every point $x=((x_{1,k})_{k\in A_1},...,(x_{n,k})_{k\in A_n})\in [0,2]^A$ we denote $$s_i:=\sum_{k\in A_i} x_{i,k}.$$

As an input to $\textsc{Deg-5-GD-FixedPoint}$ we consider the function $\phi:[0,2]^A \rightarrow \mathbb{R}$ that is given by
\begin{align}\label{eq:cong-game-pot}
\phi(x):=-C \sum_{i\in [n]} ( s_i -1 )^2 + U(x)
\end{align}
where $C=4\cdot 2^A$ and we set $\epsilon=\frac{1}{3}\epsilon_N 2^{-n}$.
%and $\kappa=\epsilon^2$ and 
We argue that every 
$x=((x_{1,k})_{k\in A_1},...,(x_{n,k})_{k\in A_n})$ that is a solution for $\textsc{GD-FixedPoint}(\phi,\epsilon)$ 
%(with the parameters $\epsilon$ and $\kappa$) 
corresponds to an $\epsilon_N$-Nash equilibrium of the congestion game after normalization of the strategies. Namely, we argue that the profile of actions $((\frac{1}{s_1} x_{1,k})_{k\in A_1},...,(\frac{1}{s_n} x_{n,k})_{k\in A_n})$ is an $\epsilon_N$-Nash equilibrium.

The gradient in the direction $(i,k)$ is given by

\begin{align*}
\nabla_{i,k} \phi(x):=-2C ( s_i -1 ) + U(e_{i,k},x_{-i})
\end{align*}

If $s_i<0.5$ then $x_{i,1}<0.5$ and $\nabla_{i,1} \phi(x)>C+U(e_{i,1},x_{-i})\geq 2^A$ and it cannot be a {\sc GD-FixedPoint} solution.

If $s_i>1.5$ then there exists $x_{i,k}>0$ and $\nabla_{i,k} \phi(x) < -C+U(e_{i,k},x_{-i})\leq -2^A$ and it cannot be a {\sc GD-FixedPoint} solution.

Therefore, every solution must satisfy $0.5\leq s_i \leq 1.5$, and assume by way of contradiction that $((\frac{1}{s_1} x_{1,k})_{k\in A_1},...,(\frac{1}{s_n} x_{n,k})_{k\in A_n})$ is not an $\epsilon_N$-Nash equilibrium. Some player $i$, without loss of generality $i=1$, has a profitable deviation to a pure action $k$ which improves her utility by at least $\epsilon_N$. Without loss of generality $k=1$. Namely, $$U(e_{1,1},(\frac{1}{s_2} x_{2,k})_{k\in A_2},...,(\frac{1}{s_n} x_{n,k})_{k\in A_n})\geq U((\frac{1}{s_1} x_{1,k})_{k\in A_1},(\frac{1}{s_2} x_{2,k})_{k\in A_2},...,(\frac{1}{s_n} x_{n,k})_{k\in A_n})+\epsilon_N.$$ 

Consider the gradient in the direction $(-x_{1,1}+s_1,-x_{1,2},...-x_{1,k},0,0,...,0)$ which corresponds to the direction in which player 1 increases the weight on her first action and proportionally decreases the weights from the point $x_1$ such that the term $s_1$ will remains fixed. The gradient in this direction is given by 

\begin{align}\label{eq:grad-dir}
\begin{aligned}
&\nabla_{(-x_{1,1}+s_1,-x_{1,2},...-x_{1,k},0,0,...,0)} \phi (x) = \\
&s_1 [-2C ( s_1 -1 ) + U(e_{1,1},x_{-1})] +\sum_{k} x_{1,k} [2C ( s_1 -1 ) - U(e_{1,k},x_{-1})]= \\
&s_1 U(e_{1,1},x_{-1}) - U(x) = \\
& \Pi_{i\in [n]} {s_i} \left[ U(e_{1,1},\frac{x_2}{s_2},...,\frac{x_n}{s_n})- U(\frac{x_1}{s_1},\frac{x_2}{s_2},...,\frac{x_n}{s_n}) \right]  > \Pi_{i\in [n]} {s_i} \epsilon_N \geq \epsilon_N 2^{-n}.
\end{aligned}
\end{align}

%\newtext{
If, by way of contradiction, we have $|\nabla_{1,k} \phi(x)|\leq \epsilon$ for every $k$, then we have
\begin{align*}
|\nabla_{(-x_{1,1}+s_1,-x_{1,2},...-x_{1,k},0,0,...,0)} \phi (x)|\leq (s_1-x_{1,1})\nabla_{1,1} \phi(x) +\sum_{k\geq 2}x_{1,k}|\nabla_{1,k} \phi(x)|\leq 2s_1 \epsilon \leq \epsilon_N 2^{-n}
\end{align*}
which contradicts Inequality \eqref{eq:grad-dir}. 
%}

Finally, we observe that the potential in Equation \eqref{eq:cong-game-pot} is a valid potential for $\textsc{Deg-5-GD-FixedPoint}$. First, since $U$ is a $5$-polytensor game, the utility $U(x)$ is a multilinear degree 5 polynomial. Note also that the monomials of the term $-C\sum_i (s_i-1)^2$ are $-Cx_{i,k} x_{i,k'}$ which are either multilinear (for $k\neq k'$) or strictly concave (for $k=k'$). Namely, all the monomials of the potential are of maximal degree 5 and also are component-wise concave.   
\end{proof}

%\Yakov{Reviewer 1: "does it make sense to include some open problems/future directions?" You did not want to include such a Section and said that it should not appear at the end of the paper. If you changed your mind, we have a draft for such a section in tha past versions.}
 
\bibliographystyle{alpha}
\bibliography{ref}

\end{document}

%% file: t1.tex
% !TeX root = GD-FixedPoint.tex

\begin{longtable}{rX@{}@{}lX@{}}
\toprule
\textbf{Notation} & \textbf{Description} & \textbf{Reference} \\
\midrule
\endhead
$\alpha$ & potential's gradient Lipschitz constant & Def.~\ref{def:gdfp}\\
$D$ & depth of the modified circuit & Sec. \ref{sec:mod-cir}\\
$d_l$ & depth of gadget $l$ in the circuit & Sec. \ref{sec:mod-cir} \\
$\triangle_j(\cdot)$ & potential's difference in adjacent $2^{-N}$-grid points ($j$'th dimension) & Eq. \eqref{eq:triangle} \\
$\triangle^p_j$ & potential's difference in adjacent points computed by $p$'th circuit & Sec. \ref{sec:rep} \\
$E_j$ & imitation error in dimension $j$ & Def. \ref{def:error} \\
$\epsilon$ & gradient's approximation in {\sc GD-FixedPoint} & Def.~\ref{def:gdfp} \\
$\epsilon_C$ & weight of circuit utility & Eq. \eqref{eq:constants} \\
$\oep_C$ & $\oep_C \gg \sqrt{\epsilon_C}$  & Eq. \eqref{eq:constants} \\
$\epsilon_G$ & weight of gradient utility & Eq. \eqref{eq:constants} \\
$\epsilon_N$ & equilibrium approximation in $\textsc{Congestion}$ & Sec. \ref{sec:cong} \\
$\epsilon_P$ & weight of potential utility & Eq. \eqref{eq:constants} \\
$\epsilon_R$ & desired bound on $\triangle_j$ & Eq. \eqref{eq:N and ep-r}  \\
$\epsilon_S$ & sampling shift size & Sec. \ref{sec:constants} \\
$\epsilon_T$ & tiny constant & Eq. \eqref{eq:constants} \\
$\phi(\cdot)$ & the actual potential in {\sc GD-FixedPoint} & Def.~\ref{def:gdfp} \\
$\phi^p$ & potential computed by $p$'th circuit & Sec. \ref{sec:rep}\\
$g$ & guide player's action & Sec. \ref{sec:players-actions} \\
$g^p$ & guides player's probability of choosing point $p$ & Sec \ref{sec:rewriting} \\
$g_r$ & guides player's prob. of choosing a point where $r$ participates & Eq. \eqref{eq:gr} \\
$I=\{-1,0,1\}$ & set of sampling-shifts indices & Sec. \ref{sec:sampling}\\
$i$ & generic index of a sampling shift &  \\
$i^*_j$ & unique problematic team in dimension $j$ & Def. \ref{def:error} \\
$J=\{1,2\}$ & set of dimensions of the potential function & Def.~\ref{def:gdfp}\\
$j$ & generic index of a dimension &  \\
$K$ & number of bits in each real number & Eq.~\eqref{eq:constants}\\
$k$ & generic index of a bit & \\
$k_C=\lfloor-\log \oep_C\rfloor$ & index of least significant bit  that is more significant than circuit utility & Eq. \eqref{eq:kc} \\
$\cL$ & set of lines (wires) in circuit $C$ & Sec. \ref{sec:mod-cir} \\
$\cL_{in}$ & set of input lines & Sec. \ref{sec:mod-cir} \\
$\cL_{out}$ & set of output lines & Sec. \ref{sec:mod-cir} \\
$l$ & generic index for a gadget &  \\
$l^p$ & the $l$'th gadget in the $p$'th circuit & \ref{sec:mod-cir} \\
$\lambda$ & Lipschitz constant of $\triangle_j$ & Eq. \ref{eq:lambda} \\
%$\olam$ & $\olam \gg \lambda$ & Eq. \eqref{eq:constants} \\
$m^x_r$ & bit from which the $r$'th $x$-team veto player imposes a veto & Sec. \ref{sec:real} \\ 
$m^y_j$ & bit from which the $j$'th $y$-team veto player imposes a veto & Sec. \ref{sec:real} \\ 
$N_{in}$ & number of bits in circuit's input & Sec. \ref{sec:constants}\\
$N_{out}$ & number of bits in circuit's internal lines and output & Sec. \ref{sec:constants}\\
$P = I^2$ & set of point indexes & Sec. \ref{sec:sampling} \\
%$P_i$ & set of point indexes in $i$'th cross & Sec. \ref{sec:sampling} \\
$P(r)$ & point indexes in in which the real number $x_r$ participates & Sec. \ref{sec:sampling} \\
$p$ & generic index of a sample & \\
$R = I \times J$ & set of real number indexes that generate the points & Sec. \ref{sec:sampling} \\
$r$ & generic index of real number &  \\
$t^x_r$ & veto type of the $r$'th $x$-team veto player & Sec. \ref{sec:real} \\ 
$t^y_j$ & veto type of the $j$'th $y$-team veto player & Sec. \ref{sec:real} \\
$u$ & common utility in the game & Eq. \eqref{eq:u} \\
$u^x_r$ & utility from perspective of $x_r$-team of players & Eq. \eqref{eq:uij} \\
$u^{g,p}$ & utility of guide player from choosing $p$ & Eq. \eqref{eq:ug} \\
$u^y_j$ & utility from perspective of $y_j$-team of players & Eq. \eqref{eq:uy}\\   
$w^p_l$ & indicator of wrong computation of $l$'th line in $p$'th circuit & Sec. \ref{sec:rep}\\   
$\ox^p$ & the $p$'th point (veto's action aggregated) & Sec \ref{sec:real} and \ref{sec:sampling} \\
$\ox^p_j,\ox_r,\ox_{i,j}$ & single coordinate of a point (veto's action aggregated) & Sec \ref{sec:real} and \ref{sec:sampling} \\
$\tx^p_{j,k}$ & $\ox^p_j$ with $k$'th bit determined by bit player & Eq. \eqref{eq:tx} \\
$x_{<k},x_{\leq k}$ & value of binary string up to $k$'th bit (excluding, including) & Eq. \eqref{eq:x<} \\
$x_{>k},x_{\geq k}$ & value of tail string starting at $k$'th bit (excluding, including) & Eq. \eqref{eq:x<} \\
$(x_k)_{k\in K}$ & binary string determined by bit players choice & Sec. \ref{sec:real} \\
$x_{r,k}$ & action of $k$'th bit player in $x_r$-team & Sec. \ref{sec:real}, \ref{sec:sampling} \\
$(\ox^b_k)_{k\in K}$ & binary string that incorporates veto and bit players choices & Sec. \ref{sec:real} \\
$\oox_r$ & real number that captures location of $x_r$-team & Eq. \eqref{eq:oox} \\
$x^c_r$ & real number that is induced by $r$ circuit input players & Def. \ref{def:xcr} \\
$x^{c,p}$ & point that is induced by circuit input players & Def. \ref{def:xcr} \\
$\oy_j$ & the $j$'th coordinate of the point $y$ (veto's action aggregated) & Sec \ref{sec:real}, \ref{sec:sampling} \\
$\ty_{j,k}$ & $\oy_j$ with $k$'th bit determined by bit player & Eq. \eqref{eq:tx} \\
$\hy_{r}$ & imitation target of $x_r$-team & Eq. \eqref{eq:hy-def} \\
$\hy^b_{r}$ & binary string representation of $\hy_r$ & Before Lem. \ref{lem:imitation-No-circuit} \\
\bottomrule
\end{longtable}   